\documentclass[a4paper]{article}
\usepackage{fullpage}
\usepackage{amsmath, amsthm, amssymb, mathtools}
\usepackage{natbib}
\usepackage[dvipsnames]{xcolor}
\colorlet{UrlBlue}{MidnightBlue!70!black}
\colorlet{darkgreen}{green!60!black}
\usepackage[
  unicode=true,
  colorlinks=true,
  linkcolor=MidnightBlue,
  citecolor=OliveGreen,
  urlcolor=UrlBlue,
  filecolor=Magenta,
  bookmarksopen=true,
  bookmarksnumbered=true
]{hyperref}
\usepackage[capitalize]{cleveref}
\usepackage{enumerate}
\usepackage{booktabs}

\usepackage{algorithm}
\usepackage{algpseudocodex}

\usepackage{nicefrac}
\usepackage{subcaption}
\usepackage{bm}
\usepackage{optidef}
\usepackage{thmtools}
\usepackage{placeins}
\usepackage{setspace}
\usepackage{titlesec}
\titlespacing*{\subsection}
  {0pt}   
  {1.2ex plus .2ex minus .1ex}  
  {0.5ex plus .1ex}              

\titlespacing*{\section}
  {0pt}
  {2.5ex plus .3ex minus .2ex}
  {1.8ex plus .2ex}

\usepackage{newpxtext,newpxmath}

\usepackage{tikz}
\usetikzlibrary{shapes,patterns,calc,spy}
\usepackage{pgfplots}
\pgfplotsset{compat=1.17}
\usetikzlibrary{arrows,automata,math,backgrounds,decorations.pathreplacing,decorations.markings,patterns,arrows.meta,positioning}
\tikzset{agent/.style={draw, circle, inner sep=3, outer sep=3}}
\tikzset{trade/.style={draw, thick, -Latex}}
\usepackage{tikz-network}

\definecolor{burgundy}{HTML}{882426}
\definecolor{midnight}{HTML}{151B54}
\definecolor{bluegray}{rgb}{0.4, 0.6, 0.8}

\pgfplotsset{every tick label/.append style={font=\footnotesize}}

\pgfplotsset{tick style={ultra thick, tickwidth=0.25cm}}

\theoremstyle{plain}
\newtheorem{theorem}{Theorem}
\newtheorem{proposition}{Proposition}
\newtheorem{lemma}{Lemma}
\newtheorem{corollary}{Corollary}
\newtheorem{observation}{Observation}

\newtheorem{fact}{Fact}
\theoremstyle{definition}

\newtheorem{definition}{Definition}

\newtheorem{example}{Example}
\AtBeginEnvironment{example}{%
  \pushQED{\qed}%
}
\AtEndEnvironment{example}{\popQED}

\newenvironment{proof2}[1][Proof]{\par\vspace{\topsep}\noindent\textbf{#1.} }{\ \hfill$\square$\par\vspace{\topsep}}

\renewcommand{\vec}[1]{\bm{#1}}
\newcommand{\R}{\mathbb{R}}
\newcommand{\Rn}{\mathbb{R}^n}
\newcommand{\Z}{\mathbb{Z}}

\DeclareMathOperator*{\argmax}{\textnormal{arg\,max}}

\DeclareMathOperator*{\asc}{\textnormal{asc}}
\DeclareMathOperator*{\desc}{\textnormal{desc}}
\newcommand{\eb}{\vec{e}}

\newcommand{\bb}{\vec{b}}
\newcommand{\cb}{\vec{c}}
\newcommand{\pb}{\vec{p}}
\newcommand{\qb}{\vec{q}}

\newcommand{\ub}{\vec{u}}
\newcommand{\xb}{\vec{x}}
\newcommand{\nb}{\vec{n}}
\newcommand{\yb}{\vec{y}}
\newcommand{\zb}{\vec{z}}

\newcommand{\vb}{\vec{v}}
\newcommand{\offer}{\sigma}
\newcommand{\offers}{\vec{\offer}}
\newcommand{\market}{\mathcal{M}}
\newcommand{\auction}{\mathcal{A}}
\newcommand{\LIP}{\mathcal{L}}
\newcommand{\bbE}{\mathbb{E}}

\DeclarePairedDelimiter{\ceil}\lceil\rceil
\DeclarePairedDelimiter{\floor}\lfloor\rfloor

\makeatletter
\newcommand\blfootnote[1]{%
  \begingroup
    \renewcommand\thefootnote{}%
    \footnotetext{#1}%
    \addtocounter{footnote}{0}%
  \endgroup
}
\makeatother

\title{\textsc{Decentralized Trading Networks: Equilibria and Fairness}}

\author{
Simon Finster\footnote{Johannes Kepler University Linz and Austrian Institute of Economic Research, \href{mailto:simon.finster@jku.at}{simon.finster@jku.at}}
\and
Paul W.~Goldberg\footnote{University of Oxford, \href{mailto:paul.goldberg@cs.ox.ac.uk}{paul.goldberg@cs.ox.ac.uk}}
\and
Edwin Lock\footnote{King's College London, \href{mailto:edwin.lock@kcl.ac.uk}{edwin.lock@kcl.ac.uk}}
\and
Matilde Tullii\footnote{Fairplay Team, CREST, Ensae - IP Paris, \href{mailto:matilde.tullii@ensae.fr}{matilde.tullii@ensae.fr}}
}

\date{\vspace{0.6cm} 23 February 2026}
\begin{document}

\maketitle

\begin{abstract}
We explore stability and fairness considerations in decentralized networked markets with bilateral contracts, building on the trading networks framework \citep{hatfield2013stability}. In our trading network game, we show that a well-defined subset of Nash equilibria can be supported as competitive equilibria. Considering an offer-based trading dynamic as well as a stochastic price clock market, we prove new convergence results to Nash equilibrium and competitive equilibrium, providing a rationale for stability properties in decentralized, dynamic trading networks. Turning to the tension between fairness and (core) stability, we prove several negative results: inessential agents always receive zero utility in any core outcome, and even essential agents can get zero utility in all core outcomes.\\

\end{abstract}
\blfootnote{\emph{Acknowledgments}: We are grateful to Patrick Loiseau, Bruno Strulovici, Vijay Vazirani, and Alex Westkamp for helpful discussions. PWG was supported by UKRI/EPSRC grants EP/X040461/1 and EP/X038548/1.}

\section{Introduction}

Bilateral trading in supply chains, OTC financial markets, and labor markets is pervasive, yet we lack strategic and dynamic foundations for stability of market outcomes in these networked settings --- and we understand even less about who benefits from trade. In this paper, we prove new convergence results for decentralized trading and show that core stability imposes stark limits on fairness: essential agents, who are part of every efficient allocation, can be left with nothing, and inessential agents must receive zero surplus in every core outcome. We establish these results in the trading network market of \citet{hatfield2013stability}, a general framework in which agents act as buyers, sellers, and intermediary traders, engaging in bilateral contracts with connected counterparts. \citet{hatfield2013stability} proved fundamental results connecting competitive equilibria and stability. Dynamic stability and redistributive concerns, however, have remained largely unexplored.

We propose the trading network game, which builds on the trading network market. In our game, agents make offers to their trading counterparts. As many Nash equilibria are inefficient, we introduce a refinement, $\varepsilon$-tightness, that selects a pertinent subset of Nash equilibria and show that these can be connected to competitive equilibria under fully substitutable preferences \citep{hatfield2013stability,Hatfield-2019}. We demonstrate that both Nash equilibria and (approximate) competitive equilibria can emerge from decentralized trading. Moreover, we provide the first and fundamental insights into surplus sharing in (core) stable outcomes between agents in the network market.

Stability is the central requirement for equilibria in markets, whether competitive or strategic. The literature on trading networks to date has focused on static stability \citep{hatfield2013stability,ostrovsky2008stability,KizilkaleVohra2024ConstrainedTradingNetworks}, yet interest in dynamic price adjustment goes back to the tâtonnement process of \citet{Walras1874Elements} and the seminal works of \citet{Arrow-Hurwicz-1958,Arrow-Hurwicz-1959}.\footnote{More recently, \citet{chen2016decentralized,Cheung-et-al-2020,Goktas-et-al-2023} studied dynamic stability and convergence in related settings.} To fill this gap for trading networks, we study two dynamic processes and prove their convergence to stable outcomes. The first is an offer-based dynamic, originally proposed by \citet{lock2024decentralized}, in which agents are repeatedly called upon to adjust their offers. The second is a stochastic clock price market: central price clocks track current prices and get adjusted throughout trading, but agents' behavior remains decentralized. Agents are called upon randomly to update their demand at current prices, a much weaker requirement on coordination than concurrent updating.

Our convergence results guarantee that the market reaches a stable outcome --- but they say nothing about who benefits, raising the question of surplus sharing and fairness. A Pareto-efficient or utilitarian-efficient outcome need not be equitable, and the distribution of the trading surplus naturally matters to market participants. However, a purely egalitarian solution is often not stable or efficient. One may therefore wish to seek egalitarian solutions among outcomes that satisfy some defined fairness criterion. We show that competitive equilibria are easily too restrictive, e.g.,~they are incompatible with the egalitarian social welfare order (leximin order) \citep{Moulin-2003}. Thus, we consider fairness inside the core.\footnote{Previous works have studied many fairness concepts, for example, the nucleolus \citep{Schmeidler-1969}, or the leximin, leximax, or min-spread order \citep{Vazirani-2022-Core-imputations,vazirani2025fair}.} Our focus is on the fundamental discussion of who contributes to efficiency and should thus be rewarded with a positive surplus; e.g.,~in order to be incentivized to participate in the market, if not for normative fairness reasons. Our results for the trading network market highlight structural properties: we show that even core stability constrains fairness considerations considerably.

\subsection{Contributions}

As our first main contribution, we develop a strategic foundation for competitive equilibria in the trading network market. In our offer-based trading network game, we introduce the class of $\varepsilon$-tight Nash equilibria. Our \cref{theorem:NE-to-CE} shows that, for sufficiently small $\varepsilon$ that depends on the maximum vertex degree of the market, we can extend any $\varepsilon$-tight Nash equilibrium outcome to an integral competitive equilibrium, in polynomial time.

The second main contribution consists in the analysis of two different decentralized dynamics and strong convergence guarantees. For the first trading dynamic, defined by \citet{lock2024decentralized}, we generalize their result---which requires a substantially different proof technique---to 3-sparse markets with fully substitutable agents and bounded offers (\cref{prop:3-sparse-convergence}). Notably, our proof provides a novel approach that relies on the geometric interpretation of the demand first introduced by \citet{baldwin2019understanding}, which we believe might be useful to further extend the result. The second dynamic we propose corresponds to the evolution of a stochastic clock auction \citep{milgrom2020clock}. Our convergence result in \cref{prop:price-dynamic-converges} is obtained by establishing connections to stochastic subgradient descent, and relying on the inherited techniques \citep{garrigos2023handbook,bubeck2015convexoptimizationalgorithmscomplexity}.

Our third contribution establishes a mapping from core utility imputations of the associated cooperative market game to core outcomes in the trading network market. \Cref{prop:core-imputation-implementation} demonstrates that, starting from any efficient allocation, any core imputation can be implemented by prices as a core outcome. This result is fundamental as it allows any consideration of surplus allocation and fairness to take place in the simpler cooperative market game.

This insight also leads to one of our main theorems, an impossibility. Agents are essential if they are part of every efficient allocation in a market, and inessential otherwise, concepts that have been studied for the assignment game \citep{Vazirani-2022-Core-imputations,vazirani2025fair}. Because any core utility imputation can be implemented from any efficient allocation, and because an inessential agent must get zero utility in at least one efficient allocation, it must get zero utility in any core outcome (\cref{theorem:inessential-agents-zero-utility}). This substantially generalizes the lone wolf theorems from maximum weight matchings in matching markets \citep{McVitieWilson1970} and competitive equilibria in transferable utility exchange economies \citep{JagadeesanKominersRheingansYoo2020}. For markets with three essential agents or fewer, essential agents obtain positive utility in the leximin core outcome (\cref{prop:pos-utility-3-agents}). However, this result breaks down in markets with four or more essential agents.

On a technical level, much of our reasoning applies geometric techniques from \citet{baldwin2019understanding} to the trading network market. While it was understood previously that the trading network market is nested in unimodular demand types, we provide the first geometric interpretation of trading bundles in price space (\cref{counterexample}) and demonstrate how to explicitly connect the trading market to a goods market (\Cref{app:markets-to-auctions}).

\subsection{Further Related Literature}
The study of transferable utility markets with indivisible goods has evolved from the assignment game \citep{Koopmans-et-al-1957,Gale-1960,Shapley-1971} to markets with very general, unimodular demands \citep{baldwin2019understanding}, within which the trading network model of \citet{hatfield2013stability} is nested as a prominent special case. Beyond this literature, our work connects to several strands of research on decentralized markets, price adjustment processes, and fairness considerations.

Convergence properties of price adjustment processes have received considerable attention. E.g.,~\citet{chen2016decentralized} study the convergence of a random decentralized process in a labor market, establishing that infinite cycles cannot arise. The possible existence of cycling behavior is also a central challenge in our analysis of the two best response dynamics we study. \citet{Cheung-et-al-2020} identify a class of markets admitting a convex potential function whose negative gradient coincides with excess demand, thereby obtaining various convergence guarantees in the continuous-time setting.

A related line of work studies dynamic clock auction formats with randomized elements. \citet{Christodoulou-2022} study clock auctions in which an auctioneer serves strategic (non-truthful) bidders subject to feasibility constraints. They design a clock auction mechanism that is optimal in terms of the competitive ratio of social welfare, and propose a second, randomized mechanism that mitigates adversarial behavior arising from deterministic tie-breaking rules. \citet{Feldman-2025} address the same setting and seek to overcome the limitations imposed by the auctioneer's lack of information about bidders' private valuations. They design mechanisms for different informational regimes and evaluate performance in terms of competitive ratio.

Recent articles by \citet{Dworczak-2021} and \citet{Akbarpour-2024} have discussed equity-efficiency trade-offs and the use of market vs.~non-market mechanisms. Their approach with social welfare weights allows for very general redistributive concerns; however, their market settings are more restrictive. In the assignment game, \citet{Vazirani-2022-Core-imputations,vazirani2025fair} studies fair core imputations with a focus on computational properties and developing algorithms. In the auction literature, surplus distribution has only recently garnered attention \citep{Finster-2025-equitable-auctions}. In a Bayesian setting with privately informed bidders, they show that the trade-off between common and private values in the bidders' value structure is essential in determining the equitable pricing format, and in most cases some degree of pay-as-bid pricing is indispensable. \citet{jeong2023first} discuss bid space restrictions to incorporate diverse objectives in first-price auctions. Finally, fairness has received particular attention in the allocation of online advertisement to users \citep{Celis-2019,Chawla-2022}.

\subsection{Organization}
The rest of the paper is organized as follows. 
In \Cref{sec:network-trading-market}, we present the setting by defining the trading network market, the main relevant notions, such as stability concepts, and we discuss a geometric interpretation of the agents' demand. 
\Cref{sec:the-network-trading-game} focuses on the analysis of the market from a game theoretic viewpoint, providing a link between the common solution concept of this field, Nash equilibrium, and competitive equilibria.
Two different decentralized dynamics are introduced in \Cref{sec:decentr-trading}, which also contains the results detailing their convergence. 
Finally, the analysis of fairness in the market and the associated impossibility result are discussed in \Cref{sec:stability-fairness}. 
The proofs of the main results are deferred to \Cref{app:proofs}. 
Lastly, \Cref{app:sec:additional-examples,app:markets-to-auctions} contain explanatory examples and additional material.

\subsection{Notation}
Vectors are written in bold. Let $\Omega$ be an arbitrary ground set. For any fixed $\omega \in \Omega$, we write $\eb^\omega \in \{0, 1\}^\Omega$ for the indicator vector with value $1$ for entry $\omega$, and value $0$ for all other entries. For any $\varepsilon > 0$, let $\varepsilon \Z = \{\varepsilon z \mid z \in \Z \}$, and let $\varepsilon \Z^\Omega = \{\varepsilon \zb \mid \zb \in \Z^{\vert \Omega\vert} \}$ be the discrete $\varepsilon$-lattice. For any two vectors $\vb$ and $\zb$, let $\vb\times\zb$ be their component-wise product.

\section{The Trading Network Market}\label{sec:network-trading-market}

We first introduce the market studied in \citep{hatfield2013stability}. A market $\market = (I, \Omega, v)$ consists of a finite set $I$ of agents, a finite set $\Omega$ of possible bilateral trades, and agent valuations $v = (v^i)_{i \in I}$. The trades $\Omega$ can represent any goods or services (such as a physical object, an insurance contract, or a spectrum license), and are typically heterogeneous. Each trade $\omega \in \Omega$ has a buyer $b(\omega) \in I$ and a distinct seller $s(\omega) \in I$. Sets of trades are also called \textit{bundles}. The set $\Omega$ can contain multiple trades with the same buyer and seller, and an agent can be the buyer for some trades and the seller for others. Hence, $\Omega$ can be understood as defining a directed multi-graph on vertices $I$ in which every directed edge represents a trade. In what follows, we denote by $\Delta$ the maximum vertex degree of the market's underlying network, that is, the maximum number of trades in which any individual agent is involved.

For any agent $i \in I$ and bundle $\Psi \subseteq \Omega$ of trades, $\Psi_i$ denotes the trades in $\Psi$ that involve agent $i$; and $\Psi_{i \rightarrow}$ and $\Psi_{i \leftarrow}$ respectively denote the agent's `selling' and `buying' trades in $\Psi_i$. For each agent $i$ and trade $\omega$, let $\chi^i_\omega=1$ if $i$ is the buyer of $\omega$, $\chi^i_{\omega}=-1$ if $i$ is the seller, and $\chi^i_{\omega}=0$ otherwise.

An agent's \textit{valuation function} $v^i$ maps every possible bundle $\Psi \subseteq \Omega_i$ of her trades to an integer value or $-\infty$, and the empty bundle $\emptyset$ to $0$.%
\footnote{We assume integrality of valuations throughout for convenience. Our results can be extended to any rational valuations by scaling. The value $0$ of $\emptyset$ is without loss of generality, because we can add constants to valuation functions without affecting quasi-linear demand.}
Values for certain bundles can be negative, for instance when they represent production costs for a seller or intermediary agent, and a value of $-\infty$ allows the agent to exclude technologically infeasible bundles.
The utility $u^i(\Psi, \pb)$ an agent has for bundle $\Psi$ at prices $\pb \in \R^{\Omega_i}$ is the agent's value for $\Psi$ plus the income from selling trades $\Psi_{i \rightarrow }$ minus the spending on buying trades $\Psi_{i \leftarrow}$:

\begin{equation}
\label{eq:utility}
u^i(\Psi, \pb) \coloneqq v^i(\Psi) - \sum_{\omega \in \Psi} \chi^i_\omega p_\omega.
\end{equation}

This gives rise to an agent's demand $D^i$ consisting of the bundle $\Psi \subseteq \Omega_i$ that maximizes her utility at prices $\pb$:
\begin{equation}
\label{eq:demand}
D^i(\pb) \coloneqq \argmax_{\Psi \subseteq \Omega_i} u^i(\Psi, \pb).
\end{equation}

For several of our results (but not all), we consider agents with fully-substitutable preferences. We will state explicitly when full substitutability is assumed. An agent's preference is \textit{fully substitutable} if reducing the prices of some buying trades (weakly) decreases her demand for the remaining buying trades and (weakly) increases her demand of selling trades; and, conversely, raising the prices of some selling trades (weakly) decreases her demand for the remaining selling trades and (weakly) increases her demand of buying trades.\footnote{Our definition of full substitutability is the one from Definition A.3 in \citet{Hatfield-2019}.} For agents with only buying trades, or only selling trades, this definition coincides with the standard definition of substitutes introduced by \citet{kelso1982job} in the context of labor markets, also prevalent in the auction literature \citep{Ausubel-2006,Sun-Yang-2006,baldwin2023solving}.

\begin{definition}
\label{definition:full-substitutability}
The demand of agent $i$ is \textit{fully-substitutable} if both

\begin{enumerate}[(i)]
\item
for all price vectors $p, p' \in \mathbb{R}^{\Omega}$ such that
$p_{\omega} = p'_{\omega}$ for all $\omega \in \Omega_{i\to}$ and
$p_{\omega} \ge p'_{\omega}$ for all $\omega \in \Omega_{i \gets}$,
for every $\Psi \in D_i(p)$ there exists $\Psi' \in D_i(p')$ such that we have 
$\{\omega \in \Psi'_{i \gets} : p_{\omega} = p'_{\omega}\} \subseteq \Psi_{i \gets}$ and $\Psi_{i\to} \subseteq \Psi'_{i\to}$;

\item
for all price vectors $p, p' \in \mathbb{R}^{\Omega}$ such that
$p_{\omega} = p'_{\omega}$ for all $\omega \in \Omega_{i \gets}$ and
$p_{\omega} \le p'_{\omega}$ for all $\omega \in \Omega_{i\to}$,
for every $\Psi \in D_i(p)$ there exists $\Psi' \in D_i(p')$ such that we have 
$\{\omega \in \Psi'_{i\to} : p_{\omega} = p'_{\omega}\} \subseteq \Psi_{i\to}$ and $\Psi_{i \gets} \subseteq \Psi'_{i \gets}$.
\end{enumerate}
\end{definition}

In many cases, it is useful to define a tie-breaking rule that consistently selects a single demanded bundle from $D^i(\pb)$. For this, assume an ordering $\Omega = \{\omega_1, \ldots, \omega_m \}$ on the trades. At any prices $\pb \in \R^\Omega$, define $d^i(\pb)$ as the lexicographically smallest bundle in $D^i(\pb)$ of maximum cardinality. In \cref{lem:tie_breaking_rule} in the appendix, we show that perturbing an agent's integral valuation function by $\sum_{\omega_j} \frac{\varepsilon}{4^j}$ results in a valuation at which $d^i(\pb)$ is uniquely demanded at any prices $\pb \in \varepsilon \Z^\Omega$ in the $\varepsilon$-lattice.

\subsection{Efficiency, Competitive Equilibrium, Core}\label{sec:efficiency-and-stability}

We now define social welfare, the notions of efficiency and the stability concepts of competitive equilibrium and the core. An \textit{allocation} $(\Phi)_{i \in I}$ consists of a bundle $\Phi_i \subseteq \Omega_i$ for each agent $i$. It is \textit{feasible} if every trade $\omega \in \Omega$ is either contained in the sets $\Phi_{s(\omega)}$ and $\Phi_{b(\omega)}$ of the trade's seller and buyer, or not contained in either. 
There is a straightforward one-to-one correspondence between all feasible allocations and the allocations represented by a single set of trades $\Phi \subseteq \Omega$ with $\Phi_i = \Phi \cap \Omega_i$ for each agent $i \in I$. As we will consider only feasible allocations, we represent an allocation by a single set of trades $\Phi$ from now on.

\begin{definition}\label{def:social-welfare-market-value-efficiency}
    The \textit{social welfare} of allocation $\Phi \subseteq \Omega$ is $\sum_{i \in I} v^i(\Phi_i)$. The \emph{market value} of a market $\market$ is the maximum social welfare achievable over all feasible allocations, i.e.,~$\max_{\Phi \subseteq \Omega} \sum_{i \in I} v^i(\Phi_i)$. An allocation $\Phi$ is \emph{efficient} if it attains the market value.
\end{definition}

A \textit{market arrangement} $(\pb, \Phi)$ consists of a set of prices $\pb \in \R^\Omega$ and a feasible allocation $\Phi$ of trades. We are particularly interested in arrangements that give each agent a demanded bundle.

\begin{definition}[Competitive Equilibrium]
    A market arrangement $(\pb, \Phi)$ is a \textit{competitive equilibrium} iff, for every agent $i$, we have $\Phi_i \in D^i(\pb)$.
\end{definition}

\citet{hatfield2013stability} show that the first and second welfare theorems hold: (1) The allocation of a competitive equilibrium is efficient, and (2) every efficient set of trades can be extended by prices to form a competitive equilibrium. Moreover, a competitive equilibrium is guaranteed to exist in the trading network market if all agents' preferences are fully substitutable. We provide an alternative, substantially shorter, proof of the existence of competitive equilibrium and the first and second welfare theorems in \cref{app:markets-to-auctions}.\footnote{Note that substitutability in trading markets allows a certain complementarity between selling and buying trades, which does not break full substitutability. Generally, the substitutes condition is required to ensure the existence of competitive equilibria, with a few exceptions in markets with complements (see, e.g.,~\citet{Sun-Yang-2006} and \citet{Finster-2025}.}

We also define \textit{market outcomes} $(\pb, \Phi)$, which differ from an arrangement by the fact that the vector of prices $\pb \in \R^\Phi$ is indexed only by the trades in $\Phi$, and so the prices of all other trades are not specified.%
\footnote{Our definition of market outcomes is similar to \citet{hatfield2013stability}. They, however, define market outcomes using ``contracts'', a notion that we dispense with for simplicity.}
We will rely on this distinction in our treatment of decentralized trading.
Letting $a(\Psi)$ be the set of agents involved in at least one of the trades $\Psi$, we define the core among market outcomes by analogy with \citet{hatfield2013stability}.

\begin{definition}[Core of Market Outcomes]\label{definition:core-of-market-outcomes}
    A market outcome $(\pb, \Phi)$ is in the core if there exists no \textit{blocking outcome} $(\qb, \Psi)$ that satisfies $u^i(\qb, \Psi_i) \geq u^i(\pb, \Phi_i)$ for all agents $i \in C$ in coalition $C = a(\Psi)$, and strict inequality holds for at least one agent in $C$.
\end{definition}

Intuitively, an outcome is in the core if there exists no coalition of agents that can improve by agreeing on a different outcome between them, with at least one agent improving strictly. Results by \citet{hatfield2013stability} imply that every competitive equilibrium can be turned into a core outcome by restricting prices to the active trades, but not every core outcome can be extended to a competitive equilibrium by specifying prices for the inactive trades.

\subsection{The Geometry of Substitutes}
\label{section:geometry-preferences}

Many proofs in our paper rely on a geometric characterization of demand with indivisible goods introduced in \citet{baldwin2019understanding}. We summarize the main ideas of this geometric approach here, and refer to \citet[Appendix C]{baldwin2024implementing} for a rigorous treatment.

First, we introduce an alternative vector notation for bundles, which allows us to immediately apply the geometric tools of \citet{baldwin2019understanding}. A bundle is represented by a vector $\xb \in \{-1,0,1\}^\Omega$, where $x_\omega = -1$ denotes a \textit{sale} of trade $\omega$, $1$ denotes a \textit{purchase}, and $0$ denotes neither. For any agent $i \in I$, the bundle $\Psi_i \subseteq \Omega_i$ is thus represented by the vector $\xb \in \R^\Omega$ defined as
\[
x_\omega =
\begin{cases}
1 & \text{if } \omega \in \Psi_{i \leftarrow}, \\
-1 & \text{if } \omega \in \Psi_{i \rightarrow}, \\
0 & \text{if } \omega \in \Omega \setminus \Psi_i.
\end{cases}
\]
With this notation, we see that full substitutability is equivalent to the standard substitutes condition from the literature. This was also noted in \citep{hatfield2015full}.
\begin{definition}[\citep{milgrom2009substitute}]
A valuation $v$ is \emph{substitutes} if, for any prices $\pb' \geq \pb$ and any $\xb\in D_v(\pb)$, there exists $\xb'\in D_v(\pb')$ such that $x'_k\geq x_k$ for all $k$ such that $p_k=p'_k$.
\end{definition}
The vector notation also allows us to conveniently describe the geometric approach of \citep{baldwin2019understanding}. For any valuation $v$ and the demand correspondence $D$ arising from it, we define the \textit{Locus of Indifference Prices (LIP)} $\LIP_v = \{\pb \mid |D_v(\pb)| \geq 2 \}$ as the prices at which two or more bundles are demanded. The LIP of any valuation can be decomposed into $(n-1)$-dimensional linear pieces we call \textit{facets}. The connected components of $\R^\Omega \setminus \LIP_v$ form open convex polytopes referred to as \textit{unique demand regions (UDRs)}, as the demand is unique and constant at all prices within a UDR.

A LIP $\LIP_v$ is extended to a \textit{weighted LIP} $(\LIP, w)$ by associating each facet $F$ of $\LIP_v$ with weight $w_v(F)$ given by the (positive) greatest common divisor of the coordinate entries of $\xb - \yb$, where $\xb$ and $\yb$ are the bundles demanded in the UDRs on either side of $F$. Thus, the change in demand as we cross a facet $F$ is given by $w_v(F) \nb$, with $\nb$ being the normal vector of $F$ expressed as a primitive integer vector that points in the opposite direction to the price change. In our network market setting, the facets' weights satisfy $w_v(F) \in \{0,1\}$, as every trade is demanded in one copy and has only one seller and one buyer.
The orientations of facets also allow us to characterize substitutes geometrically. We use this characterization heavily in proofs.
\begin{fact}[{\citep{baldwin2019understanding}}]
\label{fact:substitutes-normals}
A valuation $v$ is substitutes iff every facet of its LIP is normal to $\eb^\omega$ or $\eb^\omega - \eb^\chi$ for some $\omega, \chi \in \Omega$.
\end{fact}

\cref{example:counterexample} illustrates how valuations translate into a geometric representation of individual and aggregate demand correspondences, highlighting especially the orientation of facets. The LIPs are represented by the black lines. Note that all such lines are horizontal, vertical, or have slope $1$; the conditions of \cref{fact:substitutes-normals} in two dimensions.

\begin{example}[label=counterexample]\label{example:counterexample}
Consider the following market with two agents, $I = \{1,2\}$, and two trades, $\Omega = \{\omega, \chi\}$. The trading network and the agents' valuations are given below. Agent 1's preferences are purely additive: the bundle's value for the set of trades $\omega$ and $\chi$ is simply the sum of individual trades. Note that agent 2's valuation is fully substitutable. In fact, the substitutes property is salient: being a trader of both trades gives the agent no higher value than being only a buyer of $\omega$. However, since the price of $\chi$ is received not paid, at high prices for $\chi$ and low prices of $\omega$, relatively speaking, she prefers to engage in both trades (see the upper left region in \cref{fig:demand-agent2}).
The shaded blue area in \cref{fig:aggregate-demand} marks the set of competitive equilibrium prices. At those prices, the individually demanded trades sum to the aggregate demand of $(0,0)$, and therefore the market~clears.
\end{example}

\begin{figure}[t]
\centering

\begin{subfigure}[t]{0.49\textwidth}
    \centering
    \begin{tikzpicture}[xscale=3.5, yscale=1.8]
        \node[agent] (s) at (0,0) {$1$};
        \node[agent] (b) at (1,0) {$2$};
        \draw[trade] (s) to[bend left] node[midway,fill=white] {$\omega$} (b);
        \draw[trade] (b) to[bend left] node[midway,fill=white] {$\chi$} (s);
    \end{tikzpicture}
    \caption{Trading network}
    \label{fig:two-agent-network}
\end{subfigure}\hfill
\begin{subfigure}[c]{0.49\textwidth}
    \centering
    \begin{tabular}{cccc}
        \toprule
        & $\{\omega\}$ & $\{\chi\}$ & $\{\omega,\chi\}$ \\
        \midrule
        Agent 1 & $-1$ & $1$ & $0$ \\
        Agent 2 & $1$ & $-1$ & $1$ \\
        \bottomrule
    \end{tabular}
    \caption{Valuations}
    \label{fig:two-agent-valuations}
\end{subfigure}

\vspace{0.6em}

\begin{subfigure}[t]{0.32\textwidth}
    \centering
    \includegraphics[width=\textwidth]{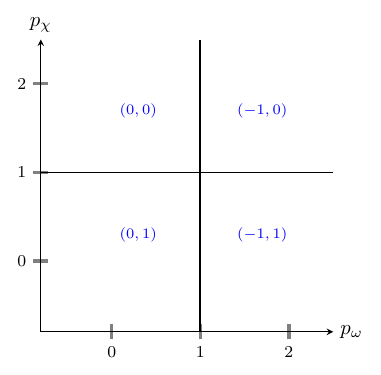}
    \caption{Demand of agent 1}
    \label{fig:demand-agent1}
\end{subfigure}\hfill
\begin{subfigure}[t]{0.32\textwidth}
    \centering
    \includegraphics[width=\textwidth]{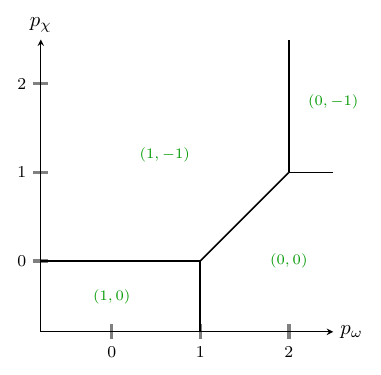}
    \caption{Demand of agent 2}
    \label{fig:demand-agent2}
\end{subfigure}\hfill
\begin{subfigure}[t]{0.32\textwidth}
    \centering
    \includegraphics[width=\textwidth]{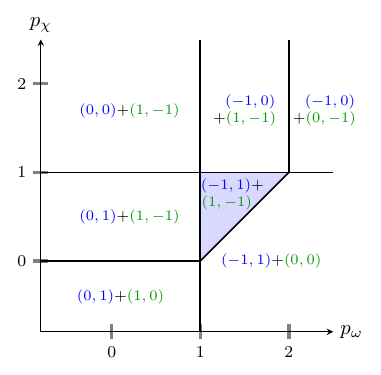}
    \caption{Aggregate demand}
    \label{fig:aggregate-demand}
\end{subfigure}

\caption{A two-agent trading network, associated valuations, and resulting demands.}
\label{fig:network-valuations-demand}
\end{figure}
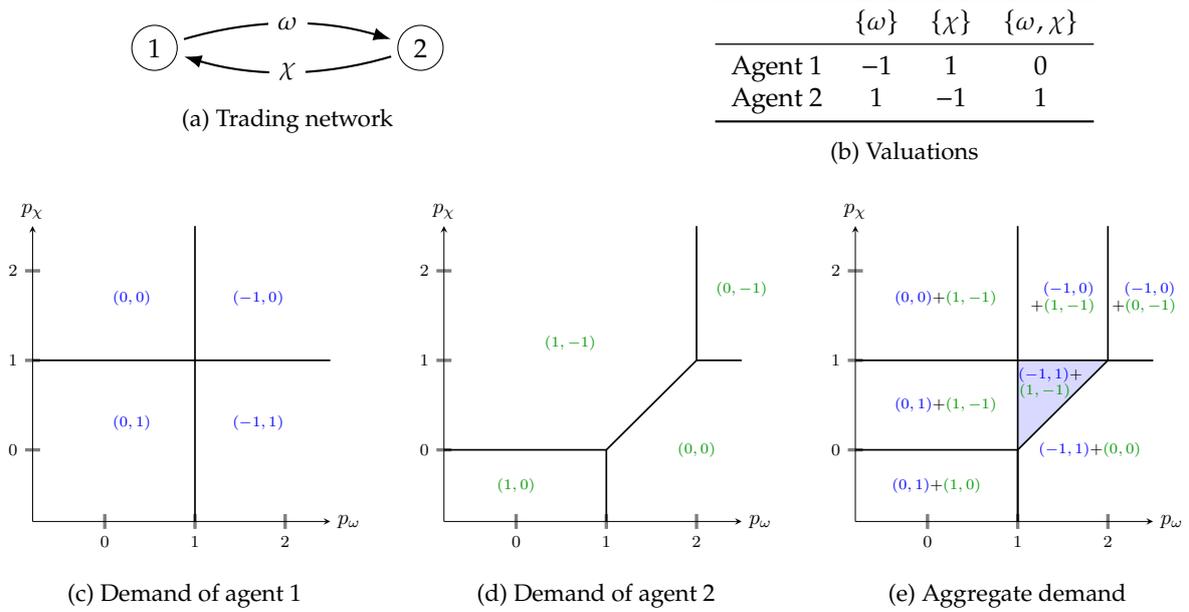

\section{The Trading Network Game}\label{sec:the-network-trading-game}
We introduce the trading network game, a strategic game played on a trading network market in which each agent submits offers for her buying and selling trades. Our central equilibrium notion is that of an \textit{$\varepsilon$-tight Nash equilibrium}, in which the buying and selling offers on every trade are within $\varepsilon$ of each other. We show that, for sufficiently small $\varepsilon$, every $\varepsilon$-tight Nash equilibrium outcome can be extended to a competitive equilibrium (\cref{theorem:NE-to-CE}). Building on this correspondence, we then analyze decentralized trading dynamics and show that they converge to $\varepsilon$-tight Nash equilibria---and hence, by \cref{theorem:NE-to-CE}, to competitive equilibria---or to competitive equilibria directly.

The \textit{trading network game} is played on a trading network market. Each agent $i$ makes an offer $\offer^i_{\omega}$ for each of its trades $\omega \in \Omega_i$. Hence, each agent's action space is $\R^{\Omega_i}$, and we denote an agent's pure strategy by the vector $\offers^i$. We are interested only in pure strategies. We write $\offers^{-i}$ for the vector of offers that agent $i$ observes from her trading counterparts.

A trade is \textit{active} if its buyer's offer is equal to the seller's offer, and \textit{inactive} otherwise. Trading prices $\pb \in \R^\Phi$ for active trades $\Phi$ are given by $p_\omega = \offer^{b(\omega)}_\omega = \offer^{s(\omega)}_\omega$. These prices $\pb$, together with the active trades $\Phi$, thus lead to the \textit{market outcome $(\pb, \Phi)$ of $\offers$}. The utility of agent $i$ from $\offers$ is given by
\[
u^i(\offers) \coloneqq u^i(\Phi, \pb),
\]
where $(\pb, \Phi)$ is the outcome of $\offers$. This allows us to define Nash equilibria in our game.

\begin{definition}
    A \textit{Nash equilibrium} in the trading network game consists of offers $\offers$ such that no agent can profitably deviate unilaterally by updating her offers.
    A Nash Equilibrium $\offers$ is \textit{$\varepsilon$-tight} if $|\offer^{b(\omega)}_\omega - \offer^{s(\omega)}_\omega| \leq \varepsilon$ for every trade $\omega \in \Omega$.
\end{definition}

\begin{observation}
    In any Nash equilibrium $\offers$, we have $\offer^{b(\omega)}_\omega \leq \offer^{s(\omega)}_{\omega}$ for every trade $\omega$, and $\offer^{b(\omega)}_{\omega} = \offer^{s(\omega)}_{\omega}$ if the trade is active.
\end{observation}

It is immediate that every game admits at least one Nash equilibrium. If all buying offers on all trades are very small (e.g.,~$-\infty$), and the selling offers are very high (e.g.,~$\infty$), for instance, then no trade is active and no agent can improve by updating her offer. This example also demonstrates that the outcomes of Nash equilibria need not lie in the core. Conversely, \cref{example:core-not-NE} in \cref{app:sec:additional-examples} demonstrates that not every core outcome can be implemented by a Nash equilibrium.

The ``bad'' equilibria with low buying and high selling offers are clearly not $\varepsilon$-tight for sufficiently small $\varepsilon > 0$. As we are interested in Nash equilibria that achieve efficient outcomes, we exclude such degenerate settings and focus on $\varepsilon$-tight Nash equilibria with small values for $\varepsilon$.

To guarantee existence of $\varepsilon$-tight Nash equilibria, we must restrict the preference domain. For instance, as we show, fully substitutable preferences are sufficient for the existence. In contrast, \cref{example:two-agents-two-opposite-trades} in \cref{app:sec:additional-examples} demonstrates non-existence of Nash equilibrium in a market with one substitutes agent and one complements agent if $\varepsilon$ is chosen sufficiently small.

More generally, if a market admits a competitive equilibrium, it is straightforward to construct an $\varepsilon$-tight Nash equilibrium for any $\varepsilon > 0$ from it. This holds even if not all agents are substitutes: For each active trade, set both the buying and selling offer to the price; for each inactive trade, subtract $\frac{\varepsilon}{2}$ from the price for the buying offer, and add this term to the selling offer. Intuitively, an agent unwilling to buy a trade $\omega$ at $p_\omega$ will also be unwilling to buy at $p_\omega + \frac{\varepsilon}{2}$, and similarly for a sale. The full proof is deferred to \Cref{proof:proposition:substitutes-implies-NE}.

\begin{proposition}
\label{proposition:substitutes-implies-NE}
Fix an arbitrary $\varepsilon > 0$. If a market admits a competitive equilibrium, then there exists an $\varepsilon$-tight Nash equilibrium. In particular, every market with fully substitutable preferences admits an $\varepsilon$-tight Nash equilibrium.
\end{proposition}

\subsection{From Nash Equilibria to Competitive Equilibria}\label{sec:decentral-trading-impl-comp-equilibria}

\Cref{proposition:substitutes-implies-NE} demonstrates a strong connection between competitive and Nash equilibria, as Nash equilibria can be constructed from competitive equilibria. We now explore the converse, by understanding when a competitive equilibrium can be constructed from a tight Nash equilibrium in a way that preserves the trades which are active.

Recall that $\Delta$ is the maximum vertex degree in the network of the given market. We also say that a market arrangement is an \textit{$\varepsilon$-approximate competitive equilibrium} if every agent receives a bundle at $\pb$ that is within $\varepsilon$ of maximizing utility.  

\begin{definition}
\label{definition:competitive-equilibrium}
A market arrangement $(\pb, \Phi)$ is an \textit{$\varepsilon$-approximate competitive equilibrium} if
\[
u^i(\Phi_i, \pb) \geq u^i(\Psi, \pb) - \varepsilon, \forall \Psi \subseteq \Omega_i,\ \text{ for all } i \in I.
\]
\end{definition}
The following approximation result is immediate from the definition of demand and quasilinearity.
\begin{proposition}\label{prop:tight-NE-is-approx-CE}
Let $\Delta$ be the maximal vertex degree of any agent in the market network. If $\offers$ is an $\varepsilon$-tight Nash equilibrium with outcome $(\pb, \Phi)$, then we can construct an $\varepsilon\Delta$-approximate competitive equilibrium with price $p_\omega$ for each active trade $\omega \in \Phi$.
\end{proposition}
The proof of this proposition is in \Cref{proof:prop:tight-NE-is-approx-CE}. When $\varepsilon > 0$ is sufficiently small, we can improve on this result by constructing an exact competitive equilibrium from any $\varepsilon$-tight Nash equilibrium.

\begin{theorem}\label{theorem:NE-to-CE}
Fix a market. For any strictly positive $\varepsilon < \frac{1}{2\Delta - 2}$, the outcome of any $\varepsilon$-tight Nash equilibrium $\offers$ can be extended to an integral competitive equilibrium in polynomial time.
\end{theorem}
\begin{proof}[Proof sketch.]
(The complete proof can be found in \Cref{proof:theorem:NE-to-CE}.)

Let $\offers$ be an $\varepsilon$-tight Nash equilibrium, with $\varepsilon < \frac{1}{2\Delta -2}$. We demonstrate that we can extend the outcome of the Nash equilibrium to a competitive equilibrium by finding integral prices for the inactive trades.
The first part of our proof proceeds along an analogous argument to the proof of Theorem 6 in \citep{hatfield2013stability}. First we reduce to a sub-market consisting only of the inactive trades, by removing the active trades and modifying each agent's valuation to reflect that the active trades are always available at their given prices. Specifically, for any agent $i$ and any subset $\Psi$ of inactive trades, the modified valuation captures the maximum utility the agent can achieve by combining trades from $\Psi$ with the active trades. This modified valuation inherits full substitutability from the original preferences.

The key step is then to show that this sub-market admits a competitive equilibrium in which every agent demands the empty bundle. This is the most technically challenging part of the proof and deviates from the proof in \citep{hatfield2013stability}. Once we have equilibrium prices for the sub-market where all inactive trades are rejected, we extend these prices back to the original market, combining them with the prices from the Nash equilibrium's outcome for the active trades to obtain a full competitive equilibrium.

To tackle the technically challenging part of the proof, consider a sub-market with only inactive trades, so each agent demands the empty bundle at their counterparties' current offers. The key insight is that the set $P^i$ of prices at which any agent $i$ demands nothing forms a convex polyhedron. Moreover, because preferences are fully substitutable, we can apply our geometric understanding of demand to show that these polyhedra have a special structure, in the sense that their boundaries are facets normal to specific ``substitutes'' vectors.
The main technical step constructs equilibrium prices by carefully combining these individual polyhedra. The proof shifts each agent's polyhedron $P^i$ slightly (by $\varepsilon \eb^{\Omega_{i \leftarrow}}$), finds a vertex $\vb$ of the intersection of these shifted polyhedra, and then rounds $\vb$ to obtain the final integral prices $\pb^*$. The bound $\varepsilon < \frac{1}{2\Delta-2}$ ensures that the rounded vector lies in each $P^i$, so each agent demands $\emptyset$ at $\pb^*$. Moreover, these final prices turn out to be simply the floor of each seller's offer or the ceiling of each buyer's offer; this follows from the geometric structure of substitutes preferences with integral valuations.
\end{proof}

The upper bound on $\varepsilon$ in \cref{theorem:NE-to-CE} is tight. \cref{example:counterexample} demonstrates, for instance, that the theorem statement fails to hold when $\varepsilon = \frac{1}{2\Delta - 2}$.

\begin{example}[continues=counterexample]
Let $\varepsilon \coloneqq \frac{1}{2\Delta -2} = \frac{1}{2}$. Clearly, the market admits an $\varepsilon$-tight Nash equilibrium $\offers$ given by $\offers^1 = (1.5, 0.5)$ and $\offers^2 = (1, 1)$, and neither trade happens. We now argue that there do not exist any prices for these two trades at which both agents demand $\emptyset$.
The set of prices at which agent $1$ demands $\emptyset$ is given by $\{\qb \in \R^\Omega \mid q_\omega \leq 1 \text{ and } q_\chi \geq 1\}$, and the set of prices at which agent $2$ demands $\emptyset$ is $\{\qb \in \R^\Omega \mid q_\omega \geq 1, q_\chi \leq 1, q_\omega - q_\chi \geq 1\}$. As these two sets are disjoint, there are no prices at which both agents demand $\emptyset$.
\end{example}

\section{Decentralized Trading in the Network Game}\label{sec:decentr-trading}

In this section, we discuss two models of decentralized dynamics behavior in the trading network game, and prove that the proposed dynamics indeed converge to an $\varepsilon$-tight Nash equilibrium and an $\varepsilon$-approximate competitive equilibrium, respectively.

Our first procedure is an offer-based best response dynamics due to \citet{lock2024decentralized}. We prove convergence of this dynamics for $3$-sparse market networks while \citet{lock2024decentralized} does it only for $2$-sparse networks, but conjecture it for all $m$, with experimental evidence. In this dynamics environment, agents are ``activated'' to change their strategy, i.e.,~the offers towards their trading partners. In doing so, they act myopically and adjust their offers only according to whether their current demand at current prices leads to clearing the trades they are facing or not.

Our second trading dynamics is a decentralized, stochastic clock market. As in conventional clock auctions, there is a central shared price for each trade. Instead of updating their own offer strategies, agents only report their demand at given prices to a central auctioneer. Subsequently, prices of the trades the agent does not demand change upwards for selling trades, and downwards for buying trades. The market is still partially decentralized: instead of all agents reporting their demand at current prices, the updating process is stochastic, with one agent at a time, chosen randomly, updating and reporting their demand. In networked markets, the fully centralized counterpart, in which all agents report their updated demand at current prices simultaneously, suffers from increased coordination complexity. Thus, we focus in this article on the decentralized version, and leave the fully centralized clock auction (resembling a tâtonnement process) to future work.\footnote{For concurrent updating of agents' demands, the natural update rule would be to increase the price of trades demanded by a buyer but not a seller, and vice versa, and leave the price of all trades that are (not) demanded by both buyer and seller untouched. As each agent may be involved in multiple trades with many different trading partners, the resulting price path remains complex, and convergence an open question.} The decentralized version would be a natural candidate to implement as an equilibrium-finding algorithm on the blockchain, automating the reporting of myopic demands.

\subsection{An Offer-Based Best Response Dynamic}
\label{sec:convergence-of-decentralized-trading}

Our proof of convergence for the offer-based best response dynamics requires significant new ideas compared to the proof for $2$-sparse markets by \citet{lock2024decentralized}. We believe our techniques might also lead to promising approaches for the case of $m$-sparse markets with $m \geq 4$. Moreover, we prove convergence to an $\varepsilon$-tight Nash equilibrium; thus, by our \cref{theorem:NE-to-CE} the outcome of this trading dynamics is a competitive equilibrium, too.

The dynamics provides a stylized depiction of how agents may interact with each other in real-world markets, while making some concessions to tractability. Agents make offers to their trading partners, as described in the trading network game. In a dynamic, multi-period setting, these offers may be adapted across periods. It is assumed that in each period, a single agent is \emph{active} and allowed to change their offers. The extent to which an agent changes their offer is reactive to their trading partners: it conveys the message ``I am willing/not willing to trade at your proposed price''. If for some trade, the active agent is not willing to trade, they set their counter offer for this trade to an $\varepsilon$-lower price, otherwise they match the offer they are facing on that trade. 

Agents are activated uniformly at random. However, an agent who was previously active, cannot be activated again until one of the offers they are facing from their own trading partners has changed. More formally, the best response dynamics maintains a set $U$ of unsatisfied agents, which initially contains all agents. In each round, one agent is selected uniformly at random from $U$. Once agent $i$ has best responded, she is removed from the set $U$ while any neighboring agents who have witnessed a modified offer from agent $i$ are added back to $U$. The process continues until all agents are satisfied. The full definition of the dynamics is given in \cref{alg:offer-based-dynamics}. Moreover, \Cref{ex:dynamics} illustrates an example of the evolution of the dynamics for a simple 3-agent market.

Note that we additionally assume that, for each agent, the vector of offers $\offers$ is bounded. Equivalently, this means that for every feasible bundle of trades $\Phi \in \Omega_i$, the valuation 
$v^i(\Phi)$ is finite. This assumption is not restrictive, since the bound can be chosen arbitrarily large. Moreover, Lemma 11 in \citet{lock2024decentralized} shows that, under such boundedness of offers, the resulting dynamics either converge or enter a cycle. Furthermore, we define an $m$-sparse market as follows.

\begin{algorithm}[htp]
\begin{algorithmic}[1]
\Require Market $\market$ with agents $I$, trades $\Omega$, and step size $\varepsilon > 0$.
\Ensure Nash Equilibrium $\offers$ of the Network Market Game.
\State Initialize offers $\offers$ for all trades and agents.
\State Let $U = I$ be the set of unsatisfied agents.
\While{$U \neq \emptyset$}
    \State Sample an agent $i \sim U$ uniformly at random.
    \State Let $\pb$ be a vector of prices such that $p_\omega = \sigma^{-i}_\omega$ for all $\omega \in \Omega_i$ be the offers facing agent $i$.
    \State Determine bundle $d^i(\pb)$ demanded by agent $i$ at prices $\pb$.
    \State Update agent $i$'s offer for each trade $\omega \in \Omega_i$ to $\sigma^i_\omega \leftarrow p_\omega$ if $\omega \in d^i(\pb)$, and $\sigma^i_\omega \leftarrow p_\omega - \varepsilon \chi^i_\omega$ otherwise.
    \State Remove agent $i$ from $U$.
    \State Add all agents $j \neq i$ facing modified offers to $U$.
\EndWhile
\State \Return final offers $\offers$.
\end{algorithmic}
\caption{The offer dynamics from \citet{lock2024decentralized}.}
\label{alg:offer-based-dynamics}
\end{algorithm}

\begin{definition}
A market is $m$-sparse for $m \in \mathbb{N}$ if every subgraph of its underlying network can be divided into two or more disjoint components by removing at most $m$ trades.
\end{definition}

Then, we recover the following convergence result.
\begin{proposition}
\label{prop:3-sparse-convergence}
For any $3$-sparse market with fully substitutable agents and bounded offers, the best response dynamics described in \cref{alg:offer-based-dynamics} converges to an $\varepsilon$-tight Nash equilibrium.
\end{proposition}

\begin{proof}[Proof sketch]
(The complete proof can be found in \Cref{app:convergence-offers-dynamic}.)

Firstly, we leverage the equivalence between an $m$-sparse market and a market with two agents and $m$ trades, established by \citet{lock2024decentralized}. Thus, we restrict ourselves to such a market.

Our proof considers a potential function $\varphi$ that counts, for given offers $\offers$, how many trades the two agents disagree on. We prove that this potential function weakly decreases throughout the dynamics by making use of a geometric characterization of demand described in \cref{section:geometry-preferences}. This allows us to characterize the ways in which the demand of the non-best responding agent evolves during an update step of the opponent.

Furthermore, the same lemma allows us to fully identify the cases in which the dynamics is stationary, reducing them to three possible scenarios.
The first two, corresponding to when the sequence of reached offers does not loop, or loops in a unique demand area, are easily proved to be impossible as offers are bounded.

In the third case, the dynamics performs a stationary cycle, i.e.,~the demand bundle evolves while the potential function remains constant. We show this case is not possible in a market with only three trades either. This is done by defining a specific partition of $\Omega$ into four sets, and showing that during one of these loops, a trade has to eventually belong to each of these sets. 
Finally, proving that an update to the composition of sets in the partition is possible only through swaps of trades, we conclude that this cannot occur if all four sets do not count at least one element, implying that at least 4 trades are necessary for this case to occur.
\end{proof}

\subsection{The Stochastic Clock Market}\label{sec:clock-market}

In our stochastic clock market, agents interact not via offers, but their myopic interaction simply consists of reporting their demand at current prices. The market is still decentralized: agents are sampled uniformly at random and ``activated''. In each round, only the active agent reports her current demand. A central clock tracks the price of each trade. This price vector is modified in accordance with the demand reports. Formally the update of the prices is defined as follows.
\begin{definition}[Price updates]
    Let $i$ be the active agent at time $t+1$ and $\pb^t$ be the current vector of prices. Let $\Phi \in D^i(\pb^t)$ be the demanded bundle of agent $i$ at prices $\pb^t$. The new prices $\pb^{t+1}$ are
    \begin{align}\label{eq:price_update}
        p^{t+1}_\omega \coloneqq \begin{cases}
             p^t_\omega &\text{if $\omega \in \Phi$},\\
            p^t_\omega - \varepsilon \chi_{\omega}^i &\text{else}.
        \end{cases}
    \end{align}
\end{definition}
The prices of trades demanded by the active agents at current prices $\pb^t$ remain unchanged. The prices of buying trades not demanded at $\pb^t$ are adjusted downwards, while the prices of the selling trades not demanded at $\pb^t$ are adjusted upwards.

In contrast to the offer-based best response dynamics, agents are sampled uniformly at random from the entire population of agents (not merely among unsatisfied agents). As a consequence, an agent may be selected to update their demand repeatedly in consecutive rounds without the dynamics reaching a stationary state.
The complete dynamics is formally defined in \Cref{alg:price-based-dynamics}. \Cref{ex:dynamics} shows an example of the evolution of the dynamics over a simple market, and exemplifies the difference to the offer dynamics in the previous section.

\begin{algorithm}[tb!]
\begin{algorithmic}[1]
\Require Market $\market$ with agents $I$ and trades $\Omega$; number of time steps $T$, and bound on the valuation function $R$; step size $\varepsilon > 0$.
\State Initialize arbitrary market prices $\pb^0 \in [-R, R]^{\vert\Omega\vert}$.
\For{$t \in \{1, \ldots, T\}$}
    \State Pick an agent $i \in I$ uniformly at random.
    \State Compute bundle $d^i(\pb^t)$ demanded by the active agent $i$ at $\pb^t$.
    \State Update price $\pb^{t+1}$ by letting $p_{\omega}^{t+1} \leftarrow p_{\omega}^t - \varepsilon \chi^i_\omega$ if $\omega \in \Omega_i \setminus d^i(\pb^t)$.
\EndFor
\State \Return $\pb^T$.
\end{algorithmic}
\caption{Price dynamics}
\label{alg:price-based-dynamics}

\end{algorithm}
As before, the vector of prices $\pb$ is assumed to be bounded throughout the interaction. This is equivalent to requiring that the valuation function of the agents is bounded, which excludes the cases where an agent either demands or rejects a trade no matter the price. The following holds.

\begin{proposition}
\label{prop:price-dynamic-converges}
Let $\market$ be a market with maximum vertex degree $\Delta$, where the absolute value of the valuation function is bounded by $R$ for every agent. Let $P^*$ be the set of its competitive equilibrium prices, and $(\pb^t)_{t\leq T}$ the sequence of prices generated by \cref{alg:price-based-dynamics} with step size $\varepsilon = R\sqrt{\frac{2m}{T\Delta}}$. For any $T\geq 2n^2R^2m\Delta$,
\[
    \bbE[d(\bar{\pb}, P^*)] \leq \mathcal{O}\left(\frac{1}{\sqrt{T}}\right),
\]
where $\bar{\pb} = \frac{1}{T} \sum_{t=1}^T \pb^t$ and  $d(\pb, P) = \min_{\qb \in P} \| \pb - \qb \|_\infty$.
\end{proposition}

\begin{proof}[Proof Sketch]
(For the complete proof, see \Cref{proof:prop:price-dynamic-converges}.) The proof proceeds by introducing a Lyapunov function $L$, which sums up the indirect utilities of all agents at prices $\pb$. $L$ can be composed as the sum of functions $L^i$, one per agent $i \in I$. We then show that the price update performed by agent $i$ in \cref{alg:price-based-dynamics} is equivalent to a stochastic subgradient descent step of $L$. 
In fact, unlike the offer-based dynamics, the active agent at each round is chosen in the complete set $I$, thus the resulting update is an unbiased estimate of a subgradient of the Lyapunov function. By convexity of $L$, this implies convergence of the function value to its minimum.

Finally, since the slope of $L$ outside the equilibrium price region is bounded from below by one, reflecting the fact that an agent must adjust at least one trade, we infer that small variations in the value of $L$ correspond to small deviations of the price vector from equilibrium. This establishes convergence and completes the proof.
\end{proof}

\section{Stability and Fairness}\label{sec:stability-fairness}

In the preceding sections, we analyzed decentralized trading and stability in the trading network game. $\varepsilon$-tight Nash equilibria are not only stable as mutual best response offers, but also as competitive equilibria. However, competitive equilibria, while socially efficient, need not always distribute the generated social welfare equitably between agents in the market. In this section, we address the question of fairness in the assignment of surplus between agents. 

Fairness cannot be evaluated as a single-dimensional criterion. For example, not executing any trades in the market is always a fair outcome yielding zero surplus to everyone; but this is clearly inefficient whenever the market value is positive. Moreover, a fair market outcome that is not stable is not desirable for market participants and a regulator or designer. Thus, we will be looking for fair, efficient and stable outcomes. The analysis turns out to be rather nuanced; we discuss different fairness notions that redistribute surplus within the core, and we establish how they relate. Our final, main result is a strong impossibility: in the core, inessential agents must always get zero utility. This is a substantial generalization of the lone wolf theorems for matching markets \citep{McVitieWilson1970} and competitive equilibria in exchange economies with transferable utility \citep{JagadeesanKominersRheingansYoo2020} to the core in the trading network market.

\subsection{The Cooperative Market Game and Implementation}

Competitive equilibrium as a stability notion can be restrictive in the context of redistribution. We illustrate this in \cref{example:leximin-cannot-be-CE} below. Thus, we consider fairness in the core and define the associated cooperative game on the trading network market.

A \textit{coalition} is a subset $C$ of agents $I$ of the market. We can restrict any market $\mathcal{M}$ to coalition $C$ by removing all other agents and all trades involving these agents. The \textit{characteristic function} $w : 2^I \to \R_+$ maps every coalition $C \subseteq I$ to the social welfare $w(C)$ generated by the coalition.
Letting $\Omega_C$ be all the trades with buyer and seller in $C$, we define

\[
    w(C) \coloneqq \max_{\Phi \subseteq \Omega_C} \sum_{i \in C} v^i(\Phi_i).
\]

\begin{definition}[Cooperative market game]
    The \textit{cooperative market game} $(I, w)$ of a market $\market$ consists of the agents $I$ and the characteristic function $w$.
\end{definition}

\begin{definition}[Core of the cooperative market game]\label{definition:core-imputation}
    The core of the cooperative market game is the set of utility imputations $\xb \in \R^I_+$ that satisfy $\sum_{i \in I} x_i = w(I)$ and $\sum_{i \in C} x_i \geq w(C)$ for all $C \subsetneq I$.
\end{definition}

From the definition, it is clear that the core is an $|I|$-dimensional convex polytope.

Since our original quest is to analyze fairness and stability in the trading market, the first question we answer is one of implementation. Suppose we found a fair distribution of utilities in the core, could we implement those as an outcome, that is, through allocations of sets of trades and corresponding prices? We answer this affirmatively in \cref{prop:core-imputation-implementation}. We show that there is a one-to-one correspondence between utility profiles that arise from core outcomes, and core imputations in the new game. Trivially, utility profiles that arise from core outcomes are core imputations in the new game. But, conversely, we also show that every core imputation can be realized by some prices and allocation. We show an even stronger result: every core imputation can be realized by some prices and \textit{any efficient} allocation. This equivalence between core outcomes and core utility imputations allows us to reason about fairness and stability using the cooperative market game.

\begin{proposition}
\label{prop:core-imputation-implementation}
    Fix some market with fully substitutable agents. For any efficient set of trades~$\Phi$ and core imputation $\ub \in \R^I$, we can find prices~$\pb$ that implement $(\Phi, \pb)$ as a core outcome with utilities $u^i(\Phi, \pb) = u_i$.
\end{proposition}

The proof of \cref{prop:core-imputation-implementation}, which can be found in \Cref{proof:prop:core-imputation-implementation}, also gives us an efficient way to compute a core outcome that agrees with a given core imputation.

\begin{corollary}\label{corollary:imputation-to-outcome}
    For any market and core imputation $\xb$, a core outcome with utilities $\xb$ can be computed in polynomial time.
\end{corollary}
\subsection{Fairness Concepts}

We now define several notions of fairness for market outcomes. Common to all our fairness notions is that our redistribution concerns the utility of agents. This is a purely normative notion the regulator or market designer may care about. Fairness concerns regarding the allocation of trades or prices paid are not in the scope of our analysis.

A well-established notion for a fair payoff distribution is the Shapley value \citep{Shapley1953}. This notion is particularly natural in convex games, as the utility imputation as prescribed by the Shapley value always lies in the core of such games. The cooperative game defined on the trading network market, however, is not convex, as the associated welfare function need not be supermodular. We establish this in \cref{example:non-convex-game} below.%
\footnote{On the other hand, it is straightforward to show that the welfare function is superadditive, as the market value achievable by any two disjoint coalitions cannot exceed their union's value.}
Thus, we have to rely on different fairness notions.

\begin{example}
\label{example:non-convex-game}
A cooperative game $(I,w)$ with players $I$ and characteristic function $w$ is convex if $w$ is supermodular: $w(S \cup T) + w(S \cap T) \geq w(S) + w(T), \forall S, T \subseteq I$. Consider the market shown in \cref{fig:non-convex-game}. Social welfare is maximized if exactly one seller makes a sale to the buyer. Hence, the welfare function $w$ of this market assigns both coalitions $C = \{1,3\}$ and $D = \{2,3\}$ the same market value of $1$. Moreover, the grand coalition $\{1, 2, 3\}$ also gets value $1$, so
$w(C) + w(D) = 2 > 1 = w(C \cup D) + w(C \cap D)$
shows that $w$ is not supermodular.
\end{example}

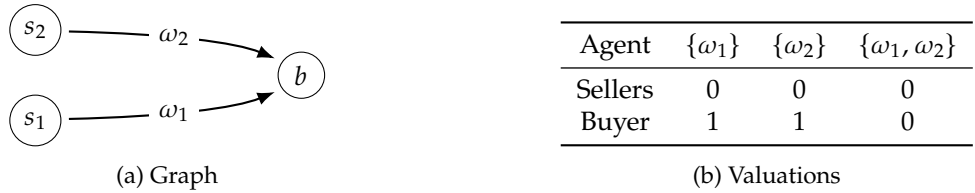
\begin{figure}[htp]
    \centering
    \begin{subfigure}[b]{0.49\textwidth}
    \centering
    \begin{tikzpicture}[xscale=3.5, yscale=0.6]
    \node[agent] (s1) at (0,0) {$s_1$};
    \node[agent] (s2) at (0,2) {$s_2$};
    \node[agent] (j) at (1,1) {$b$};
    \draw[trade] (s1) to[bend right] node[pos=0.4,fill=white] {$\omega_1$} (j) ;
    \draw[trade] (s2) to[bend left] node[pos=0.4,fill=white] {$\omega_2$} (j);
    \end{tikzpicture}
    \caption{Graph}
    \end{subfigure}
    \begin{subfigure}[b]{0.49\textwidth}
    \centering
    \begin{tabular}{cccc}
        \toprule
         Agent & $\{ \omega_1 \}$ & $\{ \omega_2 \}$ & $\{ \omega_1, \omega_2 \}$ \\
         \midrule
         Sellers    & $0$ & $0$ & $0$ \\
         Buyer      & $1$ & $1$ & $0$ \\
         \bottomrule
    \end{tabular}
    \caption{Valuations}
    \end{subfigure}
    \caption{An example of a market whose cooperative market game is not convex.}
    \label{fig:non-convex-game}
\end{figure}

As described above, while fairness considerations in isolation are trivial, they become meaningful in conjunction with efficiency and stability requirements. All core imputations are stable and efficient; however, by analogy with Shapley's distributive notion, not all agents contribute equally to the market value. We distinguish between agents who are indispensable for generating surplus and those who are not, leading to the definition of essential and inessential agents \citep{Vazirani-2022-Core-imputations}.

\begin{definition}[Essential agents]\label{def:essential-agents}
    An agent $i \in I$ is \emph{essential} if they are involved in every efficient set of trades, and \emph{inessential} otherwise.
\end{definition}

This distinction captures whether an agent, viewed alone, is structurally necessary for generating the maximal social welfare in the market. Any inessential agent, taken alone, could be removed from the market without affecting the market value. Naturally, this does not hold for a set of at least two inessential agents. 

\begin{observation}
\label{observation:essential-agents}
    The following statements are equivalent for any market $\market$:
    \begin{enumerate}
        \item Agent $i$ is essential.
        \item Removing agent $i$ from the market reduces the market value.
        \item Agent $i$ is allocated at least one trade in every competitive equilibrium.
    \end{enumerate}
\end{observation}

Within the set of core imputations, we explore three social welfare orderings and fairness metrics: leximin, leximax, and minimizing the variance of utilities. All three orderings are defined in terms of the utility profile of a core outcome, so it suffices by \cref{corollary:imputation-to-outcome} to find an optimal imputation.

We define the lexicographic order $\preceq_L$ on vectors as $\xb \preceq_L \yb$ if, at the smallest index $k$ for which $x_k \neq y_k$, we have $x_k < y_k$.

\begin{definition}
    Let $X$ be a set of $n$-dimensional vectors, let $\asc(\xb)$ and $\desc(\xb)$ denote the vector $\xb$ sorted in ascending and descending order, and let $\overline{x}$ be the mean of the entries of $\xb$.
    \begin{enumerate}
        \item A vector $\xb$ is a \textit{leximin} vector in $X$ if $\asc(\xb) \succeq_L \asc(\yb)$ for all $\yb \in X$.
        \item A vector $\xb$ is a \textit{leximax} vector in $X$ if $\desc(\xb) \preceq_L \desc(\yb)$ for all $\yb \in X$.
        \item A vector $\xb$ is a \textit{minvar} vector in $X$ if it minimizes the variance$ \frac{1}{n} \sum_{i \in [n]} (x^i - \overline{x})^2$.
    \end{enumerate}
\end{definition}

Leximin is an established social welfare ordering, also called the egalitarian social welfare ordering, that is independent of unconcerned agents, invariant to transforming utilities with an increasing bijection, and complies with the Pigou-Dalton transfer \citep{Moulin-2003}. Leximax is the symmetric sibling of leximin, which however only satisfies independence of unconcerned agents and invariance to increasing transformations. Minimizing the variance is also a central focus in the auction literature, e.g.~in \citet{Finster-2025-equitable-auctions} and \citet{mcafee2025winnerpaysbidauctionsminimizevariance}. As the minvar objective is strictly convex, it is immediate that the minvar core imputation is unique. \citet{vazirani2025fair} shows that the leximin and leximax core imputations of any cooperative game are also unique.%
\footnote{We note that multiple core outcomes, consisting of prices and allocations, can result in the same utility profile, so leximin, leximax, and minvar core outcomes may not be unique. Our focus on utilities abstracts away from this.}

It is easy to see that minimizing the variance of utilities among all core outcomes is equivalent to minimizing the sum of squared utilities, $\sum_{i \in I} x_i^2$, as all core outcomes are efficient.\footnote{The utilities of all agents sum to $w(I)$ and their mean is a constant $\overline{u} \coloneqq w(I) / |I|$ for any core outcome. Multiplying the variance of utilities of core outcome $(\xb, \pb)$ by the constant $|I|$, we see that
$\sum_{i \in I} (u^i(\xb, \pb) -  \overline{u})^2 = \sum_{i \in I} u^i(\xb, \pb)^2 - \sum_{i \in I} 2 u^i(\xb, \pb) \overline{u} + \sum_{i \in I} \overline{u}^2 = \sum_{i \in I} u^i(\xb, \pb)^2 - 2 \overline{u}^2 + |I| \overline{u}^2$, where the last two terms are constants.}
Moreover, in \cref{app:practical-optimisation}, we describe how to compute the leximin, leximax, and minvar core imputations by formulating optimization problems. These methods are not polynomial time, as they construct problems with exponentially many core imputation constraints. Some of the examples below were found using an implementation of these using Julia JuMP and Gurobi.

Finally, we note that, while the leximin core imputation is unique \citep{vazirani2025fair}, the associated core outcomes need not be unique. An example of a market with two leximin core outcomes is given in \cref{fig:leximin-not-unique} below.
\begin{observation}
    The leximin core outcome is not unique.
\end{observation}

\begin{example}\label{example:leximin-not-unique}
It is not hard to verify that all core outcomes are given by $\{(u^s, u^b, u^{t_1}, u^{t_2})| u^{t_1} = u^{t_2} = 0, u^s + u^b = 3, u^s \ge 0, u^b \ge 0\}$. Note that neither of the traders can have a positive utility in the core, as the seller and buyer would otherwise immediately form a coalition with the respective other trader. The leximin core imputation is given by $(\frac{3}{2}, \frac{3}{2}, 0, 0)$ and can be achieved by active trade set $\{\omega_1,\chi_1\}$ or $\{\omega_2,\chi_2\}$ at prices $p_{\omega_i} = p_{\chi_i} = \tfrac{1}{2}$.
\begin{figure}[htp]
    \centering
    \begin{subfigure}[b]{0.49\textwidth}
        \centering
        \begin{tikzpicture}[xscale=3, yscale=1.5]
            \node[agent] (s)  at (0,0)  {$s$};
            \node[agent] (b)  at (2,0)  {$b$};
            \node[agent] (t1) at (1,0.5)  {$t_1$};
            \node[agent] (t2) at (1,-0.5) {$t_2$};

            \draw[trade] (s)  -- (t1) node[midway,fill=white] {$\omega_1$};
            \draw[trade] (s)  -- (t2) node[midway,fill=white] {$\omega_2$};
            \draw[trade] (t1) -- (b)  node[midway,fill=white] {$\chi_1$};
            \draw[trade] (t2) -- (b)  node[midway,fill=white] {$\chi_2$};
        \end{tikzpicture}
        \caption{Trading network}
    \end{subfigure}
    \hfill
    \begin{subfigure}[b]{0.49\textwidth}
        \centering
        \begin{tabular}{ccccc}
            \toprule
            Agent & $\{\omega_i\}$ & $\{\chi_i\}$ & $\{\omega_i,\chi_i\}$ \\
            \midrule
            $t_1, t_2$ & $-\infty$ & $-\infty$ & $0$  \\
            $s$ & $1$ & -- & $-\infty$  \\
            $b$ & -- & $2$ & $-\infty$ \\
            \bottomrule
        \end{tabular}
        \caption{Valuations}
    \end{subfigure}
    \caption{Example with two pure traders, a seller $s$ and a buyer $b$. The valuations for the empty bundle are zero for each agent, valuations for bundles not listed in the table are $-\infty$.}
    \label{fig:leximin-not-unique}
\end{figure}
\end{example}

\subsection{The Impossibility of Fairness and Stability}\label{sec:impossibility-of-fairness-in-the-core}
In this section, we show several impossibilities in demanding our fairness notion in conjunction with stability. Our main result is to demonstrate that the stability notion of the core is, for non-convex games, too demanding to create fair outcomes: inessential agents get zero surplus in every core outcome. Moreover, even essential agents need not be guaranteed a positive surplus in markets with four or more essential agents.

We also discuss several examples and special cases. \Cref{example:leximin-cannot-be-CE} shows that the leximin core imputation need not be implementable as a competitive equilibrium. Thus, we only consider core stability (and the associated core outcome, by \cref{prop:core-imputation-implementation}). In the special case where we have no more than three essential agents (and any number of inessential agents), the leximin, leximax, and minvar core imputations coincide. Moreover, this core imputation gives all essential agents strictly positive utilities. This could be thought of as a minimum requirement on any ``fair'' social welfare order among efficient outcomes: if an essential agent has no incentive to stay in the market because they get zero surplus, the market value will decrease. \cref{example:zero-leximin}, demonstrates that in markets with four essential agents it is possible for an essential agent to get zero utility in every core outcome. Consequently, any fair welfare ordering, including leximin, leximax, and minvar, is rendered ineffective. In an additional example, \cref{example:leximin-not-equal-to-leximax}, we show that, with at least four essential agents, the leximin core imputation need not coincide with the leximax imputation. 

The following example shows that leximin core outcomes cannot be implemented as competitive equilibria. Thus, we consider exclusively the core as a stability criterion.
\begin{example}\label{example:leximin-cannot-be-CE}
In the trading network given in \cref{fig:leximin-cannot-be-CE} below, the welfare function is given by $w(\{1,2\}) = 1$ and $w(\Phi) = 0$ for all $\Phi \neq \{1,2\}$. It is immediate that the leximin core imputation is $(0.5, 0.5)$. Moreover, the unique competitive equilibrium prices are given by $p_\omega = 0$ and $p_\chi = 0$, and the efficient sets of trades are $\{\omega\}$ and $\{\chi\}$. It follows that the unique utility profile of either competitive equilibrium is $\ub = (0,1)$, i.e., the buyer receives the entire market value.
\begin{figure}[htp]
    \centering
    \begin{subfigure}[b]{0.49\textwidth}
        \centering
        \begin{tikzpicture}[xscale=3.5, yscale=1.8]
            \node[agent] (s) at (0,0) {$s$};
            \node[agent] (b) at (1,0) {$b$};
            \draw[trade] (s) to[bend left]  node[midway,fill=white] {$\omega$} (b);
            \draw[trade] (s) to[bend right] node[midway,fill=white] {$\chi$}   (b);
        \end{tikzpicture}
        \caption{Trading network}
    \end{subfigure}
    \hfill
    \begin{subfigure}[b]{0.49\textwidth}
        \centering
        \begin{tabular}{cccc}
            \toprule
            Agent & $\{\omega\}$ & $\{\chi\}$ & $\{\omega,\chi\}$ \\
            \midrule
            $s$ & $0$ & $0$ & $0$ \\
            $b$ & $1$ & $1$ & $0$ \\
            \bottomrule
        \end{tabular}
        \caption{Valuations}
    \end{subfigure}
    \caption{The leximin core imputation cannot be implemented as a competitive equilibrium.}
    \label{fig:leximin-cannot-be-CE}
\end{figure}
\end{example}

First, consider essential agents. The following two propositions and subsequent examples show a clear separation between markets with at most three, and at least four, agents.

\begin{proposition}\label{lem:pos-utility-3-agents}\label{prop:pos-utility-3-agents}
    In any market with at most three essential agents and a non-empty core, the leximin core imputation assigns strictly positive utility to all three essential agents.
\end{proposition}

The proof of the proposition is deferred to \cref{proof:lem:pos-utility-3-agents}. Markets with at most three essential agents also satisfy the following property.

\begin{proposition}\label{prop:leximin-leximax-minvar-coincide-3-agents}
    In any market with at most three essential agents and a non-empty core, the leximin, leximax, and minimum variance core imputations coincide.
\end{proposition}

The proof is deferred to \cref{proof:prop:leximin-leximax-minvar-coincide-3-agents}. Note that these results do not require substitutes agents. However, they break down when we have four or more essential agents in the market. We show this in the example below for \cref{prop:pos-utility-3-agents}. In \cref{example:leximin-not-equal-to-leximax} in \Cref{app:sec:additional-examples}, we demonstrate that the leximin and leximax core imputations need not coincide in markets with four or more agents.

\begin{example}\label{example:zero-leximin}
In the market given in this example all four agents are essential. We show that buyer $2$ gets zero utility in all core imputations.

The characteristic function is given in \cref{tab:characteristic-func-essential-gets-zero}. All agents are essential, as removing any agent strictly reduces the market value. The leximin core imputation for this market is $(3, 7, 5, 0)$, which implies that buyer 2 receives utility $0$ in all core imputations. We can also prove this directly. Let $u$ be an arbitrary core imputation. From the coalition constraints of sets of three agents together with the grand coalition constraint $u^{s_1} + u^{s_2} + u^{b_1} + u^{b_2} = 15$, it follows that
$u^{s_1} \le 3$, $u^{s_2} \le 7$, $u^{b_1} \le 5$, and $u^{b_2} \le 1$.
Together with the coalition constraints of sets of two agents, we obtain from these inequalities
$u^{s_1} \ge 3$, $u^{s_2} \ge 7$, $u^{b_1} \ge 5$.
So with the grand coalition constraint it follows that
$u^{b_2} = 0$.

\end{example}

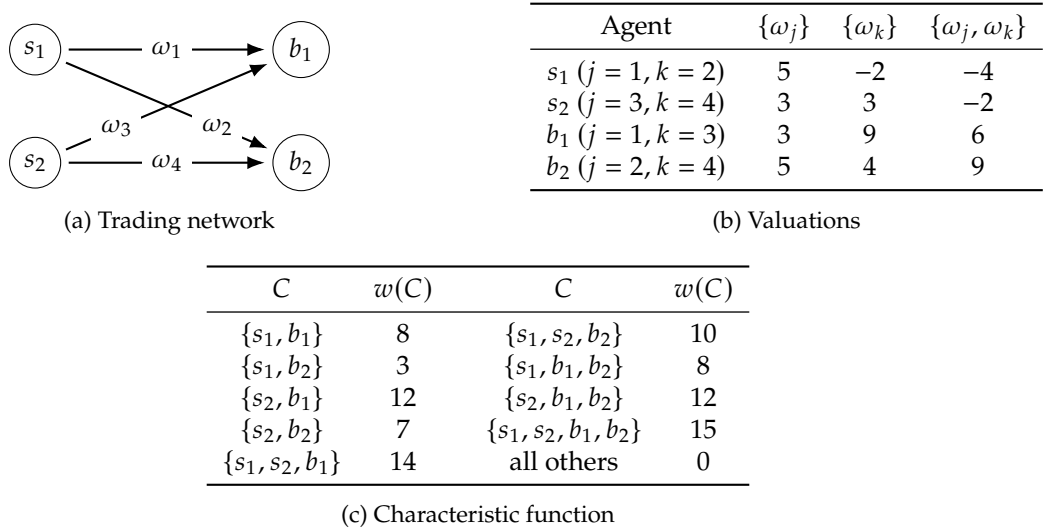
\begin{figure}[htp]
    \centering
    \begin{subfigure}[b]{0.49\textwidth}
        \centering
        \begin{tikzpicture}[xscale=3.5, yscale=1.5]
            \node[agent] (s1) at (0,1) {$s_1$};
            \node[agent] (s2) at (0,0) {$s_2$};
            \node[agent] (b1) at (1,1) {$b_1$};
            \node[agent] (b2) at (1,0) {$b_2$};

            \draw[trade] (s1) -- (b1) node[midway,fill=white] {$\omega_1$};
            \draw[trade] (s1) -- (b2) node[pos=0.75,fill=white] {$\omega_2$};
            \draw[trade] (s2) -- (b1) node[pos=0.25,fill=white] {$\omega_3$};
            \draw[trade] (s2) -- (b2) node[pos=0.5,fill=white]  {$\omega_4$};
        \end{tikzpicture}
        \caption{Trading network}
    \end{subfigure}
    \hfill
    \begin{subfigure}[b]{0.49\textwidth}
        \centering
        \begin{tabular}{cccc}
            \toprule
            Agent & $\{\omega_j\}$ & $\{\omega_k\}$ & $\{\omega_j,\omega_k\}$ \\
            \midrule
            $s_{1}$ $(j=1, k=2)$ & $5$ & $-2$ & $-4$ \\
            $s_{2}$ $(j=3, k=4)$ & $3$ & $3$  & $-2$ \\
            $b_{1}$ $(j=1, k=3)$& $3$ & $9$  & $6$  \\
            $b_{2}$ $(j=2, k=4)$ & $5$ & $4$  & $9$  \\
            \bottomrule
        \end{tabular}
        \caption{Valuations}
    \end{subfigure}

    \vspace{1em}

    \begin{subfigure}[b]{0.9\textwidth}
        \centering
        \begin{tabular}{c c@{\qquad}c c}
            \toprule
            $C$ & $w(C)$ & $C$ & $w(C)$ \\
            \midrule
            $\{s_1,b_1\}$ & $8$  & $\{s_1,s_2,b_2\}$ & $10$ \\
            $\{s_1,b_2\}$ & $3$  & $\{s_1,b_1,b_2\}$ & $8$  \\
            $\{s_2,b_1\}$ & $12$ & $\{s_2,b_1,b_2\}$ & $12$ \\
            $\{s_2,b_2\}$ & $7$  & $\{s_1,s_2,b_1,b_2\}$ & $15$ \\
            $\{s_1,s_2,b_1\}$ & $14$ & all others & $0$ \\
            \bottomrule
        \end{tabular}
        \caption{Characteristic function}
        \label{tab:characteristic-func-essential-gets-zero}
    \end{subfigure}

    \caption{Market with four essential agents. Each agent $i$'s valuation is given for $\Omega_i = \{ \omega_j, \omega_k \}$ with $j \leq k$.}
    \label{fig:one-essential-gets-zero-in-all-core}
\end{figure}

\Cref{example:zero-leximin} implies the following corollary.
\begin{corollary}
    There exist markets with four or more essential agents in which at least one essential agent gets zero utility from all core imputations.
\end{corollary}

In summary, leximin within the core does not guarantee that essential agents get positive utility. Thus, it may provide too weak incentives for essential agents to participate in the market, highlighting the tension between efficiency, stability, and fairness.

Finally, we present our main result in this section, demonstrating that even simple core stability, a much weaker notion than competitive equilibrium, may be too demanding in order to incentivize inessential agents to stay in the market. While any one inessential agent may leave the market without affecting the market value, several inessential agents leaving the market may negatively affect efficiency. \cref{theorem:inessential-agents-zero-utility} is a generalization of the lone wolf theorems from matching markets and transferable utility exchange economies, concerning the maximum weight matching and competitive equilibria. To our knowledge, it was previously unknown that the analogy holds for core outcomes also. We prove this result for the trading network market, which also nests very general exchange economies such as in \citet{Sun-Yang-2006} (see also \citet{hatfield2013stability}).

\begin{theorem}\label{theorem:inessential-agents-zero-utility}
    In every core outcome, inessential agents get utility zero.
\end{theorem}
\begin{proof}
    Let $\ub$ be the utility profile of some core outcome, and let $i \in I$ be an inessential agent. By definition, there is an efficient set $\Phi$ of trades in which agent $i$ is not involved. By \cref{prop:core-imputation-implementation}, there exist prices $\pb$ so that $(\Phi, \pb)$ is a core outcome with utility profile $\ub$, so $u^l(\Phi, \pb) = u_l$ for every agent $l \in I$. Agent $i$ has utility $u^i(\Phi, \pb) = 0$, as she is not involved in trades $\Phi$, so $u_i = 0$.
\end{proof}

Note here that we needed the strong result in \cref{prop:core-imputation-implementation} that we can extend \textit{any} efficient allocation $\Phi$ to a core outcome that achieves utility $\ub$.
\Cref{theorem:inessential-agents-zero-utility} cannot be turned into an iff statement: we can construct a 2-agent market with a single trade in which there is a stable outcome which assigns either the buyer or the seller utility $0$.

\begin{corollary}
\label{corollary:substitutes-market-essential-agents}
    If the core of a market with positive market value is not empty, then it must have at least one essential agent. In particular, this holds for all substitutes markets with positive market value.
\end{corollary}
\begin{proof}
The utilities of a core imputation add up to the market value so, if the market value is positive, at least one agent receives positive utility. Hence, that agent must be essential. When the agents in the market are substitutes, it admits at least one core outcome.
\end{proof}

Finally, we note that it is possible to guarantee a positive utility for all agents while still maintaining competitive equilibrium allocations and prices. However, the resulting utilities of such an outcome will not lie in the core. The way to implement this outcome with positive utilities for everyone, would be to simply collect a tax on every agent's utility at a rate of $\alpha \in [0,1]$, and pay out the collected taxes equally among all agents. Fix a market $\market$ and let $W(\Phi) = \sum_{i \in I} v^i(\Phi)$. Then the utility of agent~$i$ under the tax scheme with rate $\alpha$ is no longer quasi-linear but instead given by $\widehat{u}^i(\pb, \Phi) \coloneqq \frac{\alpha}{|I|}W(\Phi) + (1-\alpha) u^i(\pb, \Phi)$. As before, a competitive equilibrium is any market arrangement $(\pb, \Phi)$ that satisfies $\Phi_i \in \argmax_{\Psi \subseteq \Omega_i} \widehat{u}^i(\pb, \Psi)$ for each agent $i$.

\begin{observation}
Any competitive equilibrium in the original market $\market$ is a competitive equilibrium in the taxed market, for any $\alpha \in [0,1]$.
\end{observation}
\begin{proof}
Suppose $(\pb, \Phi)$ is a competitive equilibrium in the original setting, so $u^i(\pb, \Phi) \geq u^i(\pb, \Psi)$ holds for every agent~$i$ and all $\Psi \subseteq \Omega_i$. As $\Phi$ is an efficient set of trades, we also have $W(\Phi) \geq W(\Psi)$ for all $\Psi \subseteq \Omega$. This implies that $\widehat{u}^i(\pb, \Phi) \geq \widehat{u}^i(\pb, \Psi)$ for all $\Psi \subseteq \Omega_i$ and agents $i \in I$.
\end{proof}
The converse does not necessarily hold: not every competitive equilibrium in the taxed market is a competitive equilibrium in the original market setting. This is most easily seen when we set $\alpha = 1$: the market value is then divided equally among all agents. In that case, an outcome $(\pb, \Phi)$ is a competitive equilibrium in the taxed market, for any prices $\pb$, iff $\Phi$ is an efficient set of trades in the original market. It is also not hard to see that competitive equilibrium allocations in the taxed market are always efficient (see \cref{proposition:CE-taxed-market-is-efficient} in \Cref{sec:proofs-stability-fairness}).

\section{Conclusion}

Our article connects strategic interaction in the trading network game to competitive equilibrium outcomes and proves convergence of decentralized trading dynamics to Nash equilibria and approximate competitive equilibria. We find that core stability overwhelmingly restricts fairness considerations, and we substantially generalize existing lone wolf theorems. Future work may build on our methodology for the offer-based convergence proof to extend to $m$-sparse markets, and to design stability notions that better allow balancing stability, efficiency, and fairness.

\bibliographystyle{plainnat}
\bibliography{refs}

\appendix

\section{Proofs}\label{app:proofs}

\subsection{Proofs of \Cref{sec:network-trading-market}}\label{proofs-network-trading-market}

Fix some ordering $\Omega = \{\omega_1, \ldots, \omega_m\}$ on the trades in $\Omega$, and an agent $i \in I$. For any $\varepsilon \in (0,1)$, define the corresponding almost-integral valuation $\widehat{v}^i_\varepsilon$ as
\begin{equation}\label{eq:pert-valuation}
    \widehat{v}_\varepsilon^i(\Phi) \coloneqq v^i(\Phi) + \sum_{\omega_j \in \Phi} \frac{\varepsilon}{4^j}.
\end{equation}
where $v^i$ is the original valuation function of agent $i \in I$ in the market $\market$. Furthermore, let $\widehat{D}^i_\varepsilon$ be the demand correspondence of $\widehat{v}_\varepsilon^i$.
We show that the bundle $d^i(\pb)$ selected by the tie-breaking rule is uniquely demanded by valuation $\widehat{v}_\varepsilon$ at any $\varepsilon$-lattice prices $\pb \in \varepsilon \Z^\Omega$. For this, we first show that the LIP of $\widehat{v}^i_\varepsilon$ is obtained from the LIP of $v^i$ by a constant shift.

\begin{lemma}
\label{lemma:perturb-constant-shift}
The LIPs of $v^i$ and $\widehat{v}^i$ satisfy 
\[
    \LIP_{\widehat{v}^i} = \LIP_{v^i} + \sum_{\omega_j \in \Omega_i} \frac{\varepsilon}{4^j}\chi^i_{\omega_j} \eb^{\omega_j}.
\]
\end{lemma}
\begin{proof}
Fix $\pb \in \R^\Omega$ and $\Phi \subseteq \Omega_i$, and define $\qb = \pb + \sum_{\omega_j \in \Omega_i} \frac{\varepsilon}{4^j} \chi^i_{\omega_j} \eb^{\omega_j}$. Firstly, we see that
\begin{align*}
\widehat{u}^i(\Phi, \qb)
    &= \widehat{v}^i(\Phi) - \sum_{\omega \in \Phi} \chi^i_{\omega} q_\omega\\
    &= v^i(\Phi) + \sum_{\omega_j \in \Phi} \frac{\varepsilon}{4^j} - \sum_{\omega \in \Phi} \chi^i_\omega q_\omega \\
    &= v^i(\Phi) + \sum_{\omega_j \in \Phi} \frac{\varepsilon}{4^j} - \sum_{\omega_j \in \Phi} \chi^i_{\omega_j} (p_{\omega_j} + \frac{\varepsilon}{4^j} \chi^i_{\omega_j}) \\
    &= v^i(\Phi) + \sum_{\omega_j \in \Phi} \frac{\varepsilon}{4^j} - \sum_{\omega \in \Phi} \chi^i_\omega p_\omega - \sum_{\omega_j \in \Phi} \chi^i_{\omega_j} \chi^i_{\omega_j} \frac{\varepsilon}{4^j} \\
    &= v^i(\Phi) - \sum_{\omega \in \Phi} \chi^i_\omega p_\omega \\
    &= u^i(\Phi, \pb).
\end{align*}

As this holds for arbitrary $\Phi$ and $\pb$, it follows that $D^i(\pb) = \widehat{D}^i(\qb)$ (as the demanded bundles at given prices are the utility-maximizing bundles) for any $\pb$ and respective $\qb$. So the lemma statement holds by definition of the LIPs.
\end{proof}
\begin{lemma}\label{lem:tie_breaking_rule}
Let $\varepsilon \in (0,1)$. The valuation function $\widehat{v}_\varepsilon^i$ in \eqref{eq:pert-valuation} is fully substitutable if the original valuation $v$ is fully substitutable. Moreover, at all prices $\pb \in \varepsilon \Z^\Omega$, we have $\widehat{D}^i(\pb) = \{ d^i (\pb) \}$ and $d^i (\pb) \in D^i(\pb)$.
\end{lemma}

\begin{proof2}[Proof of \Cref{lem:tie_breaking_rule}]\phantomsection\label{proof:lem:tie_breaking_rule}
\Cref{lemma:perturb-constant-shift} shows us that the facets of the two LIPs of $v^i$ and $\widehat{v}_\varepsilon^i$ have the same orientation. From \cref{fact:substitutes-normals} in \cref{section:geometry-preferences}, we know that a valuation is substitutes iff its facets are normal to $\eb^{\omega}$ or $\eb^{\chi}-\eb^{\phi}$, so $\widehat{v}_\varepsilon^i$ is substitutes.

Now fix prices $\pb \in \varepsilon \Z^\Omega$. We show that demand is unique at $\pb$. Suppose that $\Phi$ and $\Psi$ are both in $\widehat{D}_\varepsilon^i(\pb)$. This means they both have the same utility at $\pb$, so 
\begin{align*}
    \widehat{v}_\varepsilon^i(\Phi) -\sum_{\omega \in \Phi}\chi_{\omega}^ip_{\omega} &= \widehat{v}_\varepsilon^i(\Psi) -\sum_{\omega \in \Psi}\chi_{\omega}^ip_{\omega} \\
    \Leftrightarrow \qquad \widehat{v}_\varepsilon^i(\Phi) - \widehat{v}_\varepsilon^i(\Psi) &= \sum_{\omega \in \Phi}\chi_{\omega}^ip_{\omega} -\sum_{\omega \in \Psi}\chi_{\omega}^ip_{\omega}.
\end{align*}
For any two bundles $\Phi$ and $\Psi$, the left-hand side is not in $\varepsilon \Z$, due to the way that the perturbation is defined. But the right-hand side is in $\varepsilon \Z$, a contradiction. Hence $\Phi$ and $\Psi$ cannot both be in $\widehat{D}^i_\varepsilon(\pb)$.

Lastly, observe that, for the integral valuation $v$, the difference between the highest and second highest utility achievable is at least $\varepsilon$. Thus, a perturbation by strictly less than $\frac{\varepsilon}{2}$, as implemented in the definition of $\widehat{v}$, cannot cause a previously suboptimal bundle to be preferred. This ensures that the unique bundle in $\widehat{D}_\varepsilon^i(\pb)$ lies in ${D}^i(\pb)$.
\end{proof2}

\subsection{Proofs of \Cref{sec:the-network-trading-game}}\label{sec:proofs-network-trading-game}

\begin{proof2}[Proof of \Cref{proposition:substitutes-implies-NE}]\phantomsection\label{proof:proposition:substitutes-implies-NE}
Suppose the market admits the competitive equilibrium $(\pb, \Phi)$. We construct an $\varepsilon$-tight Nash equilibrium as follows. For every trade $\omega \in \Phi$, let $\offer^{b(\omega)}_\omega = \offer^{s(\omega)}_\omega = p_\omega$. For every other trade $\omega'$, let $\offer^{b(\omega')}_{\omega'} = p_{\omega'} - \frac{\varepsilon}{2}$ and $\offer^{s(\omega')}_{\omega'} = p_{\omega'} + \frac{\varepsilon}{2}$. By the definition of competitive equilibrium, all agents trade exactly those trades that are utility maximizing for them at prices $\pb$. Trades $\omega \in \Phi$ are active, and thus neither sellers nor buyers profitably deviate. Trades $\omega' \in \Omega \setminus \Phi$ are inactive, and since observed offers $\sigma^{-i}$ are worse for both buyers and sellers, respectively, than competitive equilibrium prices, no unilateral deviation is possible. The Nash equilibrium is $\varepsilon$-tight by construction.

The second part follows from the result of \citet{hatfield2013stability} that a competitive equilibrium exists if all agents in the market have fully substitutable preferences.
\end{proof2}

\begin{proof2}[Proof of \Cref{prop:tight-NE-is-approx-CE}]\phantomsection\label{proof:prop:tight-NE-is-approx-CE}
    Fix an $\varepsilon$-tight Nash equilibrium $\offers$, let $(\pb, \Phi)$ be its corresponding outcome, and fix some agent $i \in I$. First, note that the tightness of $\offers$ and the definition of $\pb$ tells us that terms
    $\sum_{\omega \in \Phi_{i \rightarrow}} \offer^{-i}_\omega - \sum_{\omega \in \Phi_{i \leftarrow}} \offer^{-i}_\omega$
    and
    $\sum_{\omega \in \Phi_{i \rightarrow}} p_\omega - \sum_{\omega \in \Phi_{i \leftarrow}} p_\omega$
    differ by at most $\frac{1}{2}\Delta \varepsilon$. (Recall that $\offers^{-i}$ denotes the offers that agent $i$ faces.)
    So $u^i(\Phi, \pb)$ is at most $\frac{1}{2} \Delta \varepsilon$ less than $u^i(\Phi, \offers^{-i})$. Similarly, for any other set $\Psi \subseteq \Omega_i$, $u^i(\Psi, \pb)$ is at most $\frac{1}{2} \Delta \varepsilon$ more than $u^i(\Psi, \offers^{-i})$. As $\offers$ is a Nash equilibrium, we have $u^i(\Phi, \offers^{-i}) \geq u^i(\Psi, \offers^{-i})$. Combining everything, this implies $u^i(\Phi, \pb) \geq u^i(\Psi, \pb) - \Delta \varepsilon$. As the agent $i \in I$ was chosen arbitrarily, $(\Phi, \pb)$ is an $\Delta \varepsilon$-approximate competitive equilibrium.
\end{proof2}

\begin{proof2}[Proof of \cref{theorem:NE-to-CE}]\phantomsection\label{proof:theorem:NE-to-CE}
The first part of our proof proceeds like the proof of Theorem 6 in \citet{hatfield2013stability}. For completeness, we write it out here.

Suppose $\offers$ is an $\varepsilon$-tight Nash equilibrium for market $\market = (I, \Omega, v)$ with $\varepsilon < \frac{1}{2\Delta - 2}$. First, let $(\pb, \Phi)$ be the market outcome associated with offers $\offers$, so $\Phi$ is the set of active trades and $\pb \in \R^\Phi$ satisfies $p_\omega = \offer^{b(\omega)}_\omega = \offer^{s(\omega)}_\omega$ for each active trade $\omega \in \Phi$. Next, reduce $\market$ to a sub-market $\widetilde{\market} = (I, \widetilde{\Omega}, \widetilde{v})$ by removing all active trades of $\offers$, i.e., $\widetilde{\Omega} = \Omega \setminus \Phi$, and restricting the valuations to the remaining set $\widetilde{\Omega}$ of trades as follows. For each agent $i$, we construct the valuation $\widetilde{v}^i$ from the original valuation~$v^i$ as
\[
    \widetilde{v}^i(\Psi) \coloneqq \max_{ \Xi \subseteq \Phi_i} \left [ v^i(\Psi \cup \Xi) - \sum_{\omega \in \Xi} \chi^i_\omega \offer^{-i}_\omega \right ] ,
\]
for any $\Psi \subseteq \widetilde{\Omega}_i$. Intuitively, this modified valuation incorporates the maximum utility that the agent can achieve from combining $\Psi$ with any subset of active trades at the prices specified by the offers~$\offers$. \citet{hatfield2015full} show that the $\widetilde{v}^i$ constructed in this way are fully substitutable.

In \cref{proposition:empty-market-NE-CE}, we show that the reduced market $\widetilde{\market}$ admits a competitive equilibrium $(\widetilde{\pb}, \emptyset)$ that allocates every agent the empty bundle. This equilibrium specifies a price $\widetilde{p}_\omega$ for each inactive trade $\omega \in \widetilde{\Omega}$.

Now let $\qb \in \R^\Omega$ be defined by $q_\omega = p_\omega$ for each active trade $\omega \in \Phi$ and $q_\omega = \widetilde{p}_\omega$ for each inactive trade $\omega \in \widetilde{\Omega}$. Due to the construction of the reduced valuations $\widetilde{v}$ and the fact that each agent $i$ in market $\widetilde{\market}$ demands the empty bundle at $\widetilde{\pb}$, it is straightforward to see that agent $i$ in market $\market$ demands $\Phi_i$ at $\qb$. So $(\qb, \Phi)$ is a competitive equilibrium of the original market $\market$. 
\end{proof2}

\begin{proposition}
\label{proposition:empty-market-NE-CE}
Fix an $\varepsilon$-tight Nash equilibrium $\offers$ for a market $(I, v, \Omega)$ with fully substitutable valuations $v^i$, such that no trades are active. If $\varepsilon < \frac{1}{2\Delta-2}$, where $\Delta$ is the maximum vertex degree in the market, there exist integral competitive equilibrium prices $\pb^*$ at which all agents demand the empty bundle $\emptyset$. These prices are given by
\[
    p^*_\omega \in \left \{ \floor*{\offer^{s(\omega)}_\omega}, \ceil*{\offer^{b(\omega)}_\omega} \right \}.
\]    
\end{proposition}

\begin{proof}
For simplicity, assume that $\offer^{b(\omega)}_{\omega} = \offer^{s(\omega)}_{\omega} - \varepsilon$ for all trades $\omega \in \Omega$. This is without loss of generality, as we can reduce buying offers or increase selling offers, and still be in a Nash equilibrium where no trades are active.

For each agent $i \in I$, let $P^i$ be the set of all prices $\pb \in \R^{\Omega}$ at which agent $i$ demands $\emptyset$. As $\offers$ is a Nash equilibrium, agent $i$ demands $\emptyset$ at her counterparties' offers $\offers^{-i}$.
So, by construction, every price vector $\pb \in \R^\Omega$ that coincides with $\offers^{-i}$ on trades $\omega \in \Omega_i$ lies in $P^i$. Our goal is to show that the intersection of the sets $P^i$ for all agents $i$ is non-empty.

As agents have fully substitutable and integral valuations over bundles $\Phi \subseteq \Omega_i$, each set $P^i$ is a convex polyhedron that can be expressed as the intersection of integral half-spaces normal to $\eb^\omega$ or $\eb^{\chi} - \eb^{\varphi}$ for $\omega, \chi, \varphi \in \Omega$. Moreover, because the demand at prices in $P^i$ is the empty bundle $\emptyset$, $P^i$ can only arise from the following half-spaces:

\begin{itemize}
    \item $\{ \pb \mid \eb^\omega \cdot \pb \geq b_\omega \}$ for some $b_\omega \in \Z$, for every buying trade $\omega \in \Omega_{i \leftarrow}$,
    \item $\{ \pb \mid -\eb^\omega \cdot \pb \geq b_\omega \}$ for some $b_\omega \in \Z$ for every selling trade $\omega \in \Omega_{i \rightarrow}$,
    \item $\{ \pb \mid (\eb^\chi - \eb^\varphi) \cdot \pb \geq b_{\chi \varphi} \}$ for some $b_{\chi \varphi} \in \Z$, for every buying trade $\chi \in \Omega_{i \leftarrow}$ and selling trade $\varphi \in \Omega_{i \rightarrow}$.
\end{itemize}
Without loss of generality, we will assume that $P^i$ is the intersection of all these halfspaces, adding dummy hyperplanes containing $\offer^{-i}$ if necessary. (This only increases the difficulty of the problem of showing that all $P^i$ intersect.) 

We can write $P^i \coloneqq \{\pb \mid A^i \pb \geq \bb^i \}$ for appropriately chosen integral matrix $A^i$ and integral vector $\bb^i$. Note that $A^i$ consists of strong-substitutes rows, so it is unimodular (see \citet{baldwin2019understanding}).

Now define $\qb \in \R^\Omega$ by $q_\omega = \offer^{s(\omega)}_{\omega}$ for all $\omega \in \Omega$. Shift each $P^i$ by $\varepsilon \eb^{\Omega_{i \leftarrow}}$ to get
\[
Q^i \coloneqq P^i + \varepsilon \eb^{\Omega_{i \leftarrow}}
= \{\pb \mid A^i(\pb - \varepsilon \eb^{\Omega_{i \leftarrow}}) \geq \bb^i \}
= \{\pb \mid A^i \pb \geq \bb^i + A^i \varepsilon \eb^{\Omega_{i \leftarrow}} \}.
\]

First, we note that $\qb \in Q^i$ for all agents $i \in I$. Now write
\[
A = 
\begin{pmatrix}
    A^1 \\
    \vdots\\
    A^m
\end{pmatrix},
\bb =
\begin{pmatrix}
    \bb^1 \\
    \vdots \\
    \bb^m
\end{pmatrix},
\text{ and }
\cb =
\begin{pmatrix}
    A^1 \eb^{\Omega_{1 \leftarrow}}\\
    \vdots \\
    A^m \eb^{\Omega_{m \leftarrow}}
\end{pmatrix}.
\]
This allows us to write the intersection $Q$ of all $Q^i$ as
\[
Q \coloneqq \bigcap_{i \in I} Q^i = \{ \pb \mid A \pb \geq \bb + \varepsilon \cb \}.
\]

Moreover, $Q$ is a bounded convex polyhedron, as $A$ contains an upper and lower bound on each price entry $p_\omega$, $\omega \in \Omega$. So the intersection has at least one vertex. Fix such a vertex, $\vb$, and let $\widehat{A} \pb = \widehat{\bb} + \varepsilon \widehat{\cb}$ be a subsystem of $|\Omega|$ linearly independent rows of system $A \pb \geq \bb + \varepsilon \cb$ for which $\vb$ is the unique solution. Note that $\widehat{A}$ is a square matrix and invertible. So the vertex $\vb$ can be written as $v = \widehat{A}^{-1}(\widehat{\bb} + \varepsilon \widehat{\cb})$.

We next define $\pb^* = \widehat{A}^{-1} \widehat{\bb}$. This vector is integral, as $\widehat{A}$ is unimodular (see \cref{lemma:substitutes-matrix-unimodular} below) and $\widehat{\bb}$ is integral. Our goal now is to show that $\pb^*$ lies in $P^i$ for every $i \in I$. So fix any agent $i \in I$. We use the fact that $\vb \in Q^i$ to write
\begin{equation}
\label{eq:NE-to-CE-1}
A^i \pb^* = A^i \vb + A^i (\pb^* - \vb) \geq \bb^i + \varepsilon A^i \eb^{\Omega_{i \leftarrow}} + A^i (\pb^* - \vb).
\end{equation}
We can write out 
\[
A^i (\pb^* - \vb) = A^i( \widehat{A}^{-1} \widehat{\bb} - \widehat{A}^{-1}(\widehat{\bb} + \varepsilon \widehat{\cb})) = - \varepsilon A^i \widehat{A}^{-1} \widehat{\cb}.
\]
This allows us to simplify \cref{eq:NE-to-CE-1} to
\[
    A^i \pb^* \geq \bb^i + \varepsilon A^i \eb^{\Omega_{i \leftarrow}} + \varepsilon A^i \widehat{A}^{-1} \widehat{\cb} = \bb^i + \varepsilon A^i (\eb^{\Omega_{i \leftarrow}} - \widehat{A}^{-1} \widehat{\cb}).
\]

We now claim that all entries in vector $A^i (\eb^{\Omega_{i \leftarrow}} - \widehat{A}^{-1} \widehat{\cb})$ are greater than $2 - 2 \Delta > -\frac{1}{\varepsilon}$, so we have $A^i \pb^* > \bb^i - \eb^{\Omega}$, with strict inequality holding for each row. As $A^i, \pb^*$ and $\bb^i$ are integral, this claim thus implies that $A^i \pb^* \geq \bb^i$, and so $\pb^* \in P^i$.

To prove the claim, we take a closer look at the matrix $A$ and the vector $\cb$. Fix some trade $\omega \in \Omega$. By construction, $A$ contains the following rows with a non-zero entry in the column for trade $\omega$:
\begin{itemize}
    \item $\eb^\omega$ with corresponding entry $1$ in the same row of $\cb$,
    \item $-\eb^{\omega}$ with corresponding entry $0$ in the same row of $\cb$,
    \item $\eb^{\omega} - \eb^\chi$ with corresponding entry $1$ in the same row of $\cb$, for every $\chi \in \Omega_{b(\omega) \rightarrow}$,
    \item $\eb^{\chi} - \eb^{\omega}$ with corresponding entry $1$ in the same row of $\cb$, for every $\chi \in \Omega_{s(\omega) \leftarrow}$.
\end{itemize}

Recalling that $\Delta$ is the maximum vertex degree in the network, this implies that $(\widehat{A}^{-1} \cdot \widehat{c})_\omega$ lies in the interval $[1-\Delta, \Delta]$.

Finally, recall that $A^i$ only contains rows of type $\eb^\omega$ for buying trades $\omega \in \Omega_{i \leftarrow}$, type $-\eb^{\omega}$ for selling trades $\omega \in \Omega_{i \rightarrow}$, and type $\eb^\chi - \eb^{\varphi}$ for any combination of buying trade $\chi \in \Omega_{i \leftarrow}$ and selling trade $\varphi \in \Omega_{i \rightarrow}$.
The entry of $A^i (\eb^{\Omega_{i \leftarrow}} - \widehat{A}^{-1} \widehat{\cb})$ corresponding to a row of $A^i$ of type $\eb^\omega$ (with a buying trade $\omega$) is thus $(\widehat{A}^{-1} \widehat{\cb})_{\omega}$; the entry is $-(1 + \widehat{A}^{-1} \widehat{\cb})_{\omega}$ for any row of type $-\eb^\omega$ (with some selling trade $\omega$); and the entry is $(\widehat{A}^{-1} \widehat{\cb})_{\chi} - (1 + \widehat{A}^{-1} \widehat{\cb})_{\varphi}$ for any row $\eb^\chi - \eb^\varphi$ (with some buying trade $\chi$ and selling trade $\varphi$). This implies that every entry of $A^i (\eb^{\Omega_{i \leftarrow}} - \widehat{A}^{-1} \widehat{\cb})$ is greater or equal to $2 - 2\Delta$, which concludes the proof.
\end{proof}

\begin{lemma}
\label{lemma:substitutes-matrix-unimodular}
Let $A$ be a matrix in which every row contains at most one $+1$, at most one $-1$, and all other entries equal to zero. Then $A$ is totally unimodular; that is, every square submatrix of $A$ has determinant in $\{0,\pm 1\}$.
\end{lemma}

\begin{proof}
We proceed by induction on the size $k$ of a square submatrix $B$ of $A$. For $k=1$, the submatrix $B$ is a single entry of $A$, which is $0$, $+1$, or $-1$ by assumption.

Now let $k \geq 2$, let $B$ be a $k \times k$ submatrix of $A$, and suppose the result holds for all square submatrices of size less than $k$. Since $B$ inherits the row structure of $A$, every row of $B$ has at most one $+1$, at most one $-1$, and zeros elsewhere. We consider three cases.

If some column of $B$ is identically zero, then $\det(B)=0$, and we are done. If, instead, some column of $B$ contains exactly one nonzero entry, that entry is $\pm 1$. Expanding the determinant along that column gives $\det(B)=\pm \det(B')$, where $B'$ is a $(k-1)\times(k-1)$ submatrix of $A$. By the induction hypothesis, $\det(B')\in\{0,\pm 1\}$, and so $\det(B)\in\{0,\pm 1\}$.

Finally, suppose that every column of $B$ contains at least two nonzero entries. Count the total number of nonzero entries in $B$ in two ways. Each row of $B$ has at most two nonzero entries, giving a count of at most $2k$. Each of the $k$ columns has at least two nonzero entries, giving a count of at least $2k$. These bounds force equality, so every row of $B$ has exactly one $+1$ and one $-1$. In particular, every row of $B$ sums to zero. Thus the all-ones vector $\bm{1}$ satisfies $B \bm{1} = \bm{0}$, so the columns of $B$ are linearly dependent and $\det(B)=0$.
\end{proof}

\subsection{Proofs of \Cref{sec:convergence-of-decentralized-trading}}\label{app:convergence-offers-dynamic}

\begin{proof2}[Proof of \cref{prop:3-sparse-convergence}]\phantomsection\label{proof:prop:3-sparse-convergence}
We start by noting the fact that it is always possible to reduce a 3-sparse market to a market with two agents and three trades by considering a cut of the associated graph and merging the two subsets of agents obtained, while modifying the valuation function accordingly. For a detailed explanation of this process, we refer to \citealp[(Proposition 9)]{lock2024decentralized}. In the following, we directly focus on the case in which $I = \{1,2\}$ and $|\Omega|=3$.

Let $\varphi$ be the function that, for given offers $\offers$, computes the $L_1$-distance (or \textit{gap}) between each agent's current offers and their best response to their counterparts' offers. In other words, for each agent, it counts the number of trades for which they would change their offers if they best responded.

Formally, we write $BR^i(\offers^{-i})$ for the best response of agent $i$ to offers $\offers$, where $\offers^{-i}$ indicates the subvector of $\offers$ relative to agent $-i$. Define agent $i$'s \textit{gap} as
\[
    \varphi^i(\offers) \coloneqq \|\offers^i - BR^i(\offers^{-i}) \|_1
\]
and the potential function as
\begin{align}\label{eq:pot_fun_phi}
     \varphi(\offers) \coloneqq \sum_{i \in I} \varphi^i(\offers). 
\end{align}

It is clear that $0 \leq \varphi(\offers) \leq |I| \times |\Omega|$ for any $1$-tight offers $\offers$, and offers are $1$-tight as soon as each trade has seen at least one of its agents best responses. Note that when an agent best responds in the dynamics, they reduce their gap to $0$, but the other agent's gap can increase. \Cref{prop:phi_weakly_decresing} will bound this potential increase and characterize the cases in which the function remains instead constant.

Consider now two vectors $\offers$ and $\widehat{\offers}$ representing the offers respectively before and after a best responding step of agent $i$ as described in \Cref{alg:offer-based-dynamics}. Note that we are considering the tie breaking rule introduced in \Cref{sec:network-trading-market} to ensure that the demanded bundle is always well-defined and unique. Therefore, at every point it is possible to define a partition $\mathcal{P}$ over the set of trades $\Omega$, composed of the following sets 
\begin{align*}
    \Omega^{aa} = \{\omega \in \Omega \mid \omega \in d^i(\offers^{-i}) \text{ and } \omega \in d^{-i}(\offers^i)\},\\
    \Omega^{ar} = \{\omega \in \Omega \mid \omega \in d^i(\offers^{-i}) \text{ and } \omega \notin  d^{-i}(\offers^i)\},\\
    \Omega^{ra} = \{\omega \in \Omega \mid \omega \notin d^i(\offers^{-i}) \text{ and } \omega \in  d^{-i}(\offers^i)\},\\
    \Omega^{rr} = \{\omega \in \Omega \mid \omega \notin d^i(\offers^{-i}) \text{ and } \omega \notin  d^{-i}(\offers^i)\}\,,
\end{align*}
where, with some abuse of notation, we use $d^i(\offers^{-i})$ to indicate the demanded bundle computed at a vector of prices $\pb$ such that $p_{\omega}=\sigma^{-i}_{\omega}$ for all $\omega\in \Omega_i$.

We now consider the geometrical interpretation of the steps of the dynamic as movements in the space of the unique demand regions. Consider the line segment $L$ connecting $\offers$ and $\widehat{\offers}$. It might be that $L$ intersects some facet $F$ denoting the boundary between two areas in which agent $-i$ demands different bundles. The following lemma ensures that, if this happens, the two demanded bundles differ by at most two trades.
\begin{lemma}
\label{lem:well-behaved-segments}
The line segment $L$ only intersects the facets of agent $-i$'s LIP in their relative interior. Because of this, $L$ is divided into a finite number of sub-segments where the demanded bundle is unique and differs from the previous or following one, in the sense of the movement, by either having an additional trade, missing a trade, or both.
\end{lemma}

Furthermore, we are able to fully characterize the different ways in which the demanded bundle of $-i$ can change as a consequence of the crossing of a facet. 
\begin{lemma}
\label{lem:crossing_facets}
Suppose agent $i$ best responds and changes her offers from $\offers^i$ to $\hat{\offers}^i$. If a facet $F$ of agent $-i$'s LIP is crossed, then the bundle demanded by this agent evolves in one of the following ways.
\begin{itemize}
    \item[] Case 1: it can gain a trade initially in  $\Omega^{ar}$;
    \item[] Case 2: it can lose a trade initially in  $\Omega^{ra}$;
    \item[] Case 3: it can simultaneously gain two trades initially belonging to $\Omega^{ar}$;
    \item[]  Case 4: it can simultaneously drop two trades initially belonging to  $\Omega^{ra}$;
    \item[] Case 5: it can swap two trades from $\Omega^{aa}$ and $\Omega^{ar}$, losing the first one and gaining the second;
    \item[] Case 6: it can swap two trades belonging to $\Omega^{ra}$ and $\Omega^{ar}$;
    \item[] Case 7: it can swap two trades belonging to $\Omega^{ra}$ and $\Omega^{rr}$.
\end{itemize}
\end{lemma}
This precise description of the evolution of the demand allows us to quantify the number of trades for which agent $-i$ will want to modify her offers, thus computing the value of the potential function in $\widehat{\offers}$. Specifically, denoting by $c_j$ the number of crossed facets that cause a change of demand corresponding to Case $j$ in the lemma above, then the following holds.
\begin{lemma}\label{prop:phi_weakly_decresing}
    Let $\varphi$ be the potential function defined in \Cref{eq:pot_fun_phi}. Let $\offers$ and $\widehat{\offers}$ be the vectors of offers before and after a best responding step of the dynamic, as defined in \Cref{alg:offer-based-dynamics}, by agent $i$, then 
    \[
        \varphi(\widehat{\offers})-\varphi(\offers) \leq -c_1 -c_2 -2(c_3 + c_4+ c_6)
    \]
\end{lemma}
This lemma proves that the function $\varphi$ is non-increasing during the evolution of the market dynamic. To conclude the proof, we shall now prove that the function can not remain constant indefinitely. Reasoning by contradiction, let $(\offers_t)_{t\geq t_0}$ be a sequence of offers generated by \Cref{alg:offer-based-dynamics} for which $\varphi$ is stationary, meaning $\varphi(\offers^{t+1})-\varphi(\offers^t)=0, \ \forall t\geq t_0$. By \Cref{prop:phi_weakly_decresing}, we know that during none of the steps in this sequence it is possible to cross a facet that induces a change in demand corresponding to cases 1, 2, 3, 4, or 6. Therefore, the sequence $(\offers^t)_{t\geq t_0}$ must necessarily exhibit one among the following three behaviors.
\begin{itemize}
    \item[1] It performs a loop in an area where the demand bundles of both agents remain constant (no facet of any type is ever crossed).
    \item[2] It performs a loop during which some facets corresponding to cases 5 and 7 are crossed.
    \item[3] The sequence does not form any loop, and instead continues to evolve, reaching new states, while $\varphi$ remains constant.
\end{itemize}
However, the following lemma shows that none of these cases can occur.

\begin{lemma}\label{lem:no_constant_phi}
    None of the behaviors 1, 2, 3, detailed above, which would cause the function $\varphi$ to be stationary indefinitely, can happen in a market with three trades.
\end{lemma}
Having ruled out these cases, we can conclude that, throughout the dynamic, the potential function keeps decreasing, therefore, \Cref{alg:offer-based-dynamics} eventually converges.
\end{proof2}

\begin{proof2}[Proof of \Cref{lem:well-behaved-segments}]
We begin with the observation that the starting offers $\offers$ live in the $\varepsilon$-lattice $\varepsilon\mathbb{Z}^{2\Omega}$, by assumption. So, as we move along $L$, offers change by adding or subtracting the same amount, say $\delta$, to the values of a subset of trades. In other words, if $\pb \in L$, then there exists an $\delta >0$ such that $p_\omega \in \varepsilon\mathbb{Z}^{2\Omega}+ \{0, \delta, -\delta\}$ for all trades $\omega \in \Omega$. This also implies $p_\omega - p_\chi \in \varepsilon\mathbb{Z}^{2\Omega} + \{0, \pm \delta, \pm 2\delta \}$ for any two trades $\omega, \chi \in \Omega$.

Suppose, for contradiction, that $L$ intersects facets $F$ and $F'$ at point $\pb \in L$, and fix the $\delta$ for this $\pb$. We consider all the different possible orientations for $F$ and $F'$. Let $\nb$ and $\nb'$ be the normal vectors of $F$ and $F'$. Note that we cannot have $\nb = \nb'$. This leaves the following cases.

\begin{itemize}
    \item[\textbf{Case 1}] $\nb = \eb^{\omega_i}$ and $\nb' = \eb^{\omega_j}$. Then $\pb \in F$ implies $p_{\omega_i} \in \varepsilon\mathbb{Z}^{2\Omega}+ \frac{\varepsilon}{4^i}$, and so $\delta = \frac{\varepsilon}{4^i}$. Secondly, $\pb \in F'$ implies $p_{\omega_j} \in \varepsilon\mathbb{Z}^{2\Omega} + \frac{\varepsilon}{4^j}$ and so $\delta = \frac{\varepsilon}{4^j}$. But as $\frac{\varepsilon}{4^i} \neq \pm \frac{\varepsilon}{4^j}$ for $i \neq j$, this is a contradiction.
    \item[\textbf{Case 2}] $\nb = \eb^{\omega_i}$ and $\nb' = \eb^{\omega_i} - \eb^{\omega_j}$. Then $\pb \in F$ implies $p_{\omega_i} \in \varepsilon\mathbb{Z}^{2\Omega} + \frac{\varepsilon}{4^i}$, and so $\delta = \frac{\varepsilon}{4^i}$. Secondly, $\pb \in F'$ implies $p_{\omega_i} - p_{\omega_j} \in \varepsilon\mathbb{Z}^{2\Omega}+ \frac{\varepsilon}{4^i} - \frac{\varepsilon}{4^j}$, implying $p_{\omega_j} \in \varepsilon\mathbb{Z}^{2\Omega} + \frac{\varepsilon}{4^j}$. But this tells us that $\delta = \frac{\varepsilon}{4^j} = \frac{\varepsilon}{4^i}$, a contradiction.
    \item[\textbf{Case 3}] $\nb = \eb^{\omega_i}$ and $\nb' = \eb^{\omega_j} - \eb^{\omega_k}$. Then $\pb \in F$ implies $p_{\omega_i} \in \varepsilon\mathbb{Z}^{2\Omega} + \frac{\varepsilon}{4^i}$, so $\delta = \frac{\varepsilon}{4^i}$. Secondly, $\pb \in F'$ implies $p_{\omega_j} - p_{\omega_k} \in \varepsilon\mathbb{Z}^{2\Omega} + \frac{\varepsilon}{4^j} - \frac{\varepsilon}{4^k}$. But $\frac{\varepsilon}{4^j} - \frac{\varepsilon}{4^k} \notin \{0, \pm \delta, \pm 2\delta \}$, a contradiction.
    \item[\textbf{Case 4}] $\nb = \eb^{\omega_i} - \eb^{\omega_j}$ and $\nb' = \eb^{\omega_i} - \eb^{\omega_k}$. Then $\pb \in F$ implies $p_{\omega_i} - p_{\omega_j} \in \varepsilon\mathbb{Z}^{2\Omega} + \frac{\varepsilon}{4^i} - \frac{\varepsilon}{4^j}$, and $\pb \in F'$ implies $p_{\omega_i} - p_{\omega_k} \in \varepsilon\mathbb{Z}^{2\Omega} + \frac{\varepsilon}{4^i} - \frac{\varepsilon}{4^k}$. But $\frac{\varepsilon}{4^i} - \frac{\varepsilon}{4^j}$ and $\frac{\varepsilon}{4^i} - \frac{\varepsilon}{4^k}$ cannot both lie in $\{0, \pm \delta, \pm 2 \delta\}$, a contradiction.
    \item[\textbf{Case 5}] $\nb = \eb^{\omega_i} - \eb^{\omega_j}$ and $\nb' = \eb^{\omega_k} - \eb^{\omega_l}$. Then $\pb \in F$ implies $p_{\omega_i} - p_{\omega_j} \in \varepsilon\mathbb{Z}^{2\Omega} + \frac{\varepsilon}{4^i} - \frac{\varepsilon}{4^j}$, and $\pb \in F'$ implies $p_{\omega_k} - p_{\omega_l} \in \varepsilon\mathbb{Z}^{2\Omega} + \frac{\varepsilon}{4^k} - \frac{\varepsilon}{4^l}$. But $\frac{\varepsilon}{4^i} - \frac{\varepsilon}{4^j}$ and $\frac{\varepsilon}{4^k} - \frac{\varepsilon}{4^l}$ cannot both lie in $\{0, \pm \delta, \pm 2 \delta\}$, a contradiction.
\end{itemize}
\end{proof2}

\begin{proof2}[Proof of \Cref{lem:crossing_facets}] 
Due to the characterization of the demand areas in \cite{baldwin2019understanding} and as detailed in \Cref{section:geometry-preferences}, we know that the delimiting facets $F$ are orthogonal to a vector of the form either $\eb^{\omega}$ or $\eb^{\chi} - \eb^{\phi}$, since the preferences of the agents are fully-substitutable. We study these two cases separately.

Note that, if $F \perp\eb^{\omega}$, it is not possible for $\omega$ to belong to $\Omega^{aa}$ or $\Omega^{rr}$ since the agent does not modifies her offers for such trades.
This implies that in the vector $L$ the coordinate relative to $\omega$ is null, hence it would not be possible to cross a facet in these directions. 
This leaves only two other possibilities corresponding to the following two cases:
\begin{itemize}
    \item[\textbf{Case 1}] If $\omega \in \Omega^{ar}$ agent $i$ modifies her offers in favor of agent $-i$. Therefore if a facet is crossed and there is a change in demand, it implies the bundle demanded by agent $-i$ gains $\omega$. 
    \item[\textbf{Case 2}] If $\omega \in \Omega^{ra}$ agent $i$ instead modifies offers in her favor, away from the ones of agent $-i$. Thus, if a facet is crossed, it implies the bundle demanded by agent $-i$ drops $\omega$.
\end{itemize}

The remaining cases correspond to the possible crossing of facets when $ F \perp \eb^{\chi}-\eb^{\phi}$. For the same motivation as before, it is not possible for both trades $\chi$ and $\phi$ to belong to $\Omega^{aa}\cup \Omega^{rr}$, hence we are left with the following possibilities.
\begin{itemize}
    \item[\textbf{Case 3}] if $\chi, \phi \in \Omega^{ar}$ a facet like $F$ can be crossed in the case in which agent $i$ is a seller for $\chi$ and a buyer for $\phi$. In this case, $F$ is indicating the border between an area where agent $-i$ demands both trades and one where it refuses both and it is possible to cross $F$ when the direction of the update is, as in this case, $(-1,1)$. 
    \item[\textbf{Case 4}] Following an analogous argument as above, if $\chi, \phi \in \Omega^{ra}$ a crossing is possible again when agent $i$ is a buyer for $\chi$ and a seller for $\phi$, moving in the opposite direction of the one described in Case 3, resulting in the demand bundle of agent $-i$ to lose both trades.\footnote{Note that this, and the following, case, can only occur when the agent has two different roles, buyer and seller, for the two trades considered. Otherwise, $F$ could not represent the facet that indicates the change in demand.}
    \item[\textbf{Case 5}] If one trade between $\chi$ and $\phi$ belongs to $\Omega^{aa}$, the other to $\Omega^{ar}$ and $i$ is a seller for both, then, due to the lowering of the offer of the trade in $\Omega^{ar}$ it is possible to cross a $F$, which will result in a swap of the two trades.
    \item[\textbf{Case 6}] When the two trades belong to the sets $\Omega^{ar}$ and $\Omega^{ra}$ and again agent $i$ is a buyer for both the resulting geometry makes the crossing possible, resulting in a swapping of the two trades in the demand bundle of $-i$.
    \item[\textbf{Case 7}] If the trades are in the sets $\Omega^{ra}, \Omega^{rr}$, a best response movement will cross a facet $F$ for instance in the case agent $i$ is a seller for both trades. Therefore when best responding she will increase the offer for the trade in $\Omega^{ra}$ leaving the other untouched. This could lead agent $-i$ to stop demanding the increases trade in favor of a previously rejected one, thus performing a swap.
\end{itemize}
The two remaining cases are when the two trades belong to the sets $\Omega^{aa}, \Omega^{ra}$ or $\Omega^{rr}, \Omega^{ar}$. These cases are ruled out from the fact that  a vector like $\eb^{\chi} - \eb^{\phi}$ cannot represent the opposite of the change of demand between the two areas, hence it is not possible for a facet $F$ perpendicular to it to delimitate the boundary between the two areas.
\end{proof2}

\begin{proof2}[Proof of \Cref{prop:phi_weakly_decresing}]
Assume the dynamic to have reached a state in which the vector of offers corresponds to $\offers$. Furthermore, denote by $i$ the agent who is chosen to best respond at this stage and by $\widehat{\offers}$ the vector reached after the update.
We start by computing the value of the potential function in $\offers$.
Note that, due to the alternation of the agents, when this state is reached, agent $-i$ just modified her offers, bringing the $L^1$ gap to zero, therefore $\varphi^{-i}(\offers)=0$.
On the other side, due to the update rule described in \Cref{alg:offer-based-dynamics}, we know that between $\offers$ and $\widehat{\offers}$, agent $i$ will only modify the offers corresponding to the trades in $\Omega^{ar}$ and $\Omega^{ra}$.
Hence, the value of the potential function is
\begin{align}\label{eq:phi_initial}
    \varphi(\offers) = \varphi^i(\offers)  = \vert\Omega^{ar}\vert + \vert \Omega^{ra}\vert.
\end{align}

After the response, when the offers correspond to $\widehat{\offers}$, agent $i$ has reduced her gaps to zero, which implies that the value of the potential function is completely determined by $\varphi^{-i}(\widehat{\offers})$. By definition, this corresponds to the number of trades for which she wishes to modify her offers given the counteroffers $\widehat{\offers}^i$.
Specifically, those are the trades $\omega$ that are in one of the following situations
\begin{itemize}
    \item[] $\omega \in \Omega^{aa}$ and $\omega \notin d^{-i}(\widehat{\offers}^{i})$;
    \item[] $\omega \in \Omega^{ar}$ and $\omega \notin d^{-i}(\widehat{\offers}^{i})$;
    \item[] $\omega \in \Omega^{ra}$ and $\omega \in d^{-i}(\widehat{\offers}^{i})$;
    \item[] $\omega \in \Omega^{rr}$ and $\omega \in d^{-i}(\widehat{\offers}^{i})$.
\end{itemize}

Due to \Cref{lem:crossing_facets}, we can quantify the number of trades that exhibit each of these four behaviors. 
\begin{itemize}
    \item The number of trades in $\Omega^{aa}$ that agent $-i$ stops demanding corresponds to the number of facets crossed that lead to a change of demand described by Case 5 in the lemma. We denote this number by $c_5$.
    \item The number of trades in $\Omega^{ar}$ that agent $-i$ keeps rejecting are given by the total minus the ones that, due to the crossing of a facet, from $\Omega^{ar}$ end up in either $\Omega^{aa}$ or $\Omega^{ra}$. Therefore, the number of $\omega$ in this case is 
    \[
    |\Omega^{ar}| - c_1 - 2c_3 - c_5 -c_6 
    \]
    \item Similarly, the trades in $\Omega^{ra}$ that agent $-i$ keeps demanding can be recover knowing the ones that instead from $\Omega^{ra}$ end up into $\Omega^{aa}$ or $\Omega^{ar}$. Therefore the number of trades $\omega$ for this case is
    \[
    \vert \Omega^{ra}\vert -c_2 -2c_4 -c_6 -c_7
    \]
    \item Lastly, the trades that were in  $\Omega^{rr}$ and then entered the demand bundle of $-i$ are a total of $c_7$
\end{itemize}
Due to this analysis, summing all the cases we can obtain the value of $\varphi^{-i}(\widehat{\offers})$, and therefore of the complete function as
\[
\varphi(\widehat{\offers})= \varphi^{-i}(\widehat{\offers})=\vert \Omega^{ar}\vert + \vert \Omega^{ra}\vert - c_1 -c_2 -2(c_3 + c_4 + c_6)
\]
combining this with \Cref{eq:phi_initial} leads the result.
\end{proof2}

\begin{proof2}[Proof of \Cref{lem:no_constant_phi}]
As shown by \Cref{prop:phi_weakly_decresing}, the function $\varphi$ is stationary between two consecutive steps only if no facet which results in a change of demand corresponding to Cases  1, 2, 3, 4, 6 is crossed. And, to have a stationary dynamic, this behavior has to repeat endlessly. 
Let $(\offers^t)_{t\in\mathbb{N}}$ be the sequence of offers given by the evolution of the market dynamic. Initially, assume that from a certain time step $t_0$ onward, it performs a loop. Furthermore denote by $c^t_i$ the number of facets of type $i$ crossed while going from $\offers^t$ to $\offers^{t+1}$.

An instance that would yield a stationary potential function is when $c^t_i=0, \ \forall i=1, \ldots, 7, \ \forall t \geq t_0$. 
This corresponds to the case in which the dynamic never crosses any facet during its loop.
By consequence, the demanded bundles of the two agents together with the composition of the four sets of the partition $\mathcal{P}$, remains constant for all $t\geq t_0$. 
Hence, at every step, agent $i$ will continue to match the proposed counter offer for the trades belonging to $\Omega^{ar}$ and shift away her offers for the ones in $\Omega^{ra}$. 
Since agent $-i$ is responding with an opposite behavior and no trade ever reaches a state of being demanded or refused by both agent, as would imply the crossing a facet, the value of the offers will diverge to $\pm \infty$, which is against the assumption of the offers being bounded.

Alternatively, the stationary behavior of $\varphi$ could be caused by the case where $c^t_i=0, \ \text{for } i=1, 2, 3, 4, 6$ , $\ \forall t\geq t_0$,  while  there exists some $t$ such that either $c^t_5,\ c^t_7\ne 0$. 
This means there exists at least a trade $\omega$ that is involved in a swap corresponding to a change of demand as the ones detailed in either Case 5 or 7. Note that for such trade the following holds.
\begin{lemma}\label{lem:unique_loop}
    Assume that the sequence $(\offers^t)_{t\geq t_0}$ performs a loop of length $k$ during which the demand bundle for one of the two agents changes according to either Case 5 or 7 of \Cref{lem:crossing_facets}, and let $\omega$ be one of the trades involved in one of the swaps. Then, throughout the loop, there exists some instants of time $t_1, t_2, t_3, t_4$ during which trade $\omega$ belongs to each of the elements of the partition $\mathcal{P}$.
\end{lemma}
The above lemma is a powerful characterization of the evolution of $\omega$ within the partition $\mathcal{P}$. Furthermore, it also allows us to conclude that these type of loops cannot occur in a market where there are only three trades. We can deduce this thanks to the fact that $\omega$ has to cross all four sets of the partition throughout the loop, and each change in the composition of the sets in $\mathcal{P}$ that leaves $\varphi$ can happen only via an exchange of two trades, as highlighted by Lemmas \ref{prop:phi_weakly_decresing} and \ref{lem:well-behaved-segments}. This implies that for such loop to happen at any time each element of the partition must contain at least one element, which is not possible if the market counts less than four trades.

The two aforementioned instances cover all the possible cases where the potential function remains stable while the dynamic performs a loop.
Hence  we only have left to analyze the case in which the sequence $(\offers^t)_{t\in\mathbb{N}}$ does not create a loop and simultaneously does not converge.
The impossibility of this event simply follows from the fact that since by assumption the offers are bounded and the agents are modifying each of them by either 0 or $\varepsilon$, the evolution of the dynamic happens over a finite grid.
Hence either the sequence converges or eventually reaches an already visited point and, given that the evolution at every step is fully dictated by the value of the vector $\offers$, it will start repeating itself.
\end{proof2}

\begin{proof2}[Proof of \Cref{lem:unique_loop}]
We observe that since by assumption the sequence $(\offers^t)_{t\geq t_0}$ creates a loop of length $k$, the offers for each trade and both agents have to coincide every $k$ steps; Formally  $\sigma_{\omega}^{i, t}  -   \sigma_{\omega}^{i, t +k} = 0$, where these represent the components of $\offers^{t}$ and $\offers^{t+k}$ relative to the trade $\omega$ for agent $i$. 
Using a telescopic sum it is possible to decompose the difference as 
\begin{align*}
  0 = \sigma^{i,t}_{\omega}  -  \sigma^{i,t+k}_{\omega} 
  = \sum_{s=0}^{k-1} \sigma^{i,t+s}_{\omega}  -  \sigma^{i,t+s+1}_{\omega} \,,
\end{align*}
and furthermore rearranging this in the sum of the differences for the instants of times when $\omega$ belongs to each of the sets of the partition $ \mathcal{P}$, it holds
\begin{align*}
     \sum_{s=0}^{k-1} \sigma^i_{\omega, t + s}  -  \sigma^i_{\omega, t +s + 1}
  =  &\sum_{s: \omega \in \Omega^{aa}} \sigma^{i,t+s}_{\omega}  -  \sigma^{i,t+s+1}_{\omega}
    +\sum_{s: \omega \in \Omega^{ar}} \sigma^{i,t+s}_{\omega}  -  \sigma^{i,t+s+1}_{\omega}\\[1em]
    &+\sum_{s: \omega \in \Omega^{ra}} \sigma^{i,t+s}_{\omega}  -  \sigma^{i,t+s+1}_{\omega}
    +\sum_{s: \omega \in \Omega^{rr}} \sigma^{i,t+s}_{\omega}  -  \sigma^{i,t+s+1}_{\omega}\,.
\end{align*}
Due to the definition of the dynamic, note that the offers are not modified when the trade belongs to either $\Omega^{aa}$ or $\Omega^{rr}$, making the two respective sums in the above expression equal to 0. 
Otherwise, the offer is modified in favor of the moving agent, thus subtracting $\chi_{\omega}^i \varepsilon$ at each step, when $\omega\in\Omega^{ra}$.
Therefore, denoting by $$\ell_{\Omega^{ra}}^{\omega}= \big \vert \{s\in[0,k-1] \ : \omega\in\Omega^{ra} \ \text{at time} \ t+s\,, \  \text{when is agent $i$'s turn to best respond}\}\big \vert $$ the number of time steps that $\omega$ spends in $\Omega^{ra}$, it follows
\[
\sum_{s: \omega \in \Omega^{ra}} \sigma^{i,t+s}_{\omega}  -  \sigma^{i,t+s+1}_{\omega} = \chi_{\omega}^i\varepsilon \ell_{\Omega^{ra}}^{\omega}\,.
\]
This holds since the difference in offers when agent $i$ is not modifying them is obviously 0. During her turn instead, she modifies them starting from the counterparts' ones that, during the previous round, by definition, have been matched to hers, therefore the shift performed during round $t+s+1$ can be seen as a shift of $\pm \chi_{\omega}^i\varepsilon$ with respect to agent $i$'s previous offers.
An equivalent reasoning is possible when instead $\omega \in  \Omega^{ar}$ and instead the offer is modified in favor of the opposite agent of a quantity each time equivalent to $\chi_{\omega}^i\varepsilon$.
Replacing this in the initial expression gives 
\begin{align}
      0= \sum_{s=0}^{k-1} \sigma^{i,t+s}_{\omega}  -  \sigma^{i,t+s+1}_{\omega} 
      =- \chi_{\omega}^i\varepsilon \ell_{\Omega^{ar}}^{\omega} + \chi_{\omega}^i\varepsilon \ell_{\Omega^{ra}}^{\omega} = \chi_{\omega}^i\varepsilon \left( \ell_{\Omega^{ra}}^{\omega} - \ell_{\Omega^{ar}}^{\omega}\right)
\end{align}
Note that this holds for all the trades, including the ones that during the loop are involved in a swap of the kind described in Case 5 or 7.
Let $\omega$ be one of those, whose existence is guaranteed by assumption, then it is not possible for both $\ell_{\Omega^{ra}}^{\omega}$ and $\ell_{\Omega^{ar}}^{\omega}$ to be equal to zero, since in order to be involved in a swap the trade has to belong to one among $\Omega^{ar}$ or $\Omega^{ra}$ either before or after the exchange.
Furthermore, comparing the first and last element of the equality, we can conclude $ \ell_{\Omega^{ra}}=  \ell_{\Omega^{ar}}$ and $ \ell_{\Omega^{ra}},  \ell_{\Omega^{ar}}>0$. This is equivalent to say that there are instants of time throughout the loop when trade $\omega$ belongs to $\Omega^{ar}$ or $\Omega^{ra}$.
Finally, note that, if this holds, then necessarily also $\ell_{\Omega^{aa}}$ and $ \ell_{\Omega^{rr}}$ are strictly positive.
In fact, knowing that $ \ell_{\Omega^{ra}},  \ell_{\Omega^{ar}}>0$, we can conclude that the trade has undergone both a swap corresponding to Case 5 and Case 7, and therefore by the way those are defined, at some point of the loop $\omega$ has belonged to both $\Omega^{aa}$ and $\Omega^{rr}$.
\end{proof2}

\subsection{Proofs of \Cref{sec:clock-market}}\label{app:proofs-price-dynamic}
 
\begin{proof2}[Proof of \Cref{prop:price-dynamic-converges}]\phantomsection\label{proof:prop:price-dynamic-converges}
The proof works by showing that the price-based market dynamic is equivalent to optimizing a specific Lyapunov function through stochastic subgradient descent (SSD). To this end, consider $L$ the function

\begin{equation}
\label{eq:lyapunov}
L(\pb) \coloneqq \sum_{i \in I} \max_{\Phi \subseteq \Omega_i} \left [ v^i(\Phi) - \sum_{\omega \in \Phi} \chi^i_\omega p_\omega \right ]\,,
\end{equation}
which represents the sum of agents' utilities at a given price vector $\pb$. Note that it is possible to decompose such a function as  $L(\pb) = \sum_{i \in I} L^i(\pb)$, where each $L^i$ is defined as 
\[
L^i(\pb) \coloneqq \max_{\Phi \subseteq \Omega_i} \left [ v^i(\Phi) - \sum_{\omega \in \Phi} \chi^i_\omega p_\omega \right ] + \langle\boldsymbol{\chi}^i, \pb\rangle\,,
\]
$\boldsymbol{\chi}^i$ is a $\vert \Omega\vert$-dimensional vector whose components correspond to $\chi^i_{\omega}$ for each trade $\omega$, and $\langle \cdot, \cdot\rangle$ indicates the inner product.

Let $\pb^t$ be the vector of prices at time $t$ and assume that at time $t+1$ agent $i$ is chosen to best respond. As shown in \Cref{sec:network-trading-market}, it is possible to enforce a tie-breaking rule, such that the demanded bundle is unique for each $\pb^t$. From this, the update direction of agent $i$ is given by $\Delta^i \coloneqq (\eb^{\Phi} - \boldsymbol{1})\times\boldsymbol{\chi}^i$, where $\boldsymbol{1}$ is the unitary vector of dimension $\vert\Omega\vert$ and $\Phi$ is the demanded bundle at price $\pb^t$. It is therefore possible to rewrite the price update rule expressed in \Cref{eq:price_update} as 
\[
\pb^{t+1} = \pb^t+ \varepsilon \Delta^i\,.
\]

The following lemma highlights the connection between $\Delta^i$ and the functions $L^i$.
\begin{lemma}\label{lem:Delta_i_subgradient}
The update direction $\Delta^i$ of agent $i$ given current prices $\pb^t$ corresponds to the negative subgradient of $L^i$.
\end{lemma}
Thanks to this lemma it is possible to cast the market dynamic of \Cref{alg:price-based-dynamics} as a stochastic subgradient descent routine as detailed in \Cref{alg:SSD_dynamic}.

\begin{algorithm}[tb!]
\begin{algorithmic}[1]
\Require Function $L$, step size $\varepsilon$, bound on valuations $R$.
\State Initialize starting vector of prices $\pb^0 \in [-R, R]^{\vert\Omega\vert}$ arbitrarily.
\For{$t \in \{1, \ldots, T\}$}
    \State Draw an index $i \in I$ uniformly at random.
    \State Let $-\Delta^i$ be the subgradient of the function $L^i$ in $\pb^t$.
    \State Update current price $\pb^{t+1} = \pb^t + \varepsilon\Delta^i$.
\EndFor
\State \Return $\pb^T$.
\end{algorithmic}
\label{alg:SSD_dynamic}
\caption{Stochastic Subgradient Descent}
\end{algorithm}

Finally, note that reaching the minimum of the function $L$ is equivalent to reaching an equilibrium in the original market, since the following holds.

\begin{lemma}\label{lem:min_lyapunov}
$L$ is minimized at competitive equilibrium prices $\pb^*$, at which $L(\pb^*) = w(I)$. At any non-equilibrium prices $\pb$, we have $L(\pb) > w(I)$.
\end{lemma}

Once the equivalence between the two methods is established, we can prove the convergence using standard results from convex optimization \citep{garrigos2023handbook}.

Thus, let $\pb^*$ be an equilibrium price. We can assume that $\pb^*\in [-R, R]^m$. Note that this bound holds trivially for the components of the vector associated to active trades and the others can be chosen in the same interval without loss of generality.
It is possible to decompose the difference between $\pb^{t+1}$ and $\pb^*$ as 
\begin{align*}
    \Vert \pb^{t+1} - \pb^*\Vert^2 
    &= \Vert\pb^t -\pb^* + \varepsilon\Delta^{i_t}\Vert^2\\
    &= \Vert\pb^t -\pb^*\Vert^2 +2\varepsilon \langle \Delta^{i_t}, \pb^t -\pb^*\rangle + \varepsilon^2\Vert\Delta^{i_t}\Vert^2\,.
\end{align*}
Note that, by definition, the components of the vector $\Delta^{i_t}$ are bounded by 1, hence  $\Vert\Delta^{i_t}\Vert^2 \leq \Delta$, where the latter corresponds to the maximum vertex degree in the market network. Moreover, since, given $\pb^t$, the index $i$ is chosen uniformly at random among $n$ possible ones, it holds that 
\[
n \bbE[-\Delta^{i_t}\vert \pb^t]= \sum_{i=1}^n-\Delta^i\,.
\]
It is straightforward to see that this corresponds to a subgradient of the function $L$ computed in $\pb^t$.

Therefore, taking the conditional expectation with respect to $\pb^t$ in the expression above leads to 
\begin{align*}
    \bbE[  \Vert \pb^{t+1} - \pb^*\Vert^2 \vert \pb^t]
    &= \Vert\pb^t -\pb^*\Vert^2 -2 \varepsilon\langle \bbE[-\Delta^{i_t}\vert \pb^t], \pb^t -\pb^*\rangle + \varepsilon^2\bbE[\Vert\Delta^{i_t}\Vert^2\vert \pb^t]\\
    &\leq \Vert\pb^t -\pb^*\Vert^2 -\frac{2\varepsilon}{n}\langle n\bbE[-\Delta^{i_t}\vert \pb^t], \pb^t -\pb^*\rangle + \varepsilon^2\bbE[\Delta\vert \pb^t]\\
    &= \Vert\pb^t -\pb^*\Vert^2 -\frac{2\varepsilon}{n}\langle\sum_{i=1}^n-\Delta^i, \pb^t-\pb^*\rangle + \varepsilon^2\Delta\\
    & \leq \Vert\pb^t -\pb^*\Vert^2 -\frac{2\varepsilon}{n}\left(L(\pb^t)-L(\pb^*)\right)+ \varepsilon^2\Delta
\end{align*}

Therefore, taking the expectation with respect to $\pb^t$, and rearranging the terms, this can be rewritten as
\begin{align*}
    \frac{2\varepsilon}{n}\bbE[L(\pb^t)-L(\pb^*)] \leq  \bbE[  \Vert \pb^{t} - \pb^*\Vert^2 ] -  \bbE[  \Vert \pb^{t+1} - \pb^*\Vert^2] + \varepsilon^2\Delta
\end{align*}
Summing up over the rounds from $0$ to $T-1$, we obtain 
\begin{align*}
     \frac{2\varepsilon}{n} \sum_{t=0}^{T-1} \bbE[L(\pb^t)-L(\pb^*)]
     &\leq \sum_{t=0}^{T-1} \bbE[  \Vert \pb^{t} - \pb^*\Vert^2 ] -  \bbE[  \Vert \pb^{t+1} - \pb^*\Vert^2] + \varepsilon^2\Delta(T-1)\\
     &\leq \Vert \pb^0 - \pb^*\Vert^2 + \varepsilon^2\Delta(T-1)
\end{align*}
where the last inequality is obtained by computing the telescoping sum and bounding by removing the negative term $\bbE[\Vert \pb^{T} - \pb^*\Vert^2 ]$. Dividing both terms by $\nicefrac{2\varepsilon T}{n}$, the expression becomes
\begin{align*}
    \frac{1}{T}\sum_{t=0}^{T-1} \bbE[L(\pb^t)-L(\pb^*)] \leq \frac{n\Vert \pb^0 - \pb^*\Vert^2} {2\varepsilon T} + \frac{n\Delta\varepsilon}{2}\,.
\end{align*}
Consider now the following lemma
\begin{lemma}\label{lem:conv_L_tilde}
    The function $L$ defined above is convex.
\end{lemma}
Combining this with Jensen's inequality, it is possible to deduce
\[
\bbE[L(\bar{\pb}^T)-L(\pb^*)] \leq \frac{1}{T}\sum_{t=0}^{T-1} \bbE[L(\pb^t)-L(\pb^*)]\,,
\]
and thus
\[
\bbE[L(\bar{\pb}^T)-L(\pb^*)] \leq \frac{2nmR^2} {2\varepsilon T} + \frac{n\Delta\varepsilon}{2}\,,
\]
where we have used the fact that, thanks to the choice of $\pb^0$ and $\pb^*$,  $\Vert \pb^0 - \pb^*\Vert^2 \leq 2mR^2$. Finally, choosing $\varepsilon$ as in the statement, we obtain
\[
\bbE[L(\bar{\pb}^T)-L(\pb^*)] \leq nR\sqrt{\frac{2m\Delta}{T}}\,. 
\]
To obtain a final bound on the quantity  $\Vert \bar{\pb}^T - \pb^*\Vert_{\infty}$,  note that since $L$ is convex, there exists a $\Phi \subseteq \Omega$ such that $\pm \eta \eb^\Phi$ is a steepest descent direction (cf.~\citet[Theorem 2.5]{Shioura2017algorithms}). Without loss of generality, suppose it is $\eta\eb^\Phi$, as the case $-\eta\eb^\Phi$ is analogous. Due to the structure of the aggregate LIP of $v$, there is a bundle, say $\Psi$, that is demanded at both $\pb$ and $\pb' \coloneqq \pb + \eta \eb^\Phi$.
\[
    L(\pb) - L(\pb') = v(\Psi) - \langle\pb , \eb^{\Psi} \rangle- (v(\Psi) - \langle\pb' ,\eb^{\Psi}\rangle) = \langle\eta \eb^\Phi , \eb^{\Psi}\rangle > 0.
\]
The last inequality follows from the fact that $\pb \notin P^*$, so $L$ strictly decreases in the steepest descent direction $\eta \eb^\Phi$. Moreover, since the vectors $\eb^\Phi$ and $\eb^{\Psi}$ have components either 0 or 1, the fact that the difference in the function is non-zero tells us that the inner product between the two vectors is at least 1, hence $\langle\eta \eb^\Phi, \eb^\Psi\rangle > \eta$. 
This signifies the fact that a shift of $\eta$ in any direction from an initial price $\pb$ causes the function to increase of at least $\eta$. Therefore, choosing $\eta = nR\sqrt{\frac{2m\Delta}{T}}$, we can conclude that since $\bbE[L(\bar{\pb}^T)-L(\pb^*)] \leq \eta$, it must hold 
\[
    \bbE[\Vert \bar{\pb}^T - \pb^*\Vert_{\infty}] \leq nR\sqrt{\frac{2m\Delta}{T}}\,.
\]
\end{proof2}

\begin{proof2}[Proof of \Cref{lem:Delta_i_subgradient}]
Using the vector notation, let $\Phi$ be the bundle of demanded trades by agent $i$ at $\pb$. Due to the definition of $L^i$, for any other vector of prices $\qb$, it holds
\[
L^i(\qb) = \max_{\Psi} \left [ v^i(\Psi) - \langle \boldsymbol{\chi^i}\times\eb^{\Psi}, \qb \rangle \right ] + \langle\boldsymbol{\chi^i} , \qb\rangle \geq v^i(\Phi) -\langle \boldsymbol{\chi^i}\times\eb^{\Phi}, \qb \rangle + \langle \boldsymbol{\chi^i}, \qb\rangle\,.
\]
Hence
\begin{align*}
L^i(\qb) - L^i(\pb) 
&\geq v^i(\Phi) - \langle\boldsymbol{\chi^i}\times\eb^{\Phi}, \qb\rangle + \langle \boldsymbol{\chi^i},\qb\rangle -\left(v^i(\Phi) - \langle \boldsymbol{\chi^i}\times\eb^{\Phi}, \pb\rangle + \langle\boldsymbol{\chi^i}, \pb\rangle\right)\\[1em]
&=\langle \boldsymbol{\chi^i}\times\eb^{\Phi}, \pb - \qb\rangle + \langle \boldsymbol{\chi}^i ,\qb - \pb\rangle\\[1em]
&= \langle\boldsymbol{\chi^i}- \boldsymbol{\chi^i}\times\eb^{\Phi},\qb - \pb\rangle\\[1em]
&= \langle -\Delta^i ,\qb - \pb\rangle 
\end{align*}
Therefore $-\Delta^i$ is a subgradient of $L^i$ at $\pb$.
\end{proof2}

\begin{proof2}[Proof of \Cref{lem:min_lyapunov}]
This is a standard result from literature (\textit{e.g.} \cite{fujishige2025universally}) that we report for the sake of completeness.

Let $\pb^*$ be a competitive equilibrium price vector. Then there exists a corresponding competitive equilibrium allocation $\Phi$ such that
\[
L(\pb^*) = \sum_{i \in I} \left [ v^i(\Phi_i) - \sum_{\omega \in \Phi_i} \chi^i_\omega p_\omega \right ] = \sum_{i \in I} v^i(\Phi_i) = w(I)\,,
\]
where we have used the fact that $\sum_{i \in I} \sum_{\omega \in \Phi_i} \chi^i_\omega p_\omega = 0$, since each trade has one buyer and one seller.
The last equality holds because competitive equilibria are efficient. For any non-equilibrium prices $\pb$, and any agent $i \in I$, we have
\begin{equation*}
\max_{\Psi \subseteq \Omega_i} \left [ v^i(\Psi) - \sum_{\omega \in \Psi} \chi^i_\omega p_\omega \right ] \geq v^i(\Phi_i) - \sum_{\omega \in \Phi_i} \chi^i_\omega p_\omega.
\end{equation*}
Moreover, this inequality is strict for at least one agent, as $\pb$ cannot support the efficient allocation $\Phi$. This implies that $L(\pb) > L(\pb^*) = w(I)$.
\end{proof2}

\begin{proof2}[Proof of \Cref{lem:conv_L_tilde}]
In order to prove that $L$ is a convex function we will show that the individual $L^i$ are.
To this end, fix an $i$ and assume by contradiction the opposite. Specifically that, given $\pb$ and $\qb$ vector of prices, it exists some $t \in [0, 1]$ for which
\[
L^i(t\pb+(1-t)\qb) > tL^i(\pb) +(1-t)L^i(\qb)
\]
Let $\Phi$ the bundle of trades demanded at price $t\pb+(1-t)\qb$, by definition it holds
\[
L^i(t\pb+(1-t)\qb) = v^i(\Phi) - \sum_{\omega \in \Phi }\chi_{\omega}^i(tp_{\omega}+(1-t)q_{\omega}) + \langle \boldsymbol{\chi}^i, t\pb+(1-t)\qb\rangle\,.
\]
Moreover, since the value of the function $L^i$ correspond to the maximal utility achievable at a certain price, it holds
\[
L^i(\pb) \geq v^i(\Phi) - \sum_{\omega \in \Phi }\chi_{\omega}^ip_{\omega} +\langle  \boldsymbol{\chi}^i,\pb\rangle
\]
and 
\[
L^i(\pb) \geq v^i(\Phi) - \sum_{\omega \in \Phi }\chi_{\omega}^iq_{\omega} + \langle \boldsymbol{\chi}^i,\qb\rangle\,.
\]
From which 
\begin{align*}
   v^i(\Phi) - \sum_{\omega \in \Phi }\chi_{\omega}^i(tp_{\omega}&+(1-t)q_{\omega}) +\langle \boldsymbol{\chi}^i, t\pb+(1-t)\qb\rangle \\[1em]
   &> tL^i(\pb) +(1-t)L^i(\qb)\\[1em]
   &\geq t\left( v^i(\Phi) - \sum_{\omega \in \Phi }\chi_{\omega}^ip_{\omega} + \langle\boldsymbol{\chi}^i,\pb\rangle \right) +(1-t)\left(v^i(\Phi) - \sum_{\omega \in \Phi }\chi_{\omega}^iq_{\omega} + \langle \boldsymbol{\chi}^i,\qb\rangle\right)\\[1em]
   &=v^i(\Phi) - \sum_{\omega \in \Phi }\chi_{\omega}^i(tp_{\omega}+(1-t)q_{\omega}) + \langle\boldsymbol{\chi}^i, t\pb+(1-t)\qb\rangle
\end{align*}
From which, simplifying the price term we find $v^i(\Phi) > v^i(\Phi) $, which is a contradiction.

The proof for the general function $L$ then follows from the fact that a sum of convex functions is itself 
convex.
\end{proof2}

\subsection{Proofs of \Cref{sec:stability-fairness}}
\label{sec:proofs-stability-fairness}
\begin{proof2}[Proof of \cref{prop:core-imputation-implementation}]\phantomsection\label{proof:prop:core-imputation-implementation}
Fix some prices $\qb$ that extend $\Phi$ to a competitive equilibrium $(\qb, \Phi)$. \cite{hatfield2013stability} show that such prices exist and that $(\qb, \Phi)$ with prices restricted to $\Phi$ is a core outcome, as preferences are fully substitutable.

Suppose we (temporarily) remove all trades apart from $\Phi$ from the market network. Then the network is partitioned into $k \geq 1$ connected components. Let $I_1, \ldots, I_k$ be the agents of these components. First we establish that, for any $j \in [k]$,
\begin{equation}
\label{eq:component-utility-sum}
    \sum_{i \in I_j} u^i(\Phi, \qb) = w(I_j).
\end{equation}
Indeed, $\sum_{i \in I_j} u^i(\Phi, \qb) \geq w(I_j)$ follows immediately from the coalitional rationality that $(\Phi, \qb)$ satisfies as a core outcome. For the converse inequality, we note that the buyer and seller in each trade of $\Phi$ are either both in $I_j$, or both outside $I_j$ (as $I_j$ are the agents of a connected component of the market network). Letting $\Phi_{I_j}$ be the set of trades between agents in $I_j$, we get $\sum_{i \in I_j} u^i(\Phi, \qb) = \sum_{i \in I_j} u^i(\Phi_{I_j}, \qb)$. As $\Phi_{I_j}$ is also a valid set of trades for the market $\mathcal{M}[I_j]$ restricted to agents $I_j$ and $w(I_j)$ is the maximum social welfare achievable in this sub-market, we have $\sum_{i \in I_j} u^i(\Phi_{I_j}, \qb) \leq w(I_j)$.

Moreover, as the $I_1, \ldots, I_k$ are disjoint and $(\Phi, \qb)$ is a core outcome,
\[
    \sum_{j \in [k]} \sum_{i \in I_j} u^i(\Phi, \qb) = \sum_{i \in I} u^i(\Phi, \qb) = w(I).
\]

We now see what this implies about the payoff vector $\ub$. As $\ub$ lies in the payoff core, it satisfies $\sum_{i \in I} u_i = w(I)$ and $\sum_{i \in C} u_i \geq w(C)$ for all $C \subseteq I$. Now we use the following fact: for any two vectors $\yb, \zb \in \R^k$ with $\sum_{i \in [k]} y_i = \sum_{i \in [k]} z_i$ and $\yb \geq \zb$, we have $\yb = \zb$. By defining $\yb$ and $\zb$ as $y_j = \sum_{i \in I_j} u_i$ and $z_j = \sum_{i \in I_j} u^j(\Phi, \qb)$, and verifying that $y_j \geq w(I_j) = z_j$ for all $j \in [k]$, we thus see that
\begin{equation}
\label{eq:equality-by-component}
\sum_{i \in I_j} u_i = \sum_{i \in I_j} u^i(\Phi, \qb), \quad j \in [k].
\end{equation}

Finally, we can use transferability of utility between any two connected agents in the market network to show that there exist prices at which $u_i = u^i(\Phi, \qb)$ for all $i \in I$. Suppose equality doesn't hold for some agent with prices $\qb$. By \cref{eq:equality-by-component}, it follows that $u_i < u^i(\Phi, \qb)$ and $u_l > u^l(\Phi, \qb)$ for some two agents $i,l \in I_j$. As they both lie in the same component, there exists a path from $i$ to $l$ in the market network that consists of trades from $\Phi$. We can now modify the prices of all these trades by adding or subtracting the same constant to get prices $\qb'$ at which $u_i = u^i(\Phi, \qb')$ or $u_l = u^l(\Phi, \qb')$ (or both) holds. This price-adjustment procedure does not change the utility of the `interior' agents on the path. (We give more details on this procedure immediately below this proof.)

We repeat this process of finding two agents in the same connected component who do not satisfy equality $u_i = u^i(\Phi, \qb)$, and shifting prices along a path between them, until all agents satisfy the equality. Note that we are guaranteed to make progress, as in every step we strictly reduce the number of agents who do not satisfy equality. Letting $\pb$ denote the final prices, we see that all agents now satisfy the equality $u_i = u^i(\Phi, \pb)$, and so $(\Phi, \pb)$ is a core outcome that achieves payoffs~$\ub$.
\end{proof2}

Fix $(\Phi, \qb)$, and suppose $i$ and $j$ are two agents connected by a path $P$ of trades from $\Phi$. Call a trade on this path \textit{forward} if it points in the direction of $j$, and \textit{backward} if it points in the direction of $i$. We claim that by adding $\varepsilon$ to the price of all forward trades and subtracting $\varepsilon$ from the price of all backward trades, we get prices $\qb'$ that satisfy the following:
\begin{itemize}
    \item the utility of agent $i$ is increased by $\varepsilon$, so $u^i(\Phi, \qb') = u^i(\Phi, \qb) + \varepsilon$;
    \item the utility of agent $j$ is decreased by $\varepsilon$, so $u^j(\Phi, \qb') = u^j(\Phi, \qb) - \varepsilon$;
    \item the utility of all other agents along the path remains the same, so $u^l(\Phi, \qb') = u^l(\Phi, \qb)$ for $l \notin \{i,j\}$.
\end{itemize}
To see this for agent $i$, note that if the (unique) trade involving $i$ is a forward trade, then it is a selling trade for the agent whose price increases, and so her utility increases; and if the trade is a backward trade, then it is a buying trade and so her utility decreases. The argument for agent $j$ is analogous. Finally, any other agent $l$ along the path is involved in two trades. If both trades are buying, or both selling, then one is forward and one is backward, so the price changes cancel out when writing out the agent's utility at $\qb'$. On the other hand, if one trade is buying, and the other is selling, then both are forward or both are backward, and so the price changes also cancel out.

\begin{proof2}[Proof of \Cref{lem:pos-utility-3-agents}]
\phantomsection\label{proof:lem:pos-utility-3-agents}
Fix a market with welfare function $w$. Suppose there are exactly three essential agents, $1,2$ and $3$; the case for $1$ or $2$ agents is analogous but simpler. All inessential agents in the market get utility $0$ in any core imputation. Let $C$ be the set of all core imputations. If, for each essential agent $i \in \{1,2,3\}$, the core contains an imputation with $x_i > 0$, then the leximin core outcome gives each essential agent a positive utility and we are done. So, without loss of generality, suppose $C$ does not contain any imputations in which agent $3$ gets positive utility (i.e., $x_3 = 0, \forall \xb \in C $). Hence, $x_1 + x_2 = w(\{1,2,3\})$ and $x_3 = 0$ for all core imputations $\xb \in C$. Fix any core imputation $\xb$.

We now argue that $w(\{1,2\}) = w(\{1,2,3\})$. Suppose not, so $w(\{1,2\}) = w(\{1,2,3\}) - \varepsilon$ for some $\varepsilon \in (0,1]$. Construct the imputation $\yb = (x_1 - \delta^1, x_2 - \delta^2, x_3 + \delta^1 + \delta^2)$ with $\delta^1, \delta^2 \geq 0$ chosen such that $\delta^1 + \delta^2 \leq \varepsilon$ and $y_1, y_2 \geq 0$. Such an $\yb$ exists because at least one of $x_1$ and $x_2$ is strictly positive. We see that $\yb$ is a core imputation, as it satisfies all core constraints:
\begin{equation*}
\begin{aligned}
    y_1, y_2, y_3 &\geq 0,\\
    y_1 + y_2 &= x_1 + x_2 - \delta^1 - \delta^2 &&\geq w(\{1,2,3\}) - \varepsilon = w(\{1, 2\}), \\
    y_1 + y_3 &= x_1 + x_3 + \delta^2 &&\geq w(\{1, 3\}), \\
    y_2 + y_3 &= x_2 + x_3 + \delta^1 &&\geq w(\{2, 3\}),\\
    y_1 + y_2 + y_3 &= x_1 + x_2 + x_3 &&= w(\{1, 2, 3\}).
\end{aligned}
\end{equation*}
But this is a contradiction, as we have $y_3 > 0$ but $z_3 = 0$ holds for all core imputations $\zb \in C$ by assumption. So $w(\{1,2\}) = w(\{1,2,3\})$, implying that agent $3$ is inessential, as removing them does not reduce the market value. This contradicts our assumption that agent $3$ is essential, and we are done.
\end{proof2}

\begin{proof2}[Proof of \cref{prop:leximin-leximax-minvar-coincide-3-agents}]\phantomsection\label{proof:prop:leximin-leximax-minvar-coincide-3-agents}
We know that inessential agents get utility $0$ in any core imputation. So we assume, without loss of generality, that the market has three agents $\{1,2,3\}$. We also assume that the value of the market is $1$. This implies that we can describe the core $C$ by non-negative coefficients $a_1 \leq a_2 \leq a_3$ such that
\[
C \coloneqq \{ \xb \mid x_1 + x_2 + x_3 = 1, \text{ and } 0 \leq x_i \leq a_i \text{ for } i = 1,2,3 \}.
\]
As the core is non-empty, we must have $a_1 + a_2 + a_3 \geq 1$.

We now distinguish between three cases. In each case, we define a point $\xb$ in the core, and show that it is the leximin, leximax, and minimum variance core imputation. Note that $a_1 \leq a_2 \leq a_3$ implies $a_3 \geq \frac{1}{3}$. Moreover, $\xb$ is the minimum variance core imputation iff it minimises, for any scalar $c$, the function
\[
    \varphi_c(\yb) \coloneqq \sum_{i=1}^3 \left ( y_i - c \right )^2.
\]

\textbf{Case 1:} Suppose $a_1 \geq \frac{1}{3}$. We define $\xb = (\frac{1}{3}, \frac{1}{3}, \frac{1}{3})$. This point is clearly in the core and leximin as well as leximax. Moreover, $\xb$ also minimises $\varphi_{\frac{1}{3}}(\yb)$ over all $\yb \in C$, and so $\xb$ is minvar.

\textbf{Case 2:} Suppose $a_1 < \frac{1}{3}$ and $a_2 \geq \frac{1-a_1}{2}$. Note that $a_1 < \frac{1}{3} < \frac{1-a_1}{2}$. We define $\xb = (a_1, \frac{1-a_1}{2}, \frac{1-a_1}{2})$, which is clearly in the core.

We first note that if the leximin (leximax) solution $\yb$ satisfies $y_1 = x_1 = a_1$, then it follows from $a_2, a_3 \geq \frac{1-a_1}{2}$ that $\yb = \xb$ and so $\xb$ is leximin (leximax).
Every entry of the leximin $\yb$ must be weakly greater than $x_1$, so $y_1 = a_1 = x_1$ and $\xb$ is leximin. Now suppose $\yb$ is leximax and $y_1 < x_1 = a_1$. Then $y_2 + y_3 > 1 - a_1$, and so $y_2$ or $y_3$ is strictly greater than $\frac{1-a_1}{2}$. But then $\yb$ is not leximax, as it has an entry strictly greater than the maximum entry of $\xb$. So we have $y_1 = x_1 = a_1$, and $\xb$ is leximax.

Finally, $\xb$ is minvar because $\xb$ minimises $\varphi_{\frac{1-a_1}{2}}(\yb)$ over all $\yb \in C$.

\textbf{Case 3:} Suppose $a_1 < \frac{1}{3}$ and $a_2 < \frac{1-a_1}{2}$. We define $\xb = (a_1, a_2, 1 - a_1 - a_2)$.
\begin{itemize}
    \item Secondly, $a_2 < \frac{1 - a_1}{2}$ implies $1 - a_1 - a_2 > a_2$.
\end{itemize}
Clearly, $x_1 + x_2 + x_3 = 1$, $0 \leq x_1 \leq a_1$ and $0 \leq x_2 \leq a_2$. Secondly, $a_1 + a_2 + a_3 \geq 1$ and our assumptions on $a_1$ and $a_2$ imply that $\frac{1}{3} < 1 - a_1 - a_2 \leq a_3$. Hence $0 \leq x_3 \leq a_3$, so $\xb \in C$.

Next we argue that $\xb$ is leximin and leximax. For any $\yb \in C$, we see that $y_3$ is strictly larger than $y_1$ and $y_2$, as
\[
 y_3 = 1 - y_1 - y_2 \geq 1 - a_1 - a_2 > a_2 \geq a_1,
\]
and the core inequalities give us $a_2 \geq y_2$ as well as $a_1 \geq y_1$. As $x_1 \geq y_1$ and $x_2 \geq y_2$ for any core imputation $\yb$, it follows that $\xb$ is leximin and leximax.

To see that $\xb$ is minvar, we first note that $x_1 = a_1 < x_2 = a_2 < \frac{1-a_1}{2} < x_3 = 1 - a_1 - a_2$. It follows that $\xb$ is minvar, because $\xb$ minimises $\varphi_{\frac{1-a_1}{2}}(\yb)$ over all $\yb \in C$.
\end{proof2}

\begin{proposition}
\label{proposition:CE-taxed-market-is-efficient}
    The allocation $\Phi$ of any competitive equilibrium $(\pb, \Phi)$ of the taxed market defined in \cref{sec:impossibility-of-fairness-in-the-core} is an efficient set of trades for the original market.
\end{proposition}
\begin{proof}
    Fix an arbitrary competitive equilibrium $(\pb, \Phi)$ in the taxed market, and let $\Psi$ be an efficient set of trades in the original setting. We have
    \[
    W(\Phi) = \sum_{i \in I} v^i(\Phi) \leq \sum_{i \in I} v^i(\Psi) = W(\Psi).
    \]
    As $(\pb, \Phi)$ is a competitive equilibrium in the taxed market, we have $\widehat{u}^i(\pb, \Psi) \leq \widehat{u}^i(\pb, \Phi)$ for each agent $i$. Together with $W(\Phi) \leq W(\Psi)$, this implies $u^i(\pb, \Psi) \leq u^i(\pb, \Phi)$ for each agent $i$, and so
    \[
    \sum_{i \in I} v^i(\Psi) = \sum_{i \in I} u^i(\pb, \Psi) \leq \sum_{i \in I} u^i(\pb, \Phi) = \sum_{i \in I} v^i(\Phi).
    \]
    Hence $\sum_{i \in I} v^i(\Phi) = \sum_{i \in I} v^i(\Psi)$, so $\Phi$ is also efficient.
\end{proof}

\section{Additional Material and Examples}\label{app:sec:additional-examples}

Let $x^i_{\Phi} \in \{0,1\}$ denote the allocation of bundle $\Phi \in 2^{\Omega^i}$ to agents $i$. We write the social welfare maximization problem as follows.

\begin{maxi*}|s|{x^i_{\Phi}}{\sum_{i \in I} \sum_{\Phi \in 2^{\Omega^i}} v^i(\Phi) x^i_{\Phi}}
    {\label{opt:P1''}}{}\tag{$SW$}
    \addConstraint{\sum_{i\in I} \sum_{\Phi \in 2^{\Omega^i}} \chi^i_{\omega \Phi} x^i_{\Phi} = 0,}{\quad \forall \omega \in \Omega}{
    } 
    \addConstraint{\sum_{\Phi\in 2^{\Omega^i}} x^i_{\Phi} \leq 1,}{\quad \forall i \in I}{
    }
    \addConstraint{x^i_\Phi \in \{0,1\}}{\quad \forall i\in I, \Phi \in 2^{\Omega^i}}
\end{maxi*}

The SW problem of the divisible extension of the network market has a dual formulation which yields a Lyapunov function.

\begin{mini*}|s|{\pb, \ub}{\sum_{i\in I} u^i}
    {\label{opt:D1''}}{}\tag{$SWD$}
    \addConstraint{u^i + \sum_{\omega\in\Omega} \chi_{\omega\Phi}^i p_\omega \geq v^i(\Phi),}{\quad \forall i \in I, \forall \Phi \in 2^{\Omega^i}}{\quad [x^i_{\Phi}]}
    \addConstraint{u^i\geq 0}{\quad \forall i\in I}
\end{mini*}

We notice that the dual problem is equivalent to minimizing the sum of indirect utilities, which can also be seen as a linear programming formulation of the Lyapunov function developed in \cref{sec:convergence-of-decentralized-trading}.

\begin{example}\label{example:core-not-NE}
This example is taken from the supplementary material of \citet{hatfield2013stability}. Outcome $(\pb = (2,2),\Phi = (\chi,\varphi))$ is in the core but cannot be achieved by a Nash equilibrium.

\begin{figure}[htp]
\centering

\begin{subfigure}[c]{0.49\textwidth}
    \centering
    \begin{tikzpicture}[xscale=2, yscale=1]
        \node[agent] (i) at (0,0) {$i$};
        \node[agent] (j) at (2,0) {$j$};
        \draw[trade] (i) to[bend left] node[midway,fill=white] {$\chi$} (j);
        \draw[trade] (i) to[bend right] node[midway,fill=white] {$\varphi$} (j);
    \end{tikzpicture}
    \caption{Trading network}
    \label{fig:two-trade-network}
\end{subfigure}\hfill
\begin{subfigure}[c]{0.49\textwidth}
    \centering
    \begin{tabular}{cccc}
        \toprule
        Agent & $\{\chi\}$ & $\{\varphi\}$ & $\{\chi,\varphi\}$ \\
        \midrule
        $i$ & $0$ & $0$ & $-3$ \\
        $j$ & $5$ & $5$ & $9$ \\
        \bottomrule
    \end{tabular}
    \caption{Valuations}
    \label{fig:two-trade-valuations}
\end{subfigure}

\caption{A two-agent market with two parallel trades and corresponding valuations}
\label{fig:parallel-trades-example}
\end{figure}
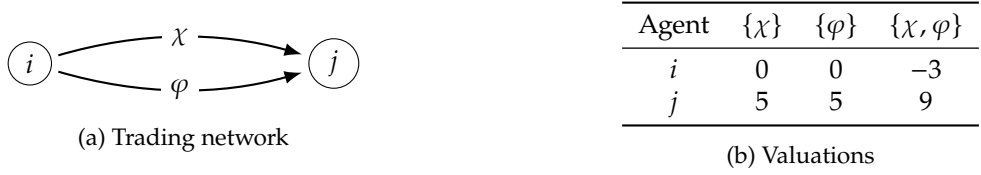
\end{example}

\begin{example}
\label{example:two-agents-two-opposite-trades}
Consider the following example from \cite{hatfield2013stability}, a market with two agents, $I = \{1,2\}$, and two trades, $\Omega = \{\chi, \varphi\}$. Agent $1$ has complements preferences, and agent $2$ has substitutes preferences.

\begin{center}
    \begin{minipage}{0.49\textwidth}
    \centering
    \begin{tikzpicture}[xscale=3.5, yscale=1.8]
    \node[agent] (s) at (0,0) {$1$};
    \node[agent] (b) at (1,0) {$2$};
    \draw[trade] (s) to[bend left] node[midway,fill=white] {$\chi$} (b);
    \draw[trade] (b) to[bend left] node[midway,fill=white] {$\varphi$} (s);
    \end{tikzpicture}
    \end{minipage}
    \begin{minipage}{0.49\textwidth}
    \centering
    \begin{tabular}{ccccc}
    \toprule
    & $\{ \chi \}$ & $\{ \varphi \}$ & $\{\chi, \varphi \}$ \\
    \midrule
    Agent $1$ & $-4$ & $-4$ & $-4$ \\
    Agent $2$ & $3$ & $3$ & $3$ \\
    \bottomrule
    \end{tabular}
    \end{minipage}
\end{center}
Any offers between $3$ and $4$ lead to unilateral deviations that prohibit trade, as both agents would be making a loss. The only pure Nash equilibrium is such that $\offer^i > 4$ and $\offer^j < 3$, which is $\varepsilon$-tight only for $\varepsilon > 1$.
\end{example}

\begin{example}\phantomsection\label{ex:dynamics}
This example aims at illustrating the evolution of the dynamics as defined in \Cref{alg:offer-based-dynamics} and \Cref{alg:price-based-dynamics}, highlighting the difference between the two. The market considered is the one detailed in \cref{fig:two-seller-example}, where two different sellers have the possibility of trading $\omega_1$ and $\omega_2$ with a common buyer. In both cases we will assume $\varepsilon = \nicefrac{1}{2}$.
\begin{figure}[h!]
    \centering
    \begin{subfigure}[b]{0.49\textwidth}
    \centering
    \begin{tikzpicture}[xscale=3.5, yscale=0.9]
    \node[agent] (s1) at (0,1) {$s_1$};
    \node[agent] (s2) at (0,-1) {$s_2$};
    \node[agent] (b) at (1,0) {$b$};
    \draw[trade] (s1) to[bend left] node[pos=0.5,fill=white] {$\omega_1$} (b);
    \draw[trade] (s2) to[bend right] node[pos=0.5,fill=white] {$\omega_2$} (b);
    \end{tikzpicture}
    \caption{Graph}
    \end{subfigure}
    \begin{subfigure}[b]{0.49\textwidth}
    \centering
    \begin{tabular}{cccc}
        \toprule
         Agent & $\{\omega_1\}$ & $\{\omega_2\}$ & $\{\omega_1,\omega_2\}$ \\
         \midrule
         Seller $s_1$ & $-2$ & -- & -- \\
         Seller $s_2$ & -- & $-2$ & -- \\
         Buyer $b$    & $3$ & $3$ & $3$ \\
         \bottomrule
    \end{tabular}
    \caption{Valuations}
    \end{subfigure}
    \caption{Two-seller, one-buyer market with two trades.}
    \label{fig:two-seller-example}
\end{figure}
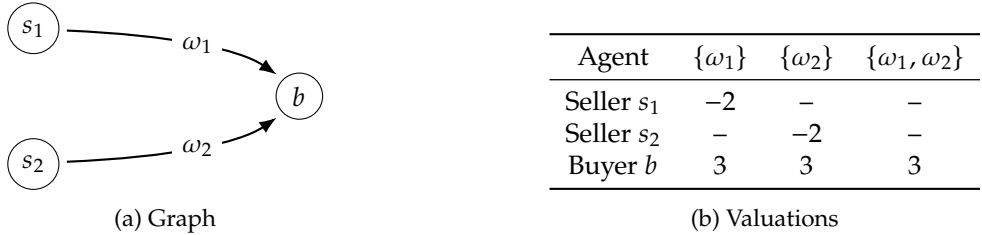

\Cref{tab:two-seller-example-offers} showcases a possible evolution of the offer dynamic. The columns of the table represent the offers of the agents, whether the trades are demanded by the active agent (Yes / No), the updates that each applies to its own offers, and the set of unsatisfied agents $U$. Each row corresponds to one step of the dynamics where the agent specified in the second column is active and updates.

Note that after the initial update from agents $s_1$ and $b$, all the offers are at most $\varepsilon$ apart from the counterpart one. Lastly, note that $s_1$ and $b$ agree immediately on $\omega_1$, hence the first agent will not re-appear in the set $U$, and instead there will be an alternation of the remaining two agents until an agreement is reached. 

\begin{table}[htb!]
\centering
\renewcommand{\arraystretch}{1.15}
\setlength{\tabcolsep}{4pt}
\begin{tabular}{c c c c c c c c c}
\toprule
Round
& active $i$
& $\sigma^b$
& $\sigma^{s_1}$
& $\sigma^{s_2}$
& $\omega_1$
& $\omega_2$
& update
& $U$ after update \\
\midrule
0
& --
& $(2,2)$
& $3$
& $3$
& --
& --
& --
& $b,s_1,s_2$ \\
\midrule
1
& $s_1$
& $(2,2)$
& $3$
& $3$
& Y
& --
& $3\!\to\!\textcolor{blue}{2}$
& $b,s_2$ \\
2
& $b$
& $(2,2)$
& \textcolor{blue}{2}
& $3$
& Y
& N
& $(2,2)\!\to\!(2,\textcolor{darkgreen}{2.5})$
& $s_2$ \\
3
& $s_2$
& $(2,\textcolor{darkgreen}{2.5})$
& $2$
& $3$
& --
& Y
& $3\!\to\!\textcolor{red}{2.5}$
& $b$ \\
4
& $b$
& $(2,2.5)$
& $2$
& \textcolor{red}{2.5}
& Y
& N
& $(2,2.5)\!\to\!(2,\textcolor{darkgreen}{2})$
& $s_2$ \\
5
& $s_2$
& $(2,\textcolor{darkgreen}{2})$
& $2$
& $2.5$
& --
& Y
& $2.5\!\to\!\textcolor{red}{2}$
& $b$ \\
6
& $b$
& $(2,2)$
& $2$
& \textcolor{red}{2}
& Y
& N
& $(2,2)\!\to\!(2,\textcolor{darkgreen}{1.5})$
& $s_2$ \\
7
& $s_2$
& $(2,\textcolor{darkgreen}{1.5})$
& $2$
& $2$
& --
& N
& $2\!\to\!2$
& -- \\
\bottomrule
\end{tabular}
\caption{Offer dynamics with two sellers $s_1,s_2$ and one buyer $b$. Colors indicate actual updates and the corresponding propagated value
in the next round:
$\textcolor{blue}{s_1}$,
$\textcolor{red}{s_2}$,
$\textcolor{darkgreen}{b}$.}
\label{tab:two-seller-example-offers}
\end{table}

In \cref{tab:two-seller-example-prices}, we demonstrate the trajectory of the alternative, price-based, dynamic. The last column of each row indicates the prices after the update. The third and fourth columns show the demand (Yes / No) for each of the two trades with respect to the prices from the preceding row. Note that the active agent is chosen uniformly at random from the set $I$ in each round, which can cause the same agent to be active for several consecutive times, as happens here in rounds~1 and~2. Observe also that the buyer switches her demand from $\omega_1$ to $\omega_2$ as prices change. Lastly, note that even when the equilibrium prices are transiently reached in round~4, the dynamics does not stop and the buyer departs from it in round~5.
\begin{table}[htb!]
\centering
\renewcommand{\arraystretch}{1.15}
\setlength{\tabcolsep}{4pt}
\begin{tabular}{c c c c c }
\toprule
Round
& active $i$
& $\omega_1$
& $\omega_2$
& Current $\pb$ \\
\midrule
0
& --
& --
& --
& $(2,2)$ \\
\midrule
1
& $b$
& Y
& N
& $(2, \textcolor{darkgreen}{1.5})$ \\
2
& $b$
& N
& Y
& $(\textcolor{darkgreen}{1.5}, 1.5)$ \\
3
& $s_1$
& N
& -
& $(\textcolor{blue}{2}, 1.5)$ \\
4
& $s_2$
& -
& N
& $(2, \textcolor{red}{2})$ \\
5
& $b$
& Y
& N
& $(2, \textcolor{darkgreen}{1.5})$ \\
\bottomrule
\end{tabular}
\caption{Price dynamics with two sellers $s_1,s_2$ and one buyer $b$. The dynamics uses the tie-breaking rule where $d^i(\pb)$ is chosen as the lexicographically smallest bundle in $D^i(\pb)$ of maximum cardinality (also defined in \cref{sec:network-trading-market}). Colors indicate updated offers:
$\textcolor{blue}{s_1}$,
$\textcolor{red}{s_2}$,
$\textcolor{darkgreen}{b}$.}
\label{tab:two-seller-example-prices}
\end{table}
\end{example}
\begin{example}\label{example:leximin-not-equal-to-leximax}
In this example with four essential agents, the leximin core imputation does not coincide with the leximax core imputation.

We show that the leximin core imputation is
$(u_{s_1}, u_{s_2}, u_{b_1}, u_{b_2}) = (11.5, 6.5, 8.5, 6.5)$
and the leximax core imputation is
$(u_{s_1}, u_{s_2}, u_{b_1}, u_{b_2}) = (11, 6, 9, 7)$.
Firstly, we can verify immediately that both vectors are core imputations.
Secondly, for any core imputation $u$, the core constraints tell us that
$u_{s_1} + u_{s_2} + u_{b_1} + u_{b_2} = w(\{s_1, s_2, b_1, b_2\}) = 33$,
$u_{s_1} + u_{b_2} \geq w(\{s_1, b_2\}) = 18$
and
$u_{s_2} + u_{b_1} \geq w(\{s_2, b_1\}) = 15$.
Thus,
$u_{s_1} + u_{b_2} = 18$.
Likewise, we can show
$u_{s_2} + u_{b_1} = 15$.

Now suppose that $u$ is the leximin core imputation.
By definition of leximin, it holds that $u_{i} \geq 6.5$ for each
$i \in \{s_1, s_2, b_1, b_2\}$.
Moreover,
$w(\{s_1, s_2, b_1, b_2\}) = 33$
and
$w(\{s_1, b_1\}) = 20$
tell us that
$u_{s_1} + u_{s_2} + u_{b_1} + u_{b_2} = 33$
and
$u_{s_1} + u_{b_1} \geq 20$,
so
$u_{s_2} + u_{b_2} \leq 13$.
Thus
$u_{s_2}, u_{b_2} \geq 6.5$
implies
$u_{s_2} = u_{b_2} = 6.5$.
By the previous paragraph, this implies
$u_{s_1} = 11.5$
and
$u_{b_1} = 8.5$,
so
$u = (11.5, 6.5, 8.5, 6.5)$
as required.

Now suppose that $u$ is the leximax core imputation.
By definition of leximax, it holds that $u_{i} \leq 11$ for each
$i \in \{s_1, s_2, b_1, b_2\}$.
We again apply the core constraints to $u$ to show that
$u = (11, 6, 9, 7)$.
Specifically,
$u_{s_1} + u_{s_2} + u_{b_1} + u_{b_2} = 33$
and
$u_{s_1} + u_{s_2} + u_{b_1} \geq 26$
implies
$u_{b_2} \leq 7$,
so by
$u_{s_1} + u_{b_2} = 18$
we have
$u_{s_1} \geq 11$.
Thus,
$u_{s_1} = 11$
and
$u_{b_2} = 7$.
Due to the grand coalition constraint,
$u_{s_2}$ and $u_{b_1}$ satisfy
$u_{s_2} + u_{b_1} = 15$.
Moreover, constraint
$u_{s_1} + u_{b_1} \geq 20$ together with $u_{s_1} = 11$ implies
$u_{b_1} \geq 9$,
and constraint
$u_{s_1} + u_{b_1} \geq 20$
implies
$u_{s_2} + u_{b_2} \leq 13$.
With $u_{b_2} = 7$, this gives
$u_{s_2} \leq 6$.
As
$(u_{s_1}, u_{s_2}, u_{b_1}, u_{b_2}) = (11, 6, 9, 7)$
dominates the leximax order,
$u = (11, 6, 9, 7)$
as required.

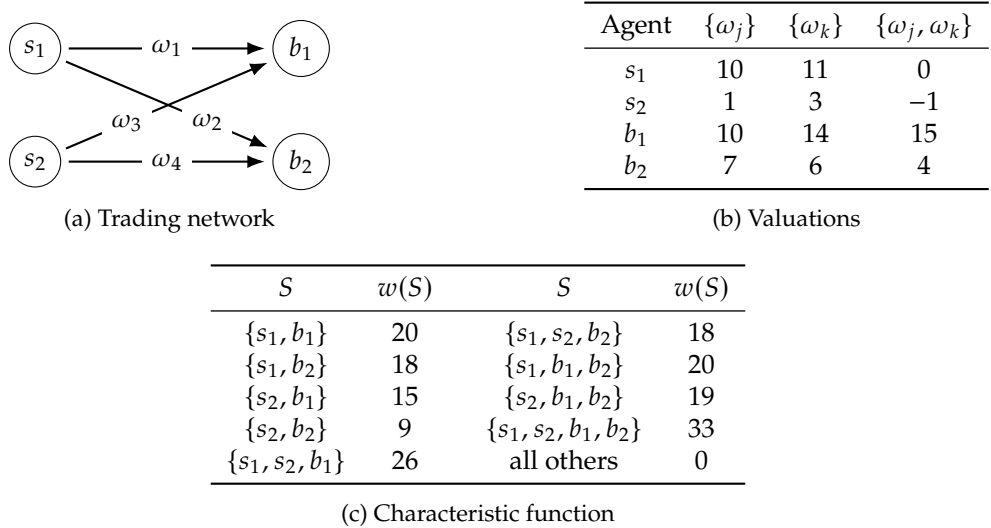
\begin{figure}[htp]
    \centering
    \begin{subfigure}[b]{0.49\textwidth}
        \centering
        \begin{tikzpicture}[xscale=3.5, yscale=1.5]
        \node[agent] (s1) at (0,1) {$s_1$};
        \node[agent] (s2) at (0,0) {$s_2$};
        \node[agent] (b1) at (1,1) {$b_1$};
        \node[agent] (b2) at (1,0) {$b_2$};
        \draw[trade] (s1) -- (b1) node[midway,fill=white] {$\omega_1$};
        \draw[trade] (s1) -- (b2) node[pos=0.7,fill=white] {$\omega_2$};
        \draw[trade] (s2) -- (b1) node[pos=0.3,fill=white] {$\omega_3$};
        \draw[trade] (s2) -- (b2) node[pos=0.5,fill=white] {$\omega_4$};
        \end{tikzpicture}
        \caption{Trading network}
    \end{subfigure}
    \hfill
    \begin{subfigure}[b]{0.49\textwidth}
        \centering
        \begin{tabular}{cccc}
            \toprule
             Agent & $\{ \omega_j \}$ & $\{ \omega_{k} \}$ & $\{ \omega_j, \omega_{k} \}$ \\
             \midrule
             $s_1$      & $10$ & $11$ & $0$ \\
             $s_2$      & $1$ & $3$ & $-1$ \\
             $b_1$      & $10$ & $14$ & $15$ \\
             $b_2$      & $7$ & $6$ & $4$ \\
             \bottomrule
        \end{tabular}
        \caption{Valuations}
    \end{subfigure}

    \vspace{1em}

    \begin{subfigure}[b]{0.9\textwidth}
        \centering
        \begin{tabular}{c c@{\qquad}c c}
            \toprule
            $S$ & $w(S)$ & $S$ & $w(S)$ \\
            \midrule
            $\{s_1,b_1\}$ & $20$  & $\{s_1,s_2,b_2\}$ & $18$ \\
            $\{s_1,b_2\}$ & $18$  & $\{s_1,b_1,b_2\}$ & $20$ \\
            $\{s_2,b_1\}$ & $15$  & $\{s_2,b_1,b_2\}$ & $19$ \\
            $\{s_2,b_2\}$ & $9$   & $\{s_1,s_2,b_1,b_2\}$ & $33$ \\
            $\{s_1,s_2,b_1\}$ & $26$ & all others & $0$ \\
            \bottomrule
        \end{tabular}
        \caption{Characteristic function}
    \end{subfigure}

    \caption{A trading market with four essential agents, where the leximin core does not coincide with the leximax core. The valuation of each agent $i$ is given for $\Omega_i = \{ \omega_j, \omega_k \}$ with $j \leq k$.}
    \label{fig:leximin-not-equal-to-leximax}
\end{figure}
\end{example}

\section{From Markets to Auctions}
\label{app:markets-to-auctions}

We show that the trading network market can be reduced to an auction with indivisible goods and a fixed supply of one unit per good. An auction $\auction = (I, G, \widehat{v})$ consists of a set $I$ of buyers, a set $G$ of goods, and valuations $\widehat{v} = (\widehat{v}^i)_{i \in I}$ for the buyers.
Quasi-linear utilities $\widehat{u}^i$ and demands $\widehat{D}^i$ are defined the usual way from valuations $\widehat{v}$.

To map a networked market $\market$ to an auction $\auction$, we identify the trades of the networked market with the goods in the auction, and each agent in the networked market becomes a buyer in the auction, by analogy with \citet{Hatfield-2019}.

\begin{definition}[Allocation mapping]
     For every active trade $\omega \in \Phi$, we allocate $\omega$ to the buyer $b(\omega)$ in the auction market. For every inactive trade $\omega \notin \Phi$, we allocate $\omega$ to $s(\omega)$. Formally, for any buyer $i$ in the auction, $\Psi^i = \tau^i(\Phi) \equiv \Phi_{i \leftarrow} \cup \Omega_i \setminus \Phi_{i \rightarrow}$.
    In vector notation, $x^{\Psi,i}_\omega = \widehat{\tau}^i(x^{\Phi,i}_\omega) \equiv x^{\Phi,i}_\omega + \mathbf{1}\{i = s(\omega)\}$.
\end{definition}

This establishes a one-to-one correspondence between market allocations and feasible auction allocation collections: Any market allocation $\Phi$ maps to auction allocation $\Psi = (\tau^i(\Phi_i))_{i\in I}$; and any feasible auction allocation $\Psi$ maps to the market allocation $\Phi = (\tau^i(\Psi_i))_{i\in I}$. We note that $\tau(\Phi) \subseteq \Omega_i$ for any $\Phi \subseteq \Omega_i$, and that $\tau = \tau^{-1}$ (but not $\widehat{\tau}$). Each buyer's valuation in the auction matches the corresponding valuation in the network market, subject to the remapping provided by $\tau$, i.e.,~$\widehat{v}(\Psi) = v^i(\tau^i(\Psi))$ for any $\Psi \subseteq \Omega$. \citet{Hatfield-2019} demonstrate that an agent's valuation in the network market is fully substitutable if and only if its corresponding auction market valuation is substitutes. We can prove this more simply by demonstrating that a valuation function has the same LIP when we shift its demanded bundles by a fixed ``offset bundle'', so \cref{fact:substitutes-normals} in \cref{section:geometry-preferences} proves the substitutes property. We go further in proving that arrangements, demand, and competitive equilibria map between the network and the auction market. An arrangement in the auction is a competitive equilibrium if the market is cleared, and all agents are satisfied with their allocation (so demand equals supply).

\begin{definition}
\label{definition:auction-CE}
    An auction arrangement $(\pb, (\Phi_i)_{i \in I})$ for auction $\auction = (I, G, \widehat{v})$ is a \textit{competitive equilibrium (CE)} iff it is feasible and 
    \begin{enumerate}[(i)]
        \item supply is allocated, so $\bigcup_{i \in I} \Phi_i = G$,
        \item every agent demands $\Phi_i$ at $\pb$, so $\Phi_i \in \widehat{D}^i(\pb)$.
    \end{enumerate}
\end{definition}

It is well-known that a competitive equilibrium exists if all buyers have gross substitutes valuations \citep{baldwin2019understanding, Ausubel-2006}. Moreover, any competitive equilibrium maximizes social welfare among all allocations that do not exceed supply, by the social welfare theorems. The social welfare of allocation $(\Psi_i)_{i \in I}$ is given by $\sum_{i \in I} \widehat{v}^i(\Psi_i)$.

Recall that $D^i$ is the demand correspondence of agent $i$ in the networked market, and $\widehat{D}^i$ is the demand correspondence of the same agent in the auction. We see that $\tau$ also maps demanded bundles:

\begin{lemma}
\label{lemma:demand-mapping}
Fix prices $\pb \in \Rn$. For any agent $i \in I$ and any two bundles $\Phi, \Psi \subseteq \Omega_i$ with $\Phi = \tau(\Psi)$, we have $\Phi \in D^i(\pb)$ iff $\Psi \in \widehat{D}^i(\pb)$.
\end{lemma}
\begin{proof}
We have
\begin{align*}
    \Phi \in D^i(\pb)
    &\iff \Phi \in \argmax_{\Phi' \in \Omega_i} \left ( v^i(\Phi') - \sum_{\omega \in \Phi'} \chi^i_{\omega} p_\omega \right )\\
    &\iff \Psi \in \argmax_{\Psi' \in \Omega_i} \left ( v^i(\tau^i(\Psi')) - \sum_{\omega \in \tau^i(\Psi')} \chi^i_{\omega} p_\omega \right ) \\
    &\iff \Psi \in \argmax_{\Psi' \in \Omega_i} \left ( \widehat{v}^i(\Psi') - \sum_{\omega \in \Psi'} p_\omega + \sum_{\omega \in \Omega_{i \rightarrow}} p_\omega \right ) \\
    &\iff \Psi \in \argmax_{\Psi' \in \Omega_i} \left ( \widehat{v}^i(\Psi') - \sum_{\omega \in \Psi'} p_\omega \right )
\end{align*}
Here the second line follows from the first line by applying the variable change $\Phi' = \tau(\Psi')$. The fourth line follows from the third line because $\sum_{\omega \in \Omega_{i \rightarrow}} p_\omega$ is a constant, as $i$ and $\pb$ are fixed. 
\end{proof}

\begin{proposition}
\label{proposition:CE-mapping}
Market arrangement $(\pb, \Phi)$ is a competitive equilibrium iff auction arrangement $(\pb, (\Psi_i)_{i \in I})$ with $\Psi_i = \tau^i(\Phi_i)$ is a competitive equilibrium.
\end{proposition}
\begin{proof}
By construction of $\tau$, we have $\omega \in \Phi \iff \omega \in \Psi_{b(\omega)} \iff \omega \notin \Psi_{s(\omega)}$. And it is immediate that $\omega \notin \Psi_{i}$ if $i \notin \{b(\omega), s(\omega)\}$. So we see that $(\Psi_i)_{i \in I}$ satisfies the first condition of \cref{definition:auction-CE}. \Cref{lemma:demand-mapping} shows us that $(\pb, (\Psi_i)_{i \in I})$ satisfies the second condition of \cref{definition:auction-CE} iff $(\pb, \Phi)$ is a competitive equilibrium.
\end{proof}

We know that the $\widehat{v}^i$ are substitutes valuations. So, the known equilibrium existence for auctions with substitutes buyers also implies equilibrium existence in our markets.

\begin{theorem}
\label{market:CE-existence}
In networked markets with fully substitutable agents there exists at least one competitive equilibrium.
\end{theorem}

Suppose $\Phi$ is a market allocation, and $(\Psi_i)_{i \in I}$ is the corresponding auction allocation with $\Psi_i = \tau^i(\Phi_i)$. Then the social welfares of these allocations in the market and the corresponding auction are identical, as
\[
\sum_{i \in I} v^i(\Phi_i) = \sum_{i \in I} v^i(\tau(\tau(\Phi_i))) = \sum_{i \in I} \widehat{v}^i(\tau(\Phi_i)) = \sum_{i \in I} \widehat{v}^i(\Psi_i).
\]

Hence, because there is a one-to-one correspondence between market arrangements and auction arrangements, we also see that the social welfare theorems hold in our networked setting.

\begin{theorem}
The social welfare theorems hold in the networked markets.
\end{theorem}

\section{Computations in Practice}
\label{app:practical-optimisation}
We now focus on formulating optimization problems to compute leximin, leximax, and minvar core imputations. Fix a market $\market$ with agents $I$, trades $\Omega$, and valuation $v^i$ for each agent $i \in I$. The main purpose of this section is to specify problems that can be solved using the commercial solver Gurobi. This allows us to search for examples computationally.

Recall that a payoff profile $\xb \in \R^I$ lies in the core iff $\sum_{i \in I} x_i = w(I)$ and $\sum_{i \in C} x_i \geq w(C)$ for all coalitions $C \subseteq I$. So we first we compute the welfare function brute-force, and assume that $w$ is known.

\paragraph{Minimum variance core payoffs.}
To find a minvar core imputation, we formulate the convex program (P2).

\begin{mini*}
{\xb \in \R^I_+}{\sum_{i \in I} x_i^2}{}{}\tag{MINVAR}\label{opt:MINVAR}
\addConstraint{\sum_{i \in C} x_i}{\geq w(C),}{\quad \forall C \subseteq I}
\addConstraint{\sum_{i \in I} x_i}{= w(I).}
\end{mini*}

We formulate the dual of \ref{opt:MINVAR}. The Lagrangian is given by
\begin{align*}
    L(x,\lambda_C,\lambda_I) = \sum_{i \in I} x_i^2 + \sum_{C\subseteq I} \lambda_C \left( w(C) - \sum_{i\in C}x_i \right) + \lambda_I \left( \sum_{i\in I}x_i - w(I) \right)
\end{align*}
Let $g(\lambda_C,\lambda_I) = \inf_{x\in \R^I_+} L(x,\lambda_C,\lambda_I)$. Then the dual is
\begin{maxi*}
{\lambda_C, \lambda_I}{g(\lambda_C,\lambda_I)}{}{}\tag{$D-MINVAR$}\label{opt:D-MINVAR}
\addConstraint{\lambda_C,\lambda_I\geq 0}
\end{maxi*}

\paragraph{Leximin and leximax core imputations.}
We compute the leximin core imputation by iteratively maximizing the minimum utility among unassigned players. Assume $I = \{1, \ldots, n\}$ for convenience. Recall the core
\[
C = \left\{ \xb \in \R^n_+ \;\middle|\; \sum_{i \in C} x_i \geq w(C) \;\; \forall C \subseteq [n], \quad \sum_{i \in [n]} x_i = w([n]) \right\}.
\]

\medskip\noindent\textbf{Leximin.}
Initialize $F_0 = [n]$ (the set of \emph{free} players) and $C_0 = C$. At step $k = 1, 2, \ldots\,$:
\begin{maxi!}
{\xb,\, t}{t}
{\label{opt:leximin-k}}
{}\tag{$P3_k$}
\addConstraint{\xb}{\in C_{k-1}}
\addConstraint{x_i}{\geq t,}{\quad \forall i \in F_{k-1}.}
\end{maxi!}
Let $t_k^*$ be the optimal value. Define
\[
A_k = \bigl\{\, i \in F_{k-1} : x_i = t_k^* \text{ for every } \xb \in C_{k-1} \text{ with } x_j \geq t_k^* \text{ for all } j \in F_{k-1} \,\bigr\},
\]
and update $F_k = F_{k-1} \setminus A_k$ and $C_k = C_{k-1} \cap \{\xb : x_i = t_k^* \text{ for } i \in A_k\}$. The procedure terminates when $F_k = \emptyset$.

\begin{lemma}\label{lemma:leximin-progress}
At each step, $A_k \neq \emptyset$, so the procedure terminates in at most $n$ steps. Moreover, the resulting imputation is the leximin core imputation.
\end{lemma}
\begin{proof}
If $A_k = \emptyset$, then every free player can be raised strictly above $t_k^*$ while all other free players remain at or above $t_k^*$. Taking a convex combination of such solutions yields a feasible point with $x_i > t_k^*$ for all $i \in F_{k-1}$, contradicting the optimality of $t_k^*$.

For correctness, let $\yb^*$ be the imputation produced by the procedure, sorted in ascending order: its components are $t_1^* \leq t_2^* \leq \cdots$ with appropriate multiplicities. Suppose some core imputation $\xb'$ has sorted vector $\zb$ with $\yb^* \prec \zb$. Let $k$ be the first step at which they differ, and let $i \in A_k$. Then $x'_i \geq z_k > t_k^*$ (since $\zb$ agrees with $\yb^*$ on all earlier components and is strictly larger at position $k$), but $\xb'$ restricted to the constraints of $C_{k-1}$ with $x_j \geq t_k^*$ for $j \in F_{k-1}$ would witness that player $i$ can exceed $t_k^*$, contradicting $i \in A_k$.
\end{proof}

We remark that the set $A_k$ can be identified with $|F_{k-1}|$ additional linear programs: player $i \in F_{k-1}$ belongs to $A_k$ if and only if the maximum of $x_i$ over $\{\xb \in C_{k-1} : x_j \geq t_k^* \;\forall j \in F_{k-1}\}$ equals $t_k^*$.

\medskip\noindent\textbf{Leximax.}
The leximax core imputation is computed symmetrically, replacing the maximum over a lower bound with a minimum over an upper bound. Initialize $G_0 = [n]$ and $C_0 = C$. At step $\ell = 1, 2, \ldots\,$:
\begin{mini!}
{\xb,\, t}{t}
{\label{opt:leximax-ell}}
{}\tag{$P4_\ell$}
\addConstraint{\xb}{\in C_{\ell-1}}
\addConstraint{x_i}{\leq t,}{\quad \forall i \in G_{\ell-1}.}
\end{mini!}
Let $t_\ell^*$ be optimal. Define $B_\ell = \{i \in G_{\ell-1} : x_i = t_\ell^*$ for every feasible $\xb$ with $x_j \leq t_\ell^*$ for all $j \in G_{\ell-1}\}$, update $G_\ell = G_{\ell-1} \setminus B_\ell$ and $C_\ell = C_{\ell-1} \cap \{\xb : x_i = t_\ell^* \text{ for } i \in B_\ell\}$. The result is the leximax core imputation by a symmetric argument.

\end{document}